\pdfoutput=1
\documentclass[a4paper,UKenglish,cleveref,autoref, thm-restate]{lipics-v2021}

\hideLIPIcs
\nolinenumbers

\title{Local Stochastic Algorithms for Alignment in \\ Self-Organizing Particle Systems} 

\author{Hridesh Kedia}{Georgia Institute of Technology, Atlanta, USA}{hrideshkedia@gmail.com}{}{Supported by  ARO MURI award W911NF-19-1-0233.}

\author{Shunhao Oh}{Georgia Institute of Technology, Atlanta, USA}{ohoh@gatech.edu}{}{Supported by NSF award CCF-1733812 and ARO MURI award W911NF-19-1-0233.}

\author{Dana Randall}{Georgia Institute of Technology, Atlanta, USA}{randall@cc.gatech.edu}{https://orcid.org/0000-0002-1152-2627}{Supported by NSF awards CCF-1733812 and CCF-2106687 and ARO MURI award W911NF-19-1-0233.}

\authorrunning{H. Kedia, S. Oh and D. Randall}

\Copyright{Hridesh Kedia, Shunhao Oh and Dana Randall}


\EventEditors{}
\EventNoEds{0}
\EventLongTitle{}
\EventShortTitle{}
\EventAcronym{}
\EventYear{}
\EventDate{}
\EventLocation{}
\EventLogo{}
\SeriesVolume{}
\ArticleNo{}

\usepackage{amsmath,amssymb,amsthm}  
\usepackage{amsfonts}
\usepackage{wasysym}
\usepackage{bbm}

\usepackage[noend]{algpseudocode}
\usepackage{algorithm}
\usepackage{hyperref}   
\usepackage{color} 

\usepackage{mathtools}

\usepackage{tikz-cd}
\usepackage{subcaption}
\DeclarePairedDelimiter{\abs}{\lvert}{\rvert}

\DeclarePairedDelimiter{\parens}{(}{)}

\DeclarePairedDelimiter{\floor}{\lfloor}{\rfloor}
\DeclarePairedDelimiter{\ceil}{\lceil}{\rceil}

\newtheorem*{theorem*}{Theorem}
\newtheorem*{lemma*}{Lemma}

\newcommand{\Mod}[1]{\ (\mathrm{mod}\ #1)}

\newcommand{\markred}[1]{{#1}}

\begin{document}

\maketitle

\begin{abstract}
We present local distributed, stochastic algorithms for \emph{alignment} in self-organizing particle systems (SOPS) on two-dimensional lattices, where particles occupy unique sites on the  lattice, and particles can make spatial moves to neighboring sites if they are unoccupied. Such models are abstractions of programmable matter, composed of individual computational particles with limited memory, strictly local communication abilities, and modest computational capabilities. We consider oriented particle systems, where particles are assigned a vector pointing in one of $q$ directions, and each particle can compute the angle between its direction and the direction of any neighboring particle, although without knowledge of global orientation with respect to a fixed underlying coordinate system. Particles move stochastically, with each particle able to either modify its direction or make a local spatial move along a lattice edge during a move. We consider two settings: (a) where particle configurations must remain simply connected at all times and (b) where spatial moves are unconstrained and configurations can disconnect.

Our algorithms are inspired by the \emph{Potts model} and its planar oriented variant known as the \emph{planar Potts model} or \emph{clock model} from statistical physics.
We prove that for any $q \geq 2,$ by adjusting a single parameter, these self-organizing particle systems can be made to collectively align along a single dominant direction (analogous to a solid or ordered state) or remain non-aligned, in which case the fraction of particles oriented along any direction is nearly equal (analogous to a gaseous or disordered state). In the connected SOPS setting, we allow for two distinct parameters, one controlling the ferromagnetic attraction between neighboring particles (regardless of orientation) and the other controlling the preference of neighboring particles to align. We show that with appropriate settings of the input parameters, we can achieve \emph{compression} and \emph{expansion}, controlling how tightly gathered the particles are, as well as \emph{alignment} or \emph{nonalignment}, producing a single dominant orientation or not. 
While alignment is known for the Potts and clock models at sufficiently low temperatures,  our proof in the SOPS setting are significantly more challenging because the particles make spatial moves, not all sites are occupied, and the total number of particles is fixed. 
\end{abstract}

\newcommand{\OmegaAggr}[1]{\widetilde{\Omega}^{#1}}
\newcommand{\OmegaAggrRho}{\widetilde{\Omega}^{\rho N}}
\newcommand{\TorusLattice}{{G}_{\Delta}}
\newcommand{\TorusDual}{{G}_{\varhexagon}}
\newcommand{\wFinite}{\tilde{w}_{\mathrm{Potts}}}
\newcommand{\piFinite}{\tilde{\pi}_{\mathrm{Potts}}}
\newcommand{\ZFinite}{\tilde{Z}_{\mathrm{Potts}}}
\newcommand{\Ewrap}{E_{\mathrm{wrap}}}
\usetikzlibrary{patterns}

\section{Introduction}
Autonomous, locally interacting agents can collectively organize to accomplish a variety of complex tasks such as foraging for food, building large-scale structures, and transporting objects many times heavier than their weight, as is routinely observed in the living world, in swarms of ants, flocks of birds, and schools of fish \cite{reynolds_flocks_1987, villa_swarming_2004, vicsek_collective_2012,qin_food_2019}. A key component of these diverse self-organized behaviors is achieving consensus in large collectives of autonomous agents with only local interactions. The problem of achieving alignment in collectives of directed agents is an important example of such a consensus problem, and is a fundamental aspect of \emph{flocking}: large scale collective motion in swarms of motile agents \cite{reynolds_flocks_1987, vicsek_novel_1995, olfati-saber_flocking_2006, zhang_general_2011, vicsek_collective_2012, beaver_overview_2021}. While flocking has been studied extensively \cite{ jadbabaie_coordination_2003, olfati-saber_flocking_2006, tanner_flocking_2007,beaver_overview_2021} with few rigorous results, the more basic problem of alignment has received considerably less attention.  

Here, we study {alignment} in  self-organizing particle systems (SOPS)---a collection of simple, active computational particles that individually execute local distributed algorithms. We consider {\it oriented particle systems} on a two-dimensional lattice, where particles are oriented in one of $q$ directions (with no global compass), for $q \geq 2$, and at most one particle occupies each lattice site. Particles perform moves independently and concurrently by making spatial moves to neighboring empty sites or reorient themselves in new directions with the goal of reaching nearly global alignment.

We consider a stochastic approach, used previously  in \cite{cannon_markov_2016-1, cannon_markov_2016} to achieve {\it compression}, where connected sets of homogeneous  particles self-organize to gather together tightly,  {\it separation} in heterogeneous particle systems, where all of the particles compress, but also gather most tightly with other particles of the same type \cite{cannon_local_2019-1, cannon_local_2019}, and {\it aggregation} of homogeneous particles that are not required to be connected, where most particles accumulate in a small, compact neighborhood \cite{li_programming_2021}. In all of these, phase changes were used to characterize desirable behaviors at stationarity, with high probability.  Following a similar approach, we begin by defining an energy function that assigns the highest weight (or lowest energy) to preferable configurations, and design a Markov chain whose long term behavior favors these low energy configurations using transition probabilities given by the Metropolis-Hastings algorithm~\cite{metropolis_equation_1953,hastings_monte_1970}. We ensure that the transition probabilities of the Markov chain can be computed locally and asynchronously, allowing them to be easily translated to a fully local, distributed algorithm that each particle can run independently. The collective behavior of this distributed algorithm is thus described by the long term behavior of the Markov chain.

\subsection{Related work}
The alignment problems we study can be viewed as finite, unsaturated variants of the ferromagnetic {\it Potts model} from statistical physics \cite{wu_potts_1982}, and a related model known as the {\it clock} or {\it planar Potts model}~\cite{wu_potts_1982,ortiz_dualities_2012}.  In the Potts model, vertices of a graph $G$ are assigned one of $q$ possible ``spins,'' represented here as orientations, and neighboring sites prefer to agree. Let $J>0$ be a parameter related to inverse temperature and let $\delta(X,Y) = 1$ if $X=Y$ and 0 otherwise.  Then the probability of a standard Potts configuration $\sigma$ is given as 
\begin{align*}
\pi(\sigma)=exp \Big(-J \sum_{x\sim y} \delta(\sigma(x),\sigma(y)) \Big)/Z,
\end{align*}
 where the sum is taken over all nearest neighbors in $G$ and $Z$ is the normalizing constant or partition function.  
In the unsaturated setting studied here, spins are identified with particles, not sites, and particles can make spatial moves to unoccupied sites in addition to updating their spins. 
We  present alignment algorithms for two natural variants: (a) the connected setting, where particles are constrained to be simply connected in the lattice, and (b) the general setting,  where particles occupy any distinct lattice sites regardless of connectivity.

Recent work on a closely related {\it site-diluted} Potts model \cite{wu_potts_1982,chayes_aggregation_1995}   also allows a non-zero fraction of lattice sites to be unoccupied, but the number of particles is not fixed, so particles can appear and disappear, in addition to making spatial moves. Chayes et al.~\cite{chayes_aggregation_1995} beautifully demonstrate the presence of ordered (aligned and occupied) and disordered (non-aligned and vacant) phases, along with novel ``staggered'' phases in this model. However, our constraint fixing the number of particles, which is necessary in SOPS models in programmable matter,  makes our system fundamentally different from the site-diluted Potts model akin to the difference between the fixed magnetization Ising model, which has a fixed number of + spins, and the Ising model in the presence of a magnetic field, where the number of + spins can vary. Notably, the coexistence of phases that characterize the aligned and compressed behaviors we are seeking will not occur unless we fix the magnetization (or numbers of particles) as these configurations are exponentially unlikely in the site-diluted model and thus do not inherit any of its properties.
  

Since  particles  can make spatial moves, the  boundary between the particle occupied sites and the unoccupied sites can assume arbitrary shapes, which makes achieving alignment more challenging than achieving compression. Consider the configurations shown in Figure~1(a),(b), where the particles can be oriented along one of two possible directions ($q=2$) shown by black  and grey circles,  with a total of $n$ particles.
While the number of unaligned pairs of adjacent particles is $O(\sqrt{n})$ for the configuration in Figure 1(a), it can be as low as $O(1)$ for the configuration shown in Figure 1(b), owing to the bottleneck shaped part of the configuration boundary, making it likely that the regions on either side of it will be aligned along different directions. Hence, achieving alignment requires suppressing the likelihood of such bottlenecks in the boundary of the particle configuration.

While the concept of an interfacial free energy can be used to constrain the shape of the boundary of a dilute system of homogeneous particles i.e., when $q=1$, as in \cite{minlos_phenomenon_1967, ioffe_dobrushinkoteckyshlosman_1998,bodineau_rigorous_2000,pfister_interface_2009}, because particle occupied sites and vacant sites are akin to distinct coexisting phases of the system. However, the same ideas do not readily generalize to the case when $q\geq 2$. Instead, we show build on the notion of compression introduced in \cite{cannon_markov_2016-1, cannon_markov_2016}, and use isoperimetric inequalities to show that for sufficiently compressed configurations, bottlenecks such as the one shown in Figure 1(b) are precluded with high probability. 


\subsection{Results}
We present the first rigorous local distributed algorithms for achieving both low perimeter boundaries and alignment, for any number of orientations $q\geq 2$, in both connected and general settings. Informally, we say a particle system is {\it aligned} if a significant percentage of the particles have the same orientation.

In the connected SOPS setting, we define an energy function that encourages compression of the entire configuration and also defines a ferromagnetic interaction between particles' orientations, inspired by the clock and Potts models. These two contributions are controlled by two independent parameters $\lambda$ and $\gamma$.  In this setting, we show that given any $\alpha > 1$, for any $\lambda > 1$ and $\gamma > 29.3(q-1)$ such that $\lambda\gamma > 7^{\alpha/(\alpha-1)}$, the algorithms achieve $\alpha$-compression with high probability. Furthermore, when $\gamma$ satisfies additional constraints given in Theorem~\ref{thm:con_align_thm}, we show that the compressed configurations are very likely to be aligned. Next, we show that setting $\lambda$ large and  $\gamma$ small will generate compressed configurations with an equitable balance of orientations (Theorem~\ref{thm:conn_non_alignment}), while setting $\lambda$ small will generate configurations that are {\it expanded}, nearly maximizing their perimeters, allowing the SOPS to explore space, potentially to forage for resources, for example (see Theorem~\ref{thm:conn_expansion}).


 For both the Potts and clock models in the connected setting, the proofs rely on the {\it cluster expansion} \cite{mayer_statistical_1937,friedli_cluster_2017,kotecky_cluster_1986} from statistical physics, 
 introducing a new so-called {\it polymer model} inspired by the relationship between flows and the Potts model \cite{essam_potts_1986}. Informally, the cluster expansion allows us to obtain upper and lower bounds on the so-called ``polymer partition function'' in terms of the volume and surface contributions, as in \cite{cannon_local_2019-1, cannon_local_2019, friedli_cluster_2017}, to prove that our algorithms achieve compression (or aggregation), with high probability. Moreover, using isoperimetric inequalities, we prove the absence of bottlenecks in sufficiently highly compressed configurations, which is necessary to get the system to globally align. Finally, we use the {\it bridging techniques} first proposed in \cite{miracle_clustering_2011} and later adapted in \cite{cannon_local_2019-1, cannon_local_2019}, to expand the information theoretic arguments in \cite{cannon_local_2019-1, cannon_local_2019} to prove that for sufficiently compressed configurations, our algorithms achieve alignment with high probability. Conversely, we show that our algorithms can achieve {\it expansion} and/or {\it non-alignment} (with all directions nearly equitably balanced), with the same algorithm by adjusting only two global parameters.

 In the general SOPS setting, with no connectivity constraints, we present an algorithm based on a single parameter coupling both compression and ferromagnetism simultaneously.  When this parameter is sufficiently large, we achieve aggregation and alignment (Theorem~\ref{theorem:aggregationwithalignment}), and when it is small we achieve dispersion and a balance among the orientations (Theorems~\ref{thm:gen_non_alignment}~and~\ref{thm:gen_dispersion}). We believe these parameters can be independently controlled in the general (disconnected) setting, but the proofs seemingly become significantly more challenging and coupling them into one parameter seems sufficient for most applications in programmable matter and swarm robotics.
 Because configurations tend to be highly disconnected,
proofs in the  general setting require additional technology to account for many small clusters that can be distributed throughout the lattice. Here we generalize the bridging techniques to account for more complex contours that form an interconnected network to show that the contour lengths of the bridging system can be made arbitrarily close to their minimum possible length and, as a result, alignment occurs with high probability.
We note that our algorithms for alignment in both settings work for all $q\geq2$; separation (where the sizes of the color classes are fixed) has only been shown for $q=2$, although the methods should also generalize to more colors \cite{cannon_local_2019}.


\begin{figure}
\begin{subfigure}[b]{.45\textwidth}
\begin{center}\begin{tikzpicture}[x=0.4cm,y=0.4cm]
\draw[lightgray] (0,-8) -- (1.73205,-9);
\draw[lightgray] (0,-1) -- (1.73205,0);
\draw[lightgray] (0,-7) -- (3.4641,-9);
\draw[lightgray] (0,-2) -- (3.4641,0);
\draw[lightgray] (0,-6) -- (5.19615,-9);
\draw[lightgray] (0,-3) -- (5.19615,0);
\draw[lightgray] (0,-5) -- (6.9282,-9);
\draw[lightgray] (0,-4) -- (6.9282,0);
\draw[lightgray] (0,-4) -- (8.66025,-9);
\draw[lightgray] (0,-5) -- (8.66025,0);
\draw[lightgray] (0,-3) -- (10.3923,-9);
\draw[lightgray] (0,-6) -- (10.3923,0);
\draw[lightgray] (0,-2) -- (12.1244,-9);
\draw[lightgray] (0,-7) -- (12.1244,0);
\draw[lightgray] (0,-1) -- (13.8564,-9);
\draw[lightgray] (0,-8) -- (13.8564,0);
\draw[lightgray] (0,0) -- (15.5885,-9);
\draw[lightgray] (0,-9) -- (15.5885,0);
\draw[lightgray] (1.73205,0) -- (15.5885,-8);
\draw[lightgray] (1.73205,-9) -- (15.5885,-1);
\draw[lightgray] (3.4641,0) -- (15.5885,-7);
\draw[lightgray] (3.4641,-9) -- (15.5885,-2);
\draw[lightgray] (5.19615,0) -- (15.5885,-6);
\draw[lightgray] (5.19615,-9) -- (15.5885,-3);
\draw[lightgray] (6.9282,0) -- (15.5885,-5);
\draw[lightgray] (6.9282,-9) -- (15.5885,-4);
\draw[lightgray] (8.66025,0) -- (15.5885,-4);
\draw[lightgray] (8.66025,-9) -- (15.5885,-5);
\draw[lightgray] (10.3923,0) -- (15.5885,-3);
\draw[lightgray] (10.3923,-9) -- (15.5885,-6);
\draw[lightgray] (12.1244,0) -- (15.5885,-2);
\draw[lightgray] (12.1244,-9) -- (15.5885,-7);
\draw[lightgray] (13.8564,0) -- (15.5885,-1);
\draw[lightgray] (13.8564,-9) -- (15.5885,-8);
\draw[lightgray] (15.5885,0) -- (15.5885,0);
\draw[lightgray] (15.5885,-9) -- (15.5885,-9);
\draw[lightgray] (0,0) -- (0,-9);
\draw[lightgray] (0.866025,-0.5) -- (0.866025,-8.5);
\draw[lightgray] (1.73205,0) -- (1.73205,-9);
\draw[lightgray] (2.59808,-0.5) -- (2.59808,-8.5);
\draw[lightgray] (3.4641,0) -- (3.4641,-9);
\draw[lightgray] (4.33013,-0.5) -- (4.33013,-8.5);
\draw[lightgray] (5.19615,0) -- (5.19615,-9);
\draw[lightgray] (6.06218,-0.5) -- (6.06218,-8.5);
\draw[lightgray] (6.9282,0) -- (6.9282,-9);
\draw[lightgray] (7.79423,-0.5) -- (7.79423,-8.5);
\draw[lightgray] (8.66025,0) -- (8.66025,-9);
\draw[lightgray] (9.52628,-0.5) -- (9.52628,-8.5);
\draw[lightgray] (10.3923,0) -- (10.3923,-9);
\draw[lightgray] (11.2583,-0.5) -- (11.2583,-8.5);
\draw[lightgray] (12.1244,0) -- (12.1244,-9);
\draw[lightgray] (12.9904,-0.5) -- (12.9904,-8.5);
\draw[lightgray] (13.8564,0) -- (13.8564,-9);
\draw[lightgray] (14.7224,-0.5) -- (14.7224,-8.5);
\draw[lightgray] (15.5885,0) -- (15.5885,-9);
\draw[fill=black] (1.73205,-2) circle (0.12);
\draw[fill=black] (1.73205,-3) circle (0.12);
\draw[fill=black] (1.73205,-4) circle (0.12);
\draw[fill=black] (1.73205,-5) circle (0.12);
\draw[fill=black] (2.59808,-1.5) circle (0.12);
\draw[fill=black] (2.59808,-2.5) circle (0.12);
\draw[fill=black] (2.59808,-3.5) circle (0.12);
\draw[fill=black] (2.59808,-4.5) circle (0.12);
\draw[fill=black] (2.59808,-5.5) circle (0.12);
\draw[fill=black] (3.4641,-1) circle (0.12);
\draw[fill=black] (3.4641,-2) circle (0.12);
\draw[fill=black] (3.4641,-3) circle (0.12);
\draw[fill=black] (3.4641,-4) circle (0.12);
\draw[fill=black] (3.4641,-5) circle (0.12);
\draw[fill=black] (3.4641,-6) circle (0.12);
\draw[fill=black] (4.33013,-1.5) circle (0.12);
\draw[fill=black] (4.33013,-2.5) circle (0.12);
\draw[fill=black] (4.33013,-3.5) circle (0.12);
\draw[fill=black] (4.33013,-4.5) circle (0.12);
\draw[fill=black] (4.33013,-5.5) circle (0.12);
\draw[black, line width=0.3mm] (4.33013,-6.5) circle (0.12);
\draw[fill=black] (5.19615,-1) circle (0.12);
\draw[fill=black] (5.19615,-2) circle (0.12);
\draw[fill=black] (5.19615,-3) circle (0.12);
\draw[fill=black] (5.19615,-4) circle (0.12);
\draw[fill=black] (5.19615,-5) circle (0.12);
\draw[black, line width=0.3mm] (5.19615,-6) circle (0.12);
\draw[black, line width=0.3mm] (5.19615,-7) circle (0.12);
\draw[fill=black] (6.06218,-1.5) circle (0.12);
\draw[fill=black] (6.06218,-2.5) circle (0.12);
\draw[fill=black] (6.06218,-3.5) circle (0.12);
\draw[fill=black] (6.06218,-4.5) circle (0.12);
\draw[black, line width=0.3mm] (6.06218,-5.5) circle (0.12);
\draw[black, line width=0.3mm] (6.06218,-6.5) circle (0.12);
\draw[black, line width=0.3mm] (6.06218,-7.5) circle (0.12);
\draw[fill=black] (6.9282,-1) circle (0.12);
\draw[fill=black] (6.9282,-2) circle (0.12);
\draw[fill=black] (6.9282,-3) circle (0.12);
\draw[fill=black] (6.9282,-4) circle (0.12);
\draw[black, line width=0.3mm] (6.9282,-5) circle (0.12);
\draw[black, line width=0.3mm] (6.9282,-6) circle (0.12);
\draw[black, line width=0.3mm] (6.9282,-7) circle (0.12);
\draw[fill=black] (7.79423,-1.5) circle (0.12);
\draw[fill=black] (7.79423,-2.5) circle (0.12);
\draw[fill=black] (7.79423,-3.5) circle (0.12);
\draw[black, line width=0.3mm] (7.79423,-4.5) circle (0.12);
\draw[black, line width=0.3mm] (7.79423,-5.5) circle (0.12);
\draw[black, line width=0.3mm] (7.79423,-6.5) circle (0.12);
\draw[black, line width=0.3mm] (7.79423,-7.5) circle (0.12);
\draw[fill=black] (8.66025,-1) circle (0.12);
\draw[fill=black] (8.66025,-2) circle (0.12);
\draw[fill=black] (8.66025,-3) circle (0.12);
\draw[black, line width=0.3mm] (8.66025,-4) circle (0.12);
\draw[black, line width=0.3mm] (8.66025,-5) circle (0.12);
\draw[black, line width=0.3mm] (8.66025,-6) circle (0.12);
\draw[black, line width=0.3mm] (8.66025,-7) circle (0.12);
\draw[fill=black] (9.52628,-1.5) circle (0.12);
\draw[fill=black] (9.52628,-2.5) circle (0.12);
\draw[black, line width=0.3mm] (9.52628,-3.5) circle (0.12);
\draw[black, line width=0.3mm] (9.52628,-4.5) circle (0.12);
\draw[black, line width=0.3mm] (9.52628,-5.5) circle (0.12);
\draw[black, line width=0.3mm] (9.52628,-6.5) circle (0.12);
\draw[black, line width=0.3mm] (9.52628,-7.5) circle (0.12);
\draw[fill=black] (10.3923,-2) circle (0.12);
\draw[black, line width=0.3mm] (10.3923,-3) circle (0.12);
\draw[black, line width=0.3mm] (10.3923,-4) circle (0.12);
\draw[black, line width=0.3mm] (10.3923,-5) circle (0.12);
\draw[black, line width=0.3mm] (10.3923,-6) circle (0.12);
\draw[black, line width=0.3mm] (10.3923,-7) circle (0.12);
\draw[black, line width=0.3mm] (11.2583,-2.5) circle (0.12);
\draw[black, line width=0.3mm] (11.2583,-3.5) circle (0.12);
\draw[black, line width=0.3mm] (11.2583,-4.5) circle (0.12);
\draw[black, line width=0.3mm] (11.2583,-5.5) circle (0.12);
\draw[black, line width=0.3mm] (11.2583,-6.5) circle (0.12);
\draw[black, line width=0.3mm] (11.2583,-7.5) circle (0.12);
\draw[black, line width=0.3mm] (12.1244,-3) circle (0.12);
\draw[black, line width=0.3mm] (12.1244,-4) circle (0.12);
\draw[black, line width=0.3mm] (12.1244,-5) circle (0.12);
\draw[black, line width=0.3mm] (12.1244,-6) circle (0.12);
\draw[black, line width=0.3mm] (12.1244,-7) circle (0.12);
\draw[black, line width=0.3mm] (12.9904,-3.5) circle (0.12);
\draw[black, line width=0.3mm] (12.9904,-4.5) circle (0.12);
\draw[black, line width=0.3mm] (12.9904,-5.5) circle (0.12);
\draw[black, line width=0.3mm] (12.9904,-6.5) circle (0.12);
\end{tikzpicture}
\end{center}
\label{fig:no_bottleneck}
\caption{A configuration without bottlenecks.}
\end{subfigure}\begin{subfigure}[b]{.5\textwidth}
\begin{center}\begin{tikzpicture}[x=0.4cm,y=0.4cm]
\draw[lightgray] (0,-8) -- (1.73205,-9);
\draw[lightgray] (0,-1) -- (1.73205,0);
\draw[lightgray] (0,-7) -- (3.4641,-9);
\draw[lightgray] (0,-2) -- (3.4641,0);
\draw[lightgray] (0,-6) -- (5.19615,-9);
\draw[lightgray] (0,-3) -- (5.19615,0);
\draw[lightgray] (0,-5) -- (6.9282,-9);
\draw[lightgray] (0,-4) -- (6.9282,0);
\draw[lightgray] (0,-4) -- (8.66025,-9);
\draw[lightgray] (0,-5) -- (8.66025,0);
\draw[lightgray] (0,-3) -- (10.3923,-9);
\draw[lightgray] (0,-6) -- (10.3923,0);
\draw[lightgray] (0,-2) -- (12.1244,-9);
\draw[lightgray] (0,-7) -- (12.1244,0);
\draw[lightgray] (0,-1) -- (13.8564,-9);
\draw[lightgray] (0,-8) -- (13.8564,0);
\draw[lightgray] (0,0) -- (15.5885,-9);
\draw[lightgray] (0,-9) -- (15.5885,0);
\draw[lightgray] (1.73205,0) -- (15.5885,-8);
\draw[lightgray] (1.73205,-9) -- (15.5885,-1);
\draw[lightgray] (3.4641,0) -- (15.5885,-7);
\draw[lightgray] (3.4641,-9) -- (15.5885,-2);
\draw[lightgray] (5.19615,0) -- (15.5885,-6);
\draw[lightgray] (5.19615,-9) -- (15.5885,-3);
\draw[lightgray] (6.9282,0) -- (15.5885,-5);
\draw[lightgray] (6.9282,-9) -- (15.5885,-4);
\draw[lightgray] (8.66025,0) -- (15.5885,-4);
\draw[lightgray] (8.66025,-9) -- (15.5885,-5);
\draw[lightgray] (10.3923,0) -- (15.5885,-3);
\draw[lightgray] (10.3923,-9) -- (15.5885,-6);
\draw[lightgray] (12.1244,0) -- (15.5885,-2);
\draw[lightgray] (12.1244,-9) -- (15.5885,-7);
\draw[lightgray] (13.8564,0) -- (15.5885,-1);
\draw[lightgray] (13.8564,-9) -- (15.5885,-8);
\draw[lightgray] (15.5885,0) -- (15.5885,0);
\draw[lightgray] (15.5885,-9) -- (15.5885,-9);
\draw[lightgray] (0,0) -- (0,-9);
\draw[lightgray] (0.866025,-0.5) -- (0.866025,-8.5);
\draw[lightgray] (1.73205,0) -- (1.73205,-9);
\draw[lightgray] (2.59808,-0.5) -- (2.59808,-8.5);
\draw[lightgray] (3.4641,0) -- (3.4641,-9);
\draw[lightgray] (4.33013,-0.5) -- (4.33013,-8.5);
\draw[lightgray] (5.19615,0) -- (5.19615,-9);
\draw[lightgray] (6.06218,-0.5) -- (6.06218,-8.5);
\draw[lightgray] (6.9282,0) -- (6.9282,-9);
\draw[lightgray] (7.79423,-0.5) -- (7.79423,-8.5);
\draw[lightgray] (8.66025,0) -- (8.66025,-9);
\draw[lightgray] (9.52628,-0.5) -- (9.52628,-8.5);
\draw[lightgray] (10.3923,0) -- (10.3923,-9);
\draw[lightgray] (11.2583,-0.5) -- (11.2583,-8.5);
\draw[lightgray] (12.1244,0) -- (12.1244,-9);
\draw[lightgray] (12.9904,-0.5) -- (12.9904,-8.5);
\draw[lightgray] (13.8564,0) -- (13.8564,-9);
\draw[lightgray] (14.7224,-0.5) -- (14.7224,-8.5);
\draw[lightgray] (15.5885,0) -- (15.5885,-9);
\draw[fill=black] (0.866025,-1.5) circle (0.12);
\draw[fill=black] (0.866025,-2.5) circle (0.12);
\draw[fill=black] (0.866025,-3.5) circle (0.12);
\draw[fill=black] (0.866025,-4.5) circle (0.12);
\draw[fill=black] (1.73205,-1) circle (0.12);
\draw[fill=black] (1.73205,-2) circle (0.12);
\draw[fill=black] (1.73205,-3) circle (0.12);
\draw[fill=black] (1.73205,-4) circle (0.12);
\draw[fill=black] (1.73205,-5) circle (0.12);
\draw[fill=black] (2.59808,-0.5) circle (0.12);
\draw[fill=black] (2.59808,-1.5) circle (0.12);
\draw[fill=black] (2.59808,-2.5) circle (0.12);
\draw[fill=black] (2.59808,-3.5) circle (0.12);
\draw[fill=black] (2.59808,-4.5) circle (0.12);
\draw[fill=black] (3.4641,-1) circle (0.12);
\draw[fill=black] (3.4641,-2) circle (0.12);
\draw[fill=black] (3.4641,-3) circle (0.12);
\draw[fill=black] (3.4641,-4) circle (0.12);
\draw[fill=black] (3.4641,-5) circle (0.12);
\draw[fill=black] (4.33013,-0.5) circle (0.12);
\draw[fill=black] (4.33013,-1.5) circle (0.12);
\draw[fill=black] (4.33013,-2.5) circle (0.12);
\draw[fill=black] (4.33013,-3.5) circle (0.12);
\draw[fill=black] (4.33013,-4.5) circle (0.12);
\draw[fill=black] (5.19615,-1) circle (0.12);
\draw[fill=black] (5.19615,-2) circle (0.12);
\draw[fill=black] (5.19615,-3) circle (0.12);
\draw[fill=black] (5.19615,-4) circle (0.12);
\draw[fill=black] (6.06218,-0.5) circle (0.12);
\draw[fill=black] (6.06218,-1.5) circle (0.12);
\draw[fill=black] (6.06218,-2.5) circle (0.12);
\draw[fill=black] (6.06218,-3.5) circle (0.12);
\draw[black, line width=0.3mm] (6.06218,-6.5) circle (0.12);
\draw[black, line width=0.3mm] (6.06218,-7.5) circle (0.12);
\draw[fill=black] (6.9282,-1) circle (0.12);
\draw[fill=black] (6.9282,-2) circle (0.12);
\draw[fill=black] (6.9282,-3) circle (0.12);
\draw[fill=black] (6.9282,-4) circle (0.12);
\draw[black, line width=0.3mm] (6.9282,-6) circle (0.12);
\draw[black, line width=0.3mm] (6.9282,-7) circle (0.12);
\draw[black, line width=0.3mm] (6.9282,-8) circle (0.12);
\draw[fill=black] (7.79423,-0.5) circle (0.12);
\draw[fill=black] (7.79423,-1.5) circle (0.12);
\draw[fill=black] (7.79423,-2.5) circle (0.12);
\draw[fill=black] (7.79423,-3.5) circle (0.12);
\draw[black, line width=0.3mm] (7.79423,-4.5) circle (0.12);
\draw[black, line width=0.3mm] (7.79423,-5.5) circle (0.12);
\draw[black, line width=0.3mm] (7.79423,-6.5) circle (0.12);
\draw[black, line width=0.3mm] (7.79423,-7.5) circle (0.12);
\draw[fill=black] (8.66025,-1) circle (0.12);
\draw[fill=black] (8.66025,-2) circle (0.12);
\draw[black, line width=0.3mm] (8.66025,-4) circle (0.12);
\draw[black, line width=0.3mm] (8.66025,-5) circle (0.12);
\draw[black, line width=0.3mm] (8.66025,-6) circle (0.12);
\draw[black, line width=0.3mm] (8.66025,-7) circle (0.12);
\draw[black, line width=0.3mm] (8.66025,-8) circle (0.12);
\draw[black, line width=0.3mm] (9.52628,-4.5) circle (0.12);
\draw[black, line width=0.3mm] (9.52628,-5.5) circle (0.12);
\draw[black, line width=0.3mm] (9.52628,-6.5) circle (0.12);
\draw[black, line width=0.3mm] (9.52628,-7.5) circle (0.12);
\draw[black, line width=0.3mm] (10.3923,-4) circle (0.12);
\draw[black, line width=0.3mm] (10.3923,-5) circle (0.12);
\draw[black, line width=0.3mm] (10.3923,-6) circle (0.12);
\draw[black, line width=0.3mm] (10.3923,-7) circle (0.12);
\draw[black, line width=0.3mm] (10.3923,-8) circle (0.12);
\draw[black, line width=0.3mm] (11.2583,-3.5) circle (0.12);
\draw[black, line width=0.3mm] (11.2583,-4.5) circle (0.12);
\draw[black, line width=0.3mm] (11.2583,-5.5) circle (0.12);
\draw[black, line width=0.3mm] (11.2583,-6.5) circle (0.12);
\draw[black, line width=0.3mm] (11.2583,-7.5) circle (0.12);
\draw[black, line width=0.3mm] (12.1244,-4) circle (0.12);
\draw[black, line width=0.3mm] (12.1244,-5) circle (0.12);
\draw[black, line width=0.3mm] (12.1244,-6) circle (0.12);
\draw[black, line width=0.3mm] (12.1244,-7) circle (0.12);
\draw[black, line width=0.3mm] (12.1244,-8) circle (0.12);
\draw[black, line width=0.3mm] (12.9904,-3.5) circle (0.12);
\draw[black, line width=0.3mm] (12.9904,-4.5) circle (0.12);
\draw[black, line width=0.3mm] (12.9904,-5.5) circle (0.12);
\draw[black, line width=0.3mm] (12.9904,-6.5) circle (0.12);
\draw[black, line width=0.3mm] (12.9904,-7.5) circle (0.12);
\draw[black, line width=0.3mm] (13.8564,-4) circle (0.12);
\draw[black, line width=0.3mm] (13.8564,-5) circle (0.12);
\draw[black, line width=0.3mm] (13.8564,-6) circle (0.12);
\draw[black, line width=0.3mm] (13.8564,-7) circle (0.12);
\end{tikzpicture}
\end{center}
\label{fig:bottleneck}
\caption{A configuration with a bottleneck.}
\end{subfigure}\caption{Configurations with two dominant orientations (black  vs. gray circles); large interfaces as in (a) are unlikely for large $\gamma,$ whereas small interfaces as in (b) are likely for any finite $\gamma$.}
\end{figure}
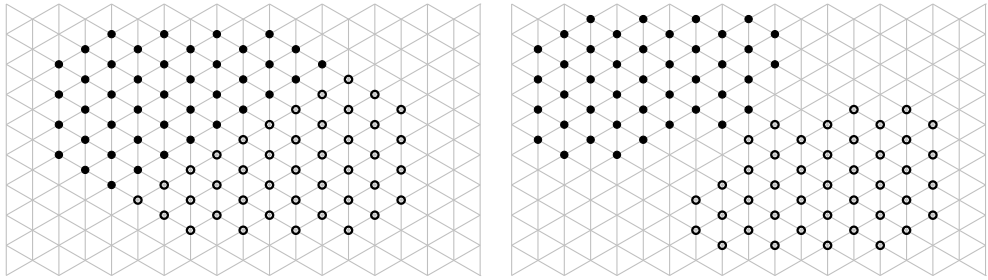

\section{Preliminaries}\label{Sec:prelims}
Our model of programmable matter is based on the \emph{amoebot model}, introduced in~\cite{derakhshandeh_amoebot_2014} and described in detail in~\cite{daymude_computing_2019}, which has served as the basis for previous stochastic algorithms for SOPS \cite{cannon_markov_2016,cannon_markov_2016-1,cannon_local_2019,cannon_local_2019-1}. In the amoebot model, particles occupy the nodes of a graph with each node occupied by at most one particle. When executing a spatial move, a particle expands into an adjacent unoccupied node, temporarily occupying both nodes and then contracts to the new node. 
Each particle stores whether it is expanded or contracted and can read whether its neighbors are expanded or contracted. No particle has access to global information such as system size or a shared co-ordinate system or compass.    

We extend the amoebot model to model heterogeneous particles, where each particle has one of $q$ orientations, akin to the variant introduced in \cite{cannon_local_2019, cannon_local_2019-1}. Each particle, when activated, chooses either a spatial move as in the original amoebot model, or an ``orientation move'' that updates its direction, each equal probability.  The system performs these \emph{atomic actions}, following the $\mathcal{A}$SYNC model of computation from distributed computing \cite{lynch_distributed_1996}. It has been shown in this model that for any concurrent asynchronous execution of atomic actions, there exists a sequential ordering of actions with the same end state provided that all conflicts arising in the concurrent asynchronous execution are resolved. We assume that conflicts due to multiple particles expanding into an unoccupied node are resolved arbitrarily so that only one particle expands into the unoccupied node, allowing us to consider only one particle to be active at any given time. 

\subsection{The Potts and clock models}
In our models, each configuration is an assignment of $n$ particles to distinct vertices of a finite triangular lattice $G_\Delta$ of $N > n$ vertices with the toroidal topology. In addition, each particle is also assigned an orientation from $\{0,1,\ldots,q-1\}$.
We assume $\TorusLattice$ to inhabit a $\sqrt{N} \times \sqrt{N}$ square region with periodic boundary conditions.
Each vertex $(x,y)$ of $\TorusLattice$ has six outgoing edges, to the vertices $(x+1,y), (x,y+1), (x+1,y+1), (x-1,y), (x,y-1), (x-1,y-1)$, where addition and subtraction is taken modulo $\sqrt{N}-1$.
Moreover, in this setup, the set of particles in our configurations must always be connected and hole-free. 
Given such a configuration, we define its \emph{boundary} $\mathcal{P}$ to be the minimal closed walk over occupied sites of $G_\Delta$ that encloses all of the occupied sites in the configuration. The \emph{perimeter} $p(\sigma)$ of a configuration $\sigma$ is then defined to be the length of this closed walk.

We consider the following Potts Hamiltonian, on $G_\Delta$, a variant of the site-diluted Potts model \cite{chayes_aggregation_1995}: 
\[ H_\mathrm{Potts}(\sigma) = -J\sum_{\langle i,j\rangle} n_i n_j\,\delta(\theta_i, \theta_j) -\kappa \sum_{\langle i,j\rangle} n_i n_j \,,\]
where the sum is over all pairs of adjacent sites: $\langle i,j \rangle$ i.e., sites connected by a single lattice edge in 
$G_\Delta$, $n_i \in \{0,1\}$ indicates whether site $i$ is occupied or not, $\theta_i$ indicates the orientation of the particle on site $i$, and $J, \kappa$ are positive constants. We only consider configurations $\sigma$ in~$\Omega$, i.e., where the total number of particles is equal to $n$, and the particle-occupied sites form a connected, hole-free region. 

The probability of a configuration $\pi_\mathrm{Potts}(\sigma)$ is given by the Boltzmann distribution:
\[ \pi_\mathrm{Potts}(\sigma) = {\mathrm{e}^{-\beta H_\mathrm{Potts}(\sigma)}}/{Z_\mathrm{Potts}}, \ \ \ {\rm where} \ \ \ Z_\mathrm{Potts} = \sum_{\sigma' \in\, \Omega} \mathrm{e}^{-\beta H_\mathrm{Potts}(\sigma')} \,,\]
where $\beta$ denotes the inverse temperature. 
Setting parameters $\lambda = \exp(\beta\kappa)$, and $\gamma = \exp(\beta J)$, the above probability distribution can be expressed as: 
\begin{equation}
   \pi_{\mathrm{Potts}}(\sigma) = \frac{w_\mathrm{Potts}(\sigma)}{Z_\mathrm{Potts}}\,,\; w_\mathrm{Potts}(\sigma) = (\lambda\,\gamma)^{-p(\sigma)} \gamma^{-h(\sigma)}\,,\; Z_\mathrm{clock} = \sum_{\sigma' \in\, \Omega} w_\mathrm{Potts}(\sigma'),
\end{equation}
where $h(\sigma)$ is the number of heterogeneous edges in $\sigma$, i.e., edges connecting particles with different orientations, and $p(\sigma)$ is its perimeter, as defined earlier. Here $\pi_\mathrm{Potts}$ is the stationary distribution for our Markov chain algorithm based on the ferromagnetic Potts model interactions.

Similarly, we consider the following clock model Hamiltonian on $G_\Delta$:
\[ H_\mathrm{clock}(\sigma) = -J\sum_{\langle i,j\rangle} n_i n_j\,\cos(2\pi(\theta_i - \theta_j)/q) -\kappa \sum_{\langle i,j\rangle} n_i n_j \,.\]

\noindent The probability of a configuration $\pi_\mathrm{clock}(\sigma)$ is given by the Boltzmann distribution as before, and can be expressed in terms of the parameters $\lambda , \gamma$ as:
\begin{equation}
    \pi_{\mathrm{clock}}(\sigma) = \frac{w_\mathrm{clock}(\sigma)}{Z_\mathrm{clock}}\,,\; w_\mathrm{clock}(\sigma) = (\lambda\,\gamma)^{-p(\sigma)} \prod_{\langle i, j\rangle}\gamma^{-d_{ij}}\,,\; Z_\mathrm{clock} = \sum_{\sigma' \in\, \Omega} w_\mathrm{clock}(\sigma'), \label{clock_distribution}
\end{equation}
where $\lambda > 0, \gamma > 0$ (as before),  $d_{ij} := 1-\cos(2\pi(\theta_i - \theta_j)/q)$, and the product is over all pairs of adjacent occupied sites. Here $\pi_\mathrm{clock}$ will be the stationary distribution for our Markov chain algorithm based on the clock model. 

For each of the above models, we will refer to $w(\sigma)$ ($w_{\mathrm{Potts}}$ or $w_{\mathrm{clock}}$) as the \emph{weight} of a configuration. The stationary probability distribution $\pi$ ($\pi_{\mathrm{Potts}}$ or $\pi_{\mathrm{clock}}$) is thus simply the weight function $w$ normalized by the \emph{partition function} $Z$ ($Z_{\mathrm{Potts}}$ or $Z_{\mathrm{clock}}$).

\subsection{Cluster expansions and bridging}
Our proofs build on several tools from statistical physics and combinatorics, so we begin by introducing two key methods.
The {cluster expansion} is one of the oldest tools in statistical physics \cite{mayer_statistical_1937,mayer_theory_1950, friedli_cluster_2017}, and has led to the development of the Pirogov-Sinai theory~\cite{pirogov_phase_1975, pirogov_phase_1976}, playing an important role in recent advances in efficient sampling and counting algorithms~\cite{helmuth_algorithmic_2019,jenssen_independent_2020,borgs_efficient_2020}. The cluster expansion expresses the logarithm of a polymer partition function as a sum over polymer clusters. 

Let $\mathcal{L}$ be a finite set of polymers $\{\xi_i\}$, where each polymer $\xi_i$ has weight $w(\xi_i)$. We also define  ``compatibility'' between polymers - each pair of polymers $\xi, \xi'$ is either compatible ($\xi\sim\xi'$) or incompatible ($\xi\nsim\xi'$).
The polymer partition function is then given by:
\[ \Xi = \sum_{\tau\in\Omega^\mathcal{L}} \prod_{\xi\in \tau}w(\xi) \,,\]
where $\Omega^\mathcal{L}$ is the set of all collections of pairwise compatible polymers in $\mathcal{L}$. 
The cluster expansion expresses the logarithm of the polymer partition function in terms of clusters, where a cluster $X$ is an ordered multiset of polymers $\{\xi_1,\ldots, \xi_k\}$ such that their incompatibility graph $H(X)$ is connected, where the incompatibility graph is constructed by representing  each polymer by a vertex and connecting two vertices if the corresponding polymers are incompatible. The cluster expansion gives:
\begin{align*}
\log\Xi = \sum_{X\in \mathcal{C}} \Psi(X)\, \text{, where }
\Psi(X) := \frac{1}{\vert X \vert !}\left( \sum_{{G\subseteq H_X } } (-1)^{\vert E(G)\vert}\right)\left( \prod_{\xi\in X} w(\xi)\right)\,,
\end{align*}
where the sum is taken over connected, spanning subgraphs $G$ and $\mathcal{C}$ is the set of all clusters. A sufficient condition for the convergence of the cluster expansion was given by Koteck{\'y} and Preiss \cite{kotecky_cluster_1986}. We will prove this condition in Lemma \ref{lem:kotecky_preiss_potts} and use the cluster expansion to separate the volume and surface contributions to the partition function, as done in \cite{friedli_cluster_2017, cannon_local_2019}.

{Bridging} is a combinatorial technique used to show that large contours are uncommon, while allowing for the possibility of many small contours corresponding to ``defects''. It was first introduced in \cite{miracle_clustering_2011} and later adapted in \cite{cannon_local_2019}. 
We note that a constant fraction of defects will be unavoidable - an example of this is in the Ising model and Potts models, where a constant fraction of the vertices will not follow the majority color even at stationarity.
Each configuration corresponds to a set of contours - informally, a bridge system comprises of a set of bridges, which are edges on the dual graph on the lattice that connect contours to the boundary of the lattice. Contours that are connected this way are called bridged contours, while the remaining contours are unbridged.

Bridge systems are defined so that the total length of the bridges is at most a constant fraction of the total length of the bridged contours, 
which allows us to bound the number of bridge systems with total bridged contour length $\ell$ by $C^\ell$ for some constant $C$. Consequently, a \emph{Peierls argument} can be used to show that the gain in energy (probability weight) by the removal of the bridged contours is greater than the loss in entropy by the removal of these contours.
Explicit constructions of bridge systems are shown in \cite{cannon_local_2019} and in our proof of alignment for general SOPS (see
Section~\ref{sec:general}).

\section{Compression and Alignment in Connected SOPS}


Starting with any simply connected set of particles, we define a local Markov chain aiming to simultaneously compresses the configuration and  align all but a small fraction of their orientations.
On each iteration, a particle is activated uniformly at random using a Poisson clock.
When activated, a particle chooses to attempt a spatial move or a reorientation move with a equal probability. Informally, spatial moves consist of the particle moving to a randomly chosen neighboring site, provided that site is unoccupied and the particle configuration remains simply connected, while a reorientation move allows the particle to change its orientation to point in a new direction.  While it is surprising that a property such as connectivity can be determined locally, a set of local moves were defined in Cannon et al. \cite{cannon_markov_2016} that prevent the configuration from disconnecting or forming holes and yet the chain remains ergodic on the infinite lattice, so all valid configurations can still be reached.  This ergodicity result carries over to our setting as the we use a lattice that while finite, is sufficiently large that self-intersections via wraparound are not possible.
%
Using the Metropolis-Hastings algorithm \cite{metropolis_equation_1953}, once a move is determined to be valid, it is implemented with probability $\min\{1,\pi(\sigma')/\pi(\sigma)\}$, where $\pi$ is the desired stationary distribution.

%

More precisely, consider a spatial move from a location $\ell$ to an empty adjacent location $\ell'$. Let the sets of lattice sites adjacent to the locations $\ell$ and $\ell'$ be $N(\ell)$ and $N(\ell')$ respectively. Furthermore, let $N(\ell \cup \ell')$ denote $N(\ell)\cup N(\ell') \setminus \{\ell, \ell'\}$, and $\mathbb{S}:= N(\ell) \cap N(\ell')$ denote the set of sites adjacent to both $\ell$ and $\ell'$ so that $|\mathbb{S}| \in \{0,1,2\}$.

\begin{definition}
A move from $\ell$ to $\ell'$ is {\it valid} if $\ell'$ is unoccupied, the number of particle-occupied sites in $N(\ell)$ is less than $5$, and either of the following two properties are satisfied:

\emph{Property 1:} $|\mathbb{S}| \geq 1$ and every particle-occupied site in $N(\ell \cup \ell')$ is connected to a particle-occupied site in $\mathbb{S}$ through $N(\ell \cup \ell')$.

\emph{Property 2:} $|\mathbb{S}| = 0$,  $\ell$ and $\ell'$ each have at least one neighbor, and all particle-occupied sites in $N(\ell) \setminus \{\ell'\}$ are connected by paths within this set, and all occupied sites in $N(\ell') \setminus \{\ell\}$ are connected by paths within this set.
\end{definition}

\noindent Note that in Section~\ref{sec:general}, we will consider almost the same algorithm in the general SOPS setting where there are no connectivity restrictions, so there all spatial moves from an occupied site to an adjacent unoccupied site are valid.

It is important to note that the ratio between the probabilities $\pi(\sigma')/\pi(\sigma)$ that arises from the Metropolis-Hastings algorithm can be calculated by an activated particle  using only local information - the positions and orientations of particles in its immediate neighborhood, as well as those in the  neighborhood of the destination site if the particle is moving.
Specifically, changes in perimeter in connected SOPS can be computed locally as shown in \cite{cannon_markov_2016, cannon_markov_2016-1}. 

We now proceed to show that when s $\lambda$ and $\gamma$ are sufficiently large, the alignment algorithm will cause the system to compress to form a low-perimeter configurations with high probability. Moreover, in both the Potts and clock model settings, in any configuration with sufficiently low-perimeter, one of the $q$ orientations will dominate with high probability.  

We note that we did not attempt to give rigorous bounds on the rates of convergence for our Markov chains.  We expect that convergence will be fast when the parameters $\lambda$ and $\gamma$ are small and the system evolves to a disordered (gaseous) state, but the connectivity constraint makes proving this challenging.  In contrast, we expect convergence to equilibrium will be slow in the ordered (solid) state when $\lambda$ is large, but we conjecture that desirable compressed and aligned states will be reached quickly, long before the system is very close to stationarity.



\subsection{Compression in Connected SOPS}

\newcommand{\polymerset}{{\Omega_\mathcal{P}^\mathcal{L}}}
\newcommand{\Omegaref}{\Omega_\mathcal{P}^{0}}
We denote the set of possible configurations in this paradigm by $\Omega$.
Recall that $N$ represents the number of sites of the lattice $G_\Delta$.
To ensure that the proof of ergodicity from~\cite{cannon_markov_2016-1} carries over to our setting, we use a sufficiently large value of $N$, namely $N \geq (n+1)^2$, although we expect the  results to hold for  smaller~$N$.

\begin{definition}[Compression]
A simply connected configuration $\sigma$ of $n$ particles on a lattice is said to be $\alpha$-compressed if its perimeter is at most $\alpha \cdot p_{\min}(n)$, where $p_{\min}(n)$ is the minimum possible perimeter of a configuration of $n$ particles.
\end{definition}


The main result of this section is the following theorem. 

\begin{theorem}\label{thm:con_comp_thm}
Given any $\alpha >1$, if constants $\lambda>1$ and $\gamma>29.3\,(q-1)$ satisfy
$
\lambda\,\gamma > 7^{\alpha/(\alpha-1)} 
$
and $n$ is sufficiently large, then the probability a configuration drawn from the stationary distribution $\pi_\mathrm{Potts}$ is not $\alpha$-compressed is exponentially small.
\end{theorem}




Let $\mathcal{P}$ denote the boundary of some configuration $\sigma$ in our configuration space $\Omega$. As~$\sigma$ is connected, hole-free, and contains a finite ($n$) number of particles, $\mathcal{P}$ is a single closed walk on $G_\Delta$ and the perimeter of the configuration, $p(\sigma)$, is equal to $\abs{\mathcal{P}}$, the total length of walk $\mathcal{P}$.
If we restrict our particle configurations to be connected and hole-free, there is a one-to-one correspondence between the possible sets of occupied sites and the possible boundaries $\mathcal{P}$.
Let $\Omega_\mathcal{P}$ denote the set of configurations in $\Omega$ with boundary $\mathcal{P}$, and let $\Lambda_\mathcal{P} \subseteq G_\Delta$ be the induced subgraph of the triangular lattice $G_\Delta$ by the particle-occupied vertices for any configuration in $\Omega_\mathcal{P}$. A configuration in $\Omega_\mathcal{P}$ thus corresponds to a mapping of the vertices of $\Lambda_\mathcal{P}$ to the orientations $\{0,\ldots,q-1\}$.

We consider the subset of configurations $\Omega_\mathcal{P}^0 \subseteq \Omega_\mathcal{P}$ where all particles on the boundary~$\mathcal{P}$ have the same color $0$. We will later analyze the weight of configurations in $\Omega_\mathcal{P}^0$ using a polymer model and the cluster expansion. 
We would first like to obtain an upper bound on $w(\Omega_\mathcal{P})$, the total weight of configurations in $\Omega_\mathcal{P}$, in terms of $w(\Omega_\mathcal{P}^0)$, the total weight of configurations in $\Omega_\mathcal{P}^0$. 

\markred{
Consider some fixed boundary $\mathcal{P}$, and take a configuration of particles $\sigma \in \Omega_\mathcal{P}$. Consider the set $E_H(\sigma)$ of heterogeneous edges (edges between particles of differing orientations) in $\Lambda_\mathcal{P}$. These edges correspond to a ``network'' of contours in $G_{\varhexagon}$. We denote by $E_C(\sigma)$ the set of heterogeneous edges that have a path to the boundary $\mathcal{P}$ over $E_H(\sigma)$ in $G_{\varhexagon}$.

Removing the edges $E_H(\sigma)$ from $\Lambda_\mathcal{P}$ subdivides $V(\Lambda_\mathcal{P})$ into connected components. A \emph{face} $F$ of the configuration refers to a union of the vertex sets of one or more of these components, that is connected, hole-free, and has all particles on its boundary of the same orientation. The orientation of a face $F$ refers to the common orientation shared by all particles on the boundary of $F$.
If a face $F$ is not a subset of any other face of $\sigma$, we call $F$ a \emph{maximal face}. We observe that there exists a unique partition of $V(\Lambda_\mathcal{P})$ into maximal faces, which can be obtained by removing only the edges in $E_C(\sigma)$ from $\Lambda_\mathcal{P}$, and taking the vertex sets of the resulting connected components of $\Lambda_\mathcal{P}$.
}

\begin{lemma}
\label{lem:monochromatic}
For $\gamma > 3q$, we have
\[w(\Omega_\mathcal{P}) < w(\Omega_{\mathcal{P}}^0) \cdot q\, 2^{\abs{\mathcal{P}}} \frac{\gamma}{\gamma - 3q}. \]
\end{lemma}

\markred{
\begin{proof}
Fix a configuration $\sigma \in \Omega_\mathcal{P}$. The set of heterogeneous edges $E_C(\sigma)$ connected to the boundary $\mathcal{P}$ partitions the vertices $V(\Lambda_\mathcal{P})$ into a set of maximal faces $\mathcal{F}(\sigma)$.

We first partition $\Omega_\mathcal{P}$ by $x = \abs{E_C(\sigma)}$ into the sets $\{\Omega_{\mathcal{P},x} \mid x = \{0,1,2,\ldots\}\}$. Fix some $x \in \{0,1,2,\ldots\}\}$. We define a map $\phi_x : \Omega_{\mathcal{P},x} \to \Omega_\mathcal{P}^0$. The map converts every face $F \in \mathcal{F}(\sigma)$ to a face of orientation $0$ as follows: for each face $F \in \mathcal{F}$ of some orientation $i \in \{0,\ldots,q-1\}$, a cyclic shift $j \mapsto (j - (i-1)) \mod q$ is applied to the orientation of every particle in $F$. As this removes exactly the heterogeneous contours corresponding to $E_C(\sigma)$, we observe that
\begin{align*}
    w(\sigma) = \gamma^{-x}w(\phi(\sigma)).
\end{align*}
We upper bound the size of the pre-image $\phi_x^{-1}(\tau)$ of any $\tau \in \Omega_{\mathcal{P}}^0$. We observe that a pre-image $\sigma \in \phi_x^{-1}(\tau)$ is fully defined by the contours $E_C(\sigma)$ and the original orientations of each of the faces. To recreate $E_C(\sigma)$, we designate whether each edge of the boundary is the starting point of a contour, of which there are at most $2^{\abs{\mathcal{P}}}$ ways to do so. There are then at most $3^x$ ways of reconstructing $E_C(\sigma)$ (where $\abs{E_C(\sigma)} = x$)  using the given starting points. This is as the dual lattice $G_{\varhexagon}$ is a regular graph of degree $3$, so at each juncture, we can choose to either go left, go right, or split into two paths. This defines our set of maximal faces $\mathcal{F}(\sigma)$. We then observe that adjacent faces are each separated by at least one (unique) edge in $E_C(\sigma)$, so as the auxiliary graph defined over the maximal faces $\mathcal{F}(\sigma)$ is connected, we must have $\abs{\mathcal{F}(\sigma)} \leq x+1$. This gives us at most $q^{x+1}$ ways to color these faces. Thus, 
\begin{align*}
    w(\Omega_{\mathcal{P},x})
    = \gamma^{-x}\sum_{\sigma \in \Omega_{\mathcal{P},x}}w(\phi(\sigma))
    \leq \gamma^{-x} 2^{\abs{\mathcal{P}}} 3^x q^{x+1} w(\Omega_{\mathcal{P}}^0).
\end{align*}
Summing this over $x \in \{0,1,\ldots\}$, we have for $\gamma > 3q$,
\begin{align*}
    w(\Omega_{\mathcal{P}}) = \sum_{x=0}^\infty \leq w(\Omega_{\mathcal{P}}^0) \cdot q 2^{\abs{\mathcal{P}}} \sum_{x=0}^\infty {\left(\frac{3q}{\gamma}\right)}^x
    = w(\Omega_{\mathcal{P}}^0) \cdot q 2^{\abs{\mathcal{P}}} \frac{\gamma}{\gamma - 3q}.
\end{align*}
\end{proof}
}

The proof is a generalized version of that in \cite{cannon_local_2019}, by defining maps from $\Omega_\mathcal{P} \to \Omegaref$ such that all vertices on boundary $\mathcal{P}$ are of orientation $0$.
We will use the cluster expansion to analyze the total weight $w(\Omega_\mathcal{P}) := \sum_{\sigma \in \Omega_\mathcal{P}} w(\sigma)$ of the configurations in $\Omega_\mathcal{P}$. Since the cluster expansion can only be applied to polymer partition functions, we begin by representing the configurations of $\Omega_\mathcal{P}$ with a polymer model.
\markred{The cluster expansion is one of the oldest tools in statistical physics \cite{mayer_statistical_1937,mayer_theory_1950, friedli_cluster_2017}, and has led to the development of the Pirogov-Sinai theory \cite{pirogov_phase_1975, pirogov_phase_1976}, and has played an important part in the recent advances in efficient sampling and counting algorithms \cite{helmuth_algorithmic_2019,jenssen_independent_2020,borgs_efficient_2020}.

The cluster expansion expresses the logarithm of a polymer partition function as a sum over polymer clusters. Let $\mathcal{L}$ be a finite set of polymers $\{\xi_i\}$, where each polymer $\xi_i$ has a weight $w(\xi)$.
We also define a notion of ``compatibility'' between polymers - each pair of polymers $\xi, \xi'$ is either compatible ($\xi\sim\xi'$) or incompatible ($\xi\nsim\xi'$).
The polymer partition function is then given by:
\[ \Xi = \sum_{\tau\in\Omega^\mathcal{L}} \prod_{\xi\in \tau}w(\xi) \,,\]
where $\Omega^\mathcal{L}$ is the set of all collections of pairwise compatible polymers in $\mathcal{L}$. 
The cluster expansion expresses the logarithm of the polymer partition function in terms of clusters, where a cluster $X$ is an ordered multiset of polymers $\{\xi_1,\ldots, \xi_k\}$ such that their incompatibility graph $H(X)$ is connected, where the incompatibility graph is constructed by representing  each polymer by a vertex and connecting two vertices if the corresponding polymers are incompatible. The cluster expansion gives:
\begin{align*}
\log\Xi = \sum_{X\in \mathcal{C}} \Psi(X)\, \text{, where }
\Psi(X) := \frac{1}{\vert X \vert !}\left( \sum_{\substack{G\subseteq H_X\\\text{$G$ connected, spanning} } } (-1)^{\vert E(G)\vert}\right)\left( \prod_{\xi\in X} w(\xi)\right)\,,
\end{align*}
and $\mathcal{C}$ is the set of all clusters. A sufficient condition for the convergence of the cluster expansion was given by Koteck{\'y} and Preiss \cite{kotecky_cluster_1986}. We  prove this condition in Lemma \ref{lem:kotecky_preiss_potts} and use the cluster expansion to separate the volume and surface contributions to the partition function, as done in \cite{friedli_cluster_2017, cannon_local_2019}. }


\subparagraph*{The Polymer Model:}
We say two edges of $G_\Delta$ are adjacent if they share a common vertex. A polymer $\xi$ in $\mathcal{L}$ is defined to be a labeling $\xi: E(G_\Delta) \to \{0,1,\ldots,q-1\}$ of the edges of~$G_\Delta$ such that the set $E(\xi)$, defined to be the edges of $G_\Delta$ with a non-zero label in $\xi$, is non-empty and connected under the above notion of adjacency.
The labeling must also be \emph{consistent}, as defined below.

\begin{definition}[Consistent Labeling]
\label{defn:consistent}
We fix a canonical direction for each edge in $G_\Delta$. This direction can be arbitrarily defined, so for simplicity we say that the edge is oriented toward the vertex with the larger $x$, followed by $y$ coordinate, where the coordinate axes are oriented such that the $x$ coordinate increases from left to right and the $y$ coordinate increases from top to bottom.

We define labels $\xi : E(G_\Delta) \to \{0,1,\ldots,q-1\}$. These edge labels represent ``flows'' in our defined canonical direction, modulo $q$. In other words, when summing up the total flow along a walk on $G_\Delta$, for each edge $e$ on the walk, we add the label $\xi(e)$ to the sum if the walk is in the canonical direction of the edge, and $q-\xi(e)$ if the walk is in the opposite direction. We call an assignment of labels \emph{consistent} if every closed walk on $G_\Delta$ has a total flow summing to $0$ modulo $q$.

\end{definition}

Consider a fixed boundary $\mathcal{P}$ as defined above, corresponding to some configuration in~$\Omega$. For a polymer $\xi$, denote by $V(\xi)$ the set of vertices incident to an edge with a non-zero label in $\xi$. We say a polymer $\xi$ is within $\mathcal{P}$ if $V(\xi) \subseteq \Lambda_\mathcal{P}$.
As described earlier, the set~$\polymerset$ of polymer configurations corresponding to $\mathcal{P}$ is the set of all subsets of $\mathcal{L}$ of pairwise compatible polymers within $\mathcal{P}$. The weight $w(\tau)$ of a configuration $\tau \in \polymerset$ is the product of the weights of its constituent polymers.

Two polymers $\xi_1$, $\xi_2$ are incompatible if there are edges $e_1 \in E(\xi_1)$ and $e_2 \in E(\xi_2)$ such that $e_1$ and $e_2$ are adjacent. The weight of a polymer $\xi$ is defined as $w(\xi) := \gamma^{-\abs{E(\xi)}}$, in the Potts model, and $w(\xi) := \prod_{e \in E(\xi)}\gamma^{\cos\parens{\frac{2\pi}{q}\xi(e)}-1}$ in the clock model.

\begin{figure}
\begin{subfigure}[b]{.14\textwidth}
\begin{center}\begin{tikzpicture}[x=0.6cm,y=0.6cm]
\draw[black, line width=0.3mm] (1,4) circle (0.12);
\draw[fill=black] (1,3) circle (0.12);
\draw[line width=0.5mm] (1.12728,2.12728) -- (0.872721,1.87272);
\draw[line width=0.5mm] (0.872721,2.12728) -- (1.12728,1.87272);
\draw[line width=0.3mm] (1.14549,0.916) -- (1,1.168);
\draw[line width=0.3mm] (1,1.168) -- (0.854508,0.916);
\draw[line width=0.3mm] (0.854508,0.916) -- (1.14549,0.916);
\node[align=left] at (1.5,5) {\scriptsize Orientations};
\node[align=left] at (1.7,4) {\scriptsize = 0};
\node[align=left] at (1.7,3) {\scriptsize = 1};
\node[align=left] at (1.7,2) {\scriptsize = 2};
\node[align=left] at (1.7,1) {\scriptsize = 3};
\end{tikzpicture}
\end{center}
\vskip 0pt
\end{subfigure}
\begin{subfigure}[b]{.4\textwidth}
\begin{center}\begin{tikzpicture}[x=0.6cm,y=0.6cm]
\draw[lightgray] (0,-4) -- (1.73205,-5);
\draw[lightgray] (0,-1) -- (1.73205,0);
\draw[lightgray] (0,-3) -- (3.4641,-5);
\draw[lightgray] (0,-2) -- (3.4641,0);
\draw[lightgray] (0,-2) -- (5.19615,-5);
\draw[lightgray] (0,-3) -- (5.19615,0);
\draw[lightgray] (0,-1) -- (6.9282,-5);
\draw[lightgray] (0,-4) -- (6.9282,0);
\draw[lightgray] (0,0) -- (8.66025,-5);
\draw[lightgray] (0,-5) -- (8.66025,0);
\draw[lightgray] (1.73205,0) -- (8.66025,-4);
\draw[lightgray] (1.73205,-5) -- (8.66025,-1);
\draw[lightgray] (3.4641,0) -- (8.66025,-3);
\draw[lightgray] (3.4641,-5) -- (8.66025,-2);
\draw[lightgray] (5.19615,0) -- (8.66025,-2);
\draw[lightgray] (5.19615,-5) -- (8.66025,-3);
\draw[lightgray] (6.9282,0) -- (8.66025,-1);
\draw[lightgray] (6.9282,-5) -- (8.66025,-4);
\draw[lightgray] (8.66025,0) -- (8.66025,0);
\draw[lightgray] (8.66025,-5) -- (8.66025,-5);
\draw[lightgray] (0,0) -- (0,-5);
\draw[lightgray] (0.866025,-0.5) -- (0.866025,-4.5);
\draw[lightgray] (1.73205,0) -- (1.73205,-5);
\draw[lightgray] (2.59808,-0.5) -- (2.59808,-4.5);
\draw[lightgray] (3.4641,0) -- (3.4641,-5);
\draw[lightgray] (4.33013,-0.5) -- (4.33013,-4.5);
\draw[lightgray] (5.19615,0) -- (5.19615,-5);
\draw[lightgray] (6.06218,-0.5) -- (6.06218,-4.5);
\draw[lightgray] (6.9282,0) -- (6.9282,-5);
\draw[lightgray] (7.79423,-0.5) -- (7.79423,-4.5);
\draw[lightgray] (8.66025,0) -- (8.66025,-5);
\draw[black, line width=0.3mm] (0.866025,-0.5) circle (0.12);
\draw[black, line width=0.3mm] (0.866025,-1.5) circle (0.12);
\draw[black, line width=0.3mm] (0.866025,-3.5) circle (0.12);
\draw[black, line width=0.3mm] (1.73205,0) circle (0.12);
\draw[black, line width=0.3mm] (1.73205,-1) circle (0.12);
\draw[black, line width=0.3mm] (1.73205,-2) circle (0.12);
\draw[black, line width=0.3mm] (1.73205,-3) circle (0.12);
\draw[black, line width=0.3mm] (1.73205,-4) circle (0.12);
\draw[black, line width=0.3mm] (2.59808,-0.5) circle (0.12);
\draw[black, line width=0.3mm] (2.59808,-1.5) circle (0.12);
\draw[black, line width=0.3mm] (2.59808,-2.5) circle (0.12);
\draw[line width=0.3mm] (2.74357,-3.584) -- (2.59808,-3.332);
\draw[line width=0.3mm] (2.59808,-3.332) -- (2.45258,-3.584);
\draw[line width=0.3mm] (2.45258,-3.584) -- (2.74357,-3.584);
\draw[black, line width=0.3mm] (2.59808,-4.5) circle (0.12);
\draw[black, line width=0.3mm] (3.4641,0) circle (0.12);
\draw[fill=black] (3.4641,-1) circle (0.12);
\draw[black, line width=0.3mm] (3.4641,-2) circle (0.12);
\draw[black, line width=0.3mm] (3.4641,-3) circle (0.12);
\draw[black, line width=0.3mm] (3.4641,-4) circle (0.12);
\draw[black, line width=0.3mm] (4.33013,-0.5) circle (0.12);
\draw[black, line width=0.3mm] (4.33013,-1.5) circle (0.12);
\draw[black, line width=0.3mm] (4.33013,-2.5) circle (0.12);
\draw[black, line width=0.3mm] (4.33013,-3.5) circle (0.12);
\draw[black, line width=0.3mm] (4.33013,-4.5) circle (0.12);
\draw[black, line width=0.3mm] (5.19615,0) circle (0.12);
\draw[black, line width=0.3mm] (5.19615,-1) circle (0.12);
\draw[line width=0.5mm] (5.32343,-1.87272) -- (5.06887,-2.12728);
\draw[line width=0.5mm] (5.06887,-1.87272) -- (5.32343,-2.12728);
\draw[line width=0.5mm] (5.32343,-2.87272) -- (5.06887,-3.12728);
\draw[line width=0.5mm] (5.06887,-2.87272) -- (5.32343,-3.12728);
\draw[black, line width=0.3mm] (5.19615,-4) circle (0.12);
\draw[black, line width=0.3mm] (6.06218,-0.5) circle (0.12);
\draw[fill=black] (6.06218,-1.5) circle (0.12);
\draw[line width=0.5mm] (6.18946,-2.37272) -- (5.9349,-2.62728);
\draw[line width=0.5mm] (5.9349,-2.37272) -- (6.18946,-2.62728);
\draw[line width=0.5mm] (6.18946,-3.37272) -- (5.9349,-3.62728);
\draw[line width=0.5mm] (5.9349,-3.37272) -- (6.18946,-3.62728);
\draw[black, line width=0.3mm] (6.06218,-4.5) circle (0.12);
\draw[black, line width=0.3mm] (6.9282,0) circle (0.12);
\draw[line width=0.3mm] (7.07369,-1.084) -- (6.9282,-0.832);
\draw[line width=0.3mm] (6.9282,-0.832) -- (6.78271,-1.084);
\draw[line width=0.3mm] (6.78271,-1.084) -- (7.07369,-1.084);
\draw[fill=black] (6.9282,-2) circle (0.12);
\draw[fill=black] (6.9282,-3) circle (0.12);
\draw[black, line width=0.3mm] (6.9282,-4) circle (0.12);
\draw[black, line width=0.3mm] (7.79423,-0.5) circle (0.12);
\draw[black, line width=0.3mm] (7.79423,-1.5) circle (0.12);
\draw[black, line width=0.3mm] (7.79423,-2.5) circle (0.12);
\draw[black, line width=0.3mm] (7.79423,-3.5) circle (0.12);
\draw[black, line width=0.3mm] (8.66025,-2) circle (0.12);
\draw[black, line width=0.3mm] (8.66025,-3) circle (0.12);
\end{tikzpicture}
\end{center}
\end{subfigure}\begin{subfigure}[b]{.4\textwidth}
\begin{center}\begin{tikzpicture}[x=0.6cm,y=0.6cm]
\draw[lightgray] (0,-4) -- (1.73205,-5);
\draw[lightgray] (0,-1) -- (1.73205,0);
\draw[lightgray] (0,-3) -- (3.4641,-5);
\draw[lightgray] (0,-2) -- (3.4641,0);
\draw[lightgray] (0,-2) -- (5.19615,-5);
\draw[lightgray] (0,-3) -- (5.19615,0);
\draw[lightgray] (0,-1) -- (6.9282,-5);
\draw[lightgray] (0,-4) -- (6.9282,0);
\draw[lightgray] (0,0) -- (8.66025,-5);
\draw[lightgray] (0,-5) -- (8.66025,0);
\draw[lightgray] (1.73205,0) -- (8.66025,-4);
\draw[lightgray] (1.73205,-5) -- (8.66025,-1);
\draw[lightgray] (3.4641,0) -- (8.66025,-3);
\draw[lightgray] (3.4641,-5) -- (8.66025,-2);
\draw[lightgray] (5.19615,0) -- (8.66025,-2);
\draw[lightgray] (5.19615,-5) -- (8.66025,-3);
\draw[lightgray] (6.9282,0) -- (8.66025,-1);
\draw[lightgray] (6.9282,-5) -- (8.66025,-4);
\draw[lightgray] (8.66025,0) -- (8.66025,0);
\draw[lightgray] (8.66025,-5) -- (8.66025,-5);
\draw[lightgray] (0,0) -- (0,-5);
\draw[lightgray] (0.866025,-0.5) -- (0.866025,-4.5);
\draw[lightgray] (1.73205,0) -- (1.73205,-5);
\draw[lightgray] (2.59808,-0.5) -- (2.59808,-4.5);
\draw[lightgray] (3.4641,0) -- (3.4641,-5);
\draw[lightgray] (4.33013,-0.5) -- (4.33013,-4.5);
\draw[lightgray] (5.19615,0) -- (5.19615,-5);
\draw[lightgray] (6.06218,-0.5) -- (6.06218,-4.5);
\draw[lightgray] (6.9282,0) -- (6.9282,-5);
\draw[lightgray] (7.79423,-0.5) -- (7.79423,-4.5);
\draw[lightgray] (8.66025,0) -- (8.66025,-5);
\draw[black] (0.866025,-0.5) circle (0.096);
\draw[black] (0.866025,-1.5) circle (0.096);
\draw[black] (0.866025,-3.5) circle (0.096);
\draw[black] (1.73205,0) circle (0.096);
\draw[black] (1.73205,-1) circle (0.096);
\draw[black] (1.73205,-2) circle (0.096);
\draw[black] (1.73205,-3) circle (0.096);
\draw[black] (1.73205,-4) circle (0.096);
\draw[black] (2.59808,-0.5) circle (0.096);
\draw[black] (2.59808,-1.5) circle (0.096);
\draw[black] (2.59808,-2.5) circle (0.096);
\draw[black] (2.59808,-3.5) circle (0.096);
\draw[black] (2.59808,-4.5) circle (0.096);
\draw[black] (3.4641,0) circle (0.096);
\draw[black] (3.4641,-1) circle (0.096);
\draw[black] (3.4641,-2) circle (0.096);
\draw[black] (3.4641,-3) circle (0.096);
\draw[black] (3.4641,-4) circle (0.096);
\draw[black] (4.33013,-0.5) circle (0.096);
\draw[black] (4.33013,-1.5) circle (0.096);
\draw[black] (4.33013,-2.5) circle (0.096);
\draw[black] (4.33013,-3.5) circle (0.096);
\draw[black] (4.33013,-4.5) circle (0.096);
\draw[black] (5.19615,0) circle (0.096);
\draw[black] (5.19615,-1) circle (0.096);
\draw[black] (5.19615,-2) circle (0.096);
\draw[black] (5.19615,-3) circle (0.096);
\draw[black] (5.19615,-4) circle (0.096);
\draw[black] (6.06218,-0.5) circle (0.096);
\draw[black] (6.06218,-1.5) circle (0.096);
\draw[black] (6.06218,-2.5) circle (0.096);
\draw[black] (6.06218,-3.5) circle (0.096);
\draw[black] (6.06218,-4.5) circle (0.096);
\draw[black] (6.9282,0) circle (0.096);
\draw[black] (6.9282,-1) circle (0.096);
\draw[black] (6.9282,-2) circle (0.096);
\draw[black] (6.9282,-3) circle (0.096);
\draw[black] (6.9282,-4) circle (0.096);
\draw[black] (7.79423,-0.5) circle (0.096);
\draw[black] (7.79423,-1.5) circle (0.096);
\draw[black] (7.79423,-2.5) circle (0.096);
\draw[black] (7.79423,-3.5) circle (0.096);
\draw[black] (8.66025,-2) circle (0.096);
\draw[black] (8.66025,-3) circle (0.096);
\draw[black,-{Stealth[length=1.3mm,width=1.3mm]}] (1.77535,-3.025) -- (2.55477,-3.475);
\node[align=center] at (2.16506,-3.4) {\tiny 3};
\draw[black,-{Stealth[length=1.3mm,width=1.3mm]}] (1.77535,-3.975) -- (2.55477,-3.525);
\node[align=center] at (2.16506,-3.9) {\tiny 3};
\draw[black,-{Stealth[length=1.3mm,width=1.3mm]}] (2.64138,-0.525) -- (3.4208,-0.975);
\node[align=center] at (3.03109,-0.9) {\tiny 1};
\draw[black,-{Stealth[length=1.3mm,width=1.3mm]}] (2.64138,-1.475) -- (3.4208,-1.025);
\node[align=center] at (3.03109,-1.4) {\tiny 1};
\draw[black,-{Stealth[length=1.3mm,width=1.3mm]}] (2.59808,-2.55) -- (2.59808,-3.45);
\node[align=center] at (2.72798,-3) {\tiny 3};
\draw[black,-{Stealth[length=1.3mm,width=1.3mm]}] (2.59808,-3.55) -- (2.59808,-4.45);
\node[align=center] at (2.72798,-4) {\tiny 1};
\draw[black,-{Stealth[length=1.3mm,width=1.3mm]}] (2.64138,-3.475) -- (3.4208,-3.025);
\node[align=center] at (3.03109,-3.4) {\tiny 1};
\draw[black,-{Stealth[length=1.3mm,width=1.3mm]}] (2.64138,-3.525) -- (3.4208,-3.975);
\node[align=center] at (3.03109,-3.9) {\tiny 1};
\draw[black,-{Stealth[length=1.3mm,width=1.3mm]}] (3.4641,-0.05) -- (3.4641,-0.95);
\node[align=center] at (3.59401,-0.5) {\tiny 1};
\draw[black,-{Stealth[length=1.3mm,width=1.3mm]}] (3.4641,-1.05) -- (3.4641,-1.95);
\node[align=center] at (3.59401,-1.5) {\tiny 3};
\draw[black,-{Stealth[length=1.3mm,width=1.3mm]}] (3.5074,-0.975) -- (4.28683,-0.525);
\node[align=center] at (3.89711,-0.9) {\tiny 3};
\draw[black,-{Stealth[length=1.3mm,width=1.3mm]}] (3.5074,-1.025) -- (4.28683,-1.475);
\node[align=center] at (3.89711,-1.4) {\tiny 3};
\draw[black,-{Stealth[length=1.3mm,width=1.3mm]}] (4.37343,-1.525) -- (5.15285,-1.975);
\node[align=center] at (4.76314,-1.9) {\tiny 2};
\draw[black,-{Stealth[length=1.3mm,width=1.3mm]}] (4.37343,-2.475) -- (5.15285,-2.025);
\node[align=center] at (4.76314,-2.4) {\tiny 2};
\draw[black,-{Stealth[length=1.3mm,width=1.3mm]}] (4.37343,-2.525) -- (5.15285,-2.975);
\node[align=center] at (4.76314,-2.9) {\tiny 2};
\draw[black,-{Stealth[length=1.3mm,width=1.3mm]}] (4.37343,-3.475) -- (5.15285,-3.025);
\node[align=center] at (4.76314,-3.4) {\tiny 2};
\draw[black,-{Stealth[length=1.3mm,width=1.3mm]}] (5.19615,-1.05) -- (5.19615,-1.95);
\node[align=center] at (5.32606,-1.5) {\tiny 2};
\draw[black,-{Stealth[length=1.3mm,width=1.3mm]}] (5.23945,-1.025) -- (6.01888,-1.475);
\node[align=center] at (5.62917,-1.4) {\tiny 1};
\draw[black,-{Stealth[length=1.3mm,width=1.3mm]}] (5.23945,-1.975) -- (6.01888,-1.525);
\node[align=center] at (5.62917,-1.9) {\tiny 3};
\draw[black,-{Stealth[length=1.3mm,width=1.3mm]}] (5.19615,-3.05) -- (5.19615,-3.95);
\node[align=center] at (5.32606,-3.5) {\tiny 2};
\draw[black,-{Stealth[length=1.3mm,width=1.3mm]}] (5.23945,-3.975) -- (6.01888,-3.525);
\node[align=center] at (5.62917,-3.9) {\tiny 2};
\draw[black,-{Stealth[length=1.3mm,width=1.3mm]}] (6.06218,-0.55) -- (6.06218,-1.45);
\node[align=center] at (6.19208,-1) {\tiny 1};
\draw[black,-{Stealth[length=1.3mm,width=1.3mm]}] (6.10548,-0.525) -- (6.8849,-0.975);
\node[align=center] at (6.49519,-0.9) {\tiny 3};
\draw[black,-{Stealth[length=1.3mm,width=1.3mm]}] (6.06218,-1.55) -- (6.06218,-2.45);
\node[align=center] at (6.19208,-2) {\tiny 1};
\draw[black,-{Stealth[length=1.3mm,width=1.3mm]}] (6.10548,-1.475) -- (6.8849,-1.025);
\node[align=center] at (6.49519,-1.4) {\tiny 2};
\draw[black,-{Stealth[length=1.3mm,width=1.3mm]}] (6.10548,-2.475) -- (6.8849,-2.025);
\node[align=center] at (6.49519,-2.4) {\tiny 3};
\draw[black,-{Stealth[length=1.3mm,width=1.3mm]}] (6.10548,-2.525) -- (6.8849,-2.975);
\node[align=center] at (6.49519,-2.9) {\tiny 3};
\draw[black,-{Stealth[length=1.3mm,width=1.3mm]}] (6.06218,-3.55) -- (6.06218,-4.45);
\node[align=center] at (6.19208,-4) {\tiny 2};
\draw[black,-{Stealth[length=1.3mm,width=1.3mm]}] (6.10548,-3.475) -- (6.8849,-3.025);
\node[align=center] at (6.49519,-3.4) {\tiny 3};
\draw[black,-{Stealth[length=1.3mm,width=1.3mm]}] (6.10548,-3.525) -- (6.8849,-3.975);
\node[align=center] at (6.49519,-3.9) {\tiny 2};
\draw[black,-{Stealth[length=1.3mm,width=1.3mm]}] (6.9282,-0.05) -- (6.9282,-0.95);
\node[align=center] at (7.05811,-0.5) {\tiny 3};
\draw[black,-{Stealth[length=1.3mm,width=1.3mm]}] (6.9282,-1.05) -- (6.9282,-1.95);
\node[align=center] at (7.05811,-1.5) {\tiny 2};
\draw[black,-{Stealth[length=1.3mm,width=1.3mm]}] (6.9715,-0.975) -- (7.75093,-0.525);
\node[align=center] at (7.36122,-0.9) {\tiny 1};
\draw[black,-{Stealth[length=1.3mm,width=1.3mm]}] (6.9715,-1.025) -- (7.75093,-1.475);
\node[align=center] at (7.36122,-1.4) {\tiny 1};
\draw[black,-{Stealth[length=1.3mm,width=1.3mm]}] (6.9715,-1.975) -- (7.75093,-1.525);
\node[align=center] at (7.36122,-1.9) {\tiny 3};
\draw[black,-{Stealth[length=1.3mm,width=1.3mm]}] (6.9715,-2.025) -- (7.75093,-2.475);
\node[align=center] at (7.36122,-2.4) {\tiny 3};
\draw[black,-{Stealth[length=1.3mm,width=1.3mm]}] (6.9282,-3.05) -- (6.9282,-3.95);
\node[align=center] at (7.05811,-3.5) {\tiny 3};
\draw[black,-{Stealth[length=1.3mm,width=1.3mm]}] (6.9715,-2.975) -- (7.75093,-2.525);
\node[align=center] at (7.36122,-2.9) {\tiny 3};
\draw[black,-{Stealth[length=1.3mm,width=1.3mm]}] (6.9715,-3.025) -- (7.75093,-3.475);
\node[align=center] at (7.36122,-3.4) {\tiny 3};
\end{tikzpicture}
\end{center}
\end{subfigure}
\caption{Particle configuration in $\Omegaref$, and its corresponding polymer configuration in $\polymerset$ (with two polymers).}
\label{fig:polymermodel}
\end{figure}

\begin{lemma}
\label{lem:polymer_bijection_potts}
There is a bijection $\phi$ between $\Omegaref$ and $\polymerset$ with the property that for any $\sigma \in \Omegaref$, we have $w(\sigma) = (\lambda\gamma)^{-p(\sigma)}w(\phi(\sigma))$.
\end{lemma}

\markred{
\begin{proof}
Consider a particle configuration $\sigma \in \Omegaref$. Let $\theta_\sigma: V(\Lambda_\mathcal{P}) \to \{0,1,\ldots,q-1\}$ be the assignment of orientations to the particles of the configuration.
We define a labeling $\xi^*: E(\Lambda_\mathcal{P}) \to \{0,1,\ldots,q-1\}$ as follows: For an edge $(u,v) \in E(\Lambda_\mathcal{P})$, where the canonical direction of the edge as defined in Definition \ref{defn:consistent} is from $u$ to $v$, define $\xi^*((u,v)) := ( \theta_\sigma(v) - \theta_\sigma(u) ) \mod q$. For $(u,v) \in E(G_\Delta)\setminus E(\Lambda_\mathcal{P})$, we set $\xi^*((u,v)) = 0$

We observe that $\xi^*$ will often not correspond to a single polymer. Denote by $E(\xi^*)$ the edges of $G_\Delta$ with a non-zero label in $\xi^*$. We partition $E(\xi^*)$ into connected components by our notion of edge adjacency defined above. Each of these connected components corresponds to a polymer. More precisely, the polymer $\xi$ corresponding to a connected component with edge set $E$ is defined as
\begin{align*}
\xi(e) =
\begin{cases}
\xi^*(e) &\text{if $e \in E$}\\
0 &\text{otherwise}
\end{cases}.
\end{align*}
Clearly, $E(\xi) = E$ and $E(\xi)$ is connected. For any closed walk $v_1, v_2, \ldots, v_k = v_1$ over vertices in $G_\Delta$, if $v_i \in V(\Lambda_\mathcal{P})$ for all $i \in \{1,2,\ldots,k\}$ the total cost along this path would be equal to $\sum_{i=1}^{k-1} \theta_\sigma(v_{i+1}) - \theta_\sigma(v_{i})$ modulo $q$, which sums to $0$. If there are vertices $v_i \not\in V(\Lambda_\mathcal{P})$, we can split the closed walk into smaller walks, each starting and ending on the vertices incident to the boundary edges of $\mathcal{P}$. Each of these walks would have a total cost summing to zero for the same reason as before, as all vertices on the boundary have orientation $0$ in $\sigma$.

Thus, $\xi$ is consistent and is hence a polymer. This gives us a set of polymers $\tau := \{\xi_1,\cdots,\xi_m\}$, one from each of the connected components. These polymers are pairwise compatible as they are created from separate edge components. It is also clear that $V(\xi_i) \subseteq V(\Lambda_\mathcal{P})$ for all $i \in \{1,\ldots,m\}$.
This construction gives us a function $\phi$ from configurations in $\Omegaref$ to polymer configurations in $\polymerset$. It is also simple to check that for the Potts model,
\begin{align*}
    w(\sigma) = (\lambda\gamma)^{-p(\sigma)}\gamma^{-h(\sigma)} = (\lambda\gamma)^{-p(\sigma)}\prod_{\xi \in \tau}\gamma^{-\abs{E_\xi}} =(\lambda\gamma)^{-p(\sigma)}w(\tau),
\end{align*}
and similarly for the clock model,
\begin{align*}
    w(\sigma) = (\lambda\gamma)^{-p(\sigma)}\prod_{(i,j)}\gamma^{-d_{ij}} = (\lambda\gamma)^{-p(\sigma)}\prod_{\xi \in \tau}\prod_{e \in E_\xi}\gamma^{\cos\parens{\frac{2\pi}{q}\xi(e)}-1} =(\lambda\gamma)^{-p(\sigma)}w(\tau).
\end{align*}

To show that $\phi$ is injective, consider any two particle configurations $\sigma_1, \sigma_2 \in \Omegaref$. If $\sigma_1 \neq \sigma_2$, there must be a vertex $v$ of $\Lambda_\mathcal{P}$ such that $\theta_{\sigma_1}(v) \neq \theta_{\sigma_2}(v)$. Let $u$ be a particle on the boundary $\mathcal{P}$. As $\theta_{\sigma_1}(u) = 1 = \theta_{\sigma_2}(u)$ and there is a path from $u$ to $v$ in $\Lambda_\mathcal{P}$, there must be two adjacent vertices $u', v'$ of $\Lambda_\mathcal{P}$ such that $\theta_{\sigma_1}(u') = \theta_{\sigma_2}(u')$ but $\theta_{\sigma_1}(v') \neq \theta_{\sigma_2}(v')$. This necessarily means the labeling $\xi^*$ will differ between $\sigma_1$ and $\sigma_2$ in the construction, and hence $\phi(\sigma_1)$ and $\phi(\sigma_2)$ will differ on at least one polymer, so $\phi(\sigma_1) \neq \phi(\sigma_2)$.

To show that $\phi$ is surjective, Take any polymer configuration $\tau = \{\xi_1,\xi_2,\ldots,\xi_m\} \in \polymerset$. We define a labeling $\xi^* := \sum_{i = 1}^{m} \xi_i$, which will also be consistent as each $\xi \in \tau$ is consistent. Set the orientations of all particles on the boundary $\mathcal{P}$ to $0$. For any particle $v$ in $\Lambda_\mathcal{P}$, there is a path $\{u = v_1,v_2,\ldots,v_k = v\}$ from an arbitrary $u$ on the boundary $\mathcal{P}$ to $v$. To compute the orientation of $v$, as in Definition \ref{defn:consistent}, we sum the labels on the edges along the path, adding $\xi^*(e)$ if the walk is in the canonical direction of the edge, and $q-\xi^*(e)$ if the walk is in the opposite direction. This sum is computed modulo $q$. This computed orientation of $v$ is independent of the path chosen as $\xi^*$ is consistent.

This gives us a mapping of the vertices of $\Lambda_\mathcal{P}$ to orientations, and hence a configuration $\sigma \in \Omegaref$. It can be easily verified that applying the construction on this configuration will give us the initial labeling $\xi^*$, and as there is one unique way to partition $\xi$ into edge components, we necessarily have $\phi(\sigma) = \tau$.
\end{proof}
}

The map $\phi$ simply encodes the orientations of particles in a configuration $\sigma \in \Omegaref$ as differences between orientations on the edges of $G_\Delta$. This is illustrated in Figure \ref{fig:polymermodel}. The full version of the paper gives a full description of this mapping and a proof that it is indeed a bijection.
From Lemma \ref{lem:polymer_bijection_potts}, we have
\begin{align*}
    w(\Omegaref) =
    \sum_{\sigma \in \Omegaref} \left(\lambda \gamma \right)^{-\abs{\mathcal{P}}} w(\phi(\sigma)) =
    \sum_{\tau \in \polymerset} \left(\lambda \gamma \right)^{-\abs{\mathcal{P}}} w(\tau) =
    \left(\lambda \gamma \right)^{-\abs{\mathcal{P}}}  \Xi_\mathcal{P},
\end{align*}
\noindent where $\Xi_\mathcal{P}$ is the partition function for the set of polymer configurations $\polymerset$:
\begin{align*}
    \Xi_\mathcal{P} := \sum_{\tau \in \polymerset} w(\tau) = \sum_{\tau \in \polymerset} \prod_{\xi \in \tau} w(\gamma).
\end{align*}

\subparagraph*{The Potts Model:}
From now, our analysis will be specific to the Potts model. The clock model will be discussed in Section \ref{section:clockmodelanalysis}.
The following Lemmas and proofs are slight variations of those used in \cite{cannon_local_2019}.

\begin{lemma}
\label{lem:kotecky_preiss_potts}
For any polymer $\xi \in \mathcal{L}$, whenever $\gamma > 29.3(q-1)$, we have for $c= 0.0001$, 
\[ \sum_{\substack{\xi'\in\mathcal{L} \\ \xi' \nsim \xi} } w(\xi') \exp(c\vert V(\xi')\vert) \leq c \vert V(\xi)\vert, \] 
where $V(\xi')$ denotes the set of vertices in the polymer $\xi'$, and $\vert V(\xi')\vert$ denotes the number of vertices in $\xi'$.
\end{lemma}

The proof is on the lines of that in \cite{cannon_local_2019}. The key part of this proof is the use of an upper bound $\nu(m,q) \leq (6e(q-1))^m/2$ from \cite{borgs_left_2010}, where $\nu(m,q)$ represents the number of polymers with $m$ edges containing some fixed vertex $v \in V(G_\Delta)$. 

\markred{
\begin{proof}
Since two polymers $\xi',\xi$ must contain a vertex in common to be incompatible: $\xi'\nsim \xi$, the summation on the left hand side of the sufficiency condition above satisfies:
\begin{align*}
    \sum_{\substack{\xi'\in\mathcal{L} \\ \xi' \nsim \xi} } \gamma^{-\vert E(\xi')\vert} \exp(c\vert V(\xi')\vert) &\leq \sum_{v\,\in V(\xi)} \sum_{\substack{\xi'\in\mathcal{L} \\ v \in  \xi'} } \gamma^{-\vert E(\xi')\vert} \exp(c\vert V(\xi')\vert) 
\end{align*}

To prove the lemma, it will be sufficient prove that the following condition is satisfied for $c=0.0001$:
\begin{equation}
  \sum_{\substack{\xi'\in\mathcal{L} \\ v \in  \xi'} } \gamma^{-\vert E(\xi')\vert} \exp(c\vert V(\xi')\vert) \leq c  \label{kotecky_preiss_simplified}
\end{equation}
when $\gamma > 29.3(q-1)$. 

Equation (\ref{kotecky_preiss_simplified}) implies the lemma:
\begin{align*}
 \sum_{\substack{\xi'\in\mathcal{L} \\ v \in  \xi'} } \gamma^{-\vert E(\xi')\vert} \exp(c\vert V(\xi')\vert) &\leq c \Rightarrow \sum_{v\,\in V(\xi)}\sum_{\substack{\xi'\in\mathcal{L} \\ v \in  \xi'} } \gamma^{-\vert E(\xi')\vert} \exp(c\vert V(\xi')\vert) \leq \sum_{v\,\in V(\xi)} c = c \vert V(\xi)\vert\\
 \Rightarrow \sum_{\substack{\xi'\in\mathcal{L} \\ \xi' \nsim \xi} } \gamma^{-\vert E(\xi')\vert} \exp(c\vert V(\xi')\vert) &\leq \sum_{v\,\in V(\xi)} \sum_{\substack{\xi'\in\mathcal{L} \\ v \in  \xi'} } \gamma^{-\vert E(\xi')\vert} \exp(c\vert V(\xi')\vert) \leq c \vert V(\xi)\vert
\end{align*}

The left-hand side of Equation~(\ref{kotecky_preiss_simplified}) is a sum over all directed edge-weighted connected subgraphs of $G_{\Delta}$ containing a given vertex. Let the number of polymers with $m$ edges containing a given vertex be $\nu(m,q)$. 

As the maximum degree of a vertex in $G_\Delta$ is $6$, the number of connected subgraphs of $G_\Delta$ with $m$ edges containing a given vertex is at most $(6e)^m/2$ \cite{borgs_left_2010}.
Each edge in the polymer takes one of $q-1$ possible values from $\{1,2,\ldots,q-1 \}$, which gives us $\nu(m,q) \leq (6e(q-1))^m/2$.
As $(V(\xi'), E(\xi'))$ is a connected graph, we have $\vert V(\xi')\vert \leq \vert E(\xi')\vert + 1$, so
the left-hand side of Equation~(\ref{kotecky_preiss_simplified}) becomes:

\begin{align*}
\sum_{\substack{\xi'\in\mathcal{L} \\ v \in  \xi'} } \gamma^{-\vert E(\xi')\vert} \exp(c\vert V(\xi')\vert) &\leq \sum_{m=1}^{\vert E(G_\Delta)\vert} \nu(m,q)\,\gamma^{-m}\,\exp(c(m+1)).
\end{align*}
As $\nu(m,q) \leq (6e(q-1))^m/2$, for sufficiently large values of $\gamma$, the above sum can be made arbitrarily small. 

For a tighter lower bound on $\gamma$,
we evaluate $\nu(m,q)$ for the lowest values of $m$. The smallest polymer consists of $6$ edges emanating from a single vertex, and has $\nu(6,q) = 7(q-1)$. Similarly calculating, we find $\nu(7,q)=\nu(8,q)=\nu(9,q)=0$, $\nu(10,q)=30(q-1)$, $\nu(11,q) = 30(q-1)(q-2)$, $\nu(12,q)=24(q-1)+28(q-1)^2$, $\nu(13,q)=0$, $\nu(14,q)=137(q-1)+72(q-1)(q-2)$, and $\nu(15,q)=24(q-1)(q-2)(q-3)+246(q-1)(q-2)$. Carefully evaluating $\nu(m,q)$ for higher values of $m$ could further lower the smallest value of $\gamma$ for which the Lemma holds. 

Substituting the above, we get:
\begin{align*}
&\sum_{m=1}^{\vert E(G_\Delta)\vert} \nu(m,q)\,\gamma^{-m}\,e^{c(m+1)} \\
&\leq e^c\Big[ 7(q-1)\left(\frac{e^c}{\gamma}\right)^6 + 30(q-1)\left(\frac{e^c}{\gamma}\right)^{10} + 
30(q-1)(q-2)\left(\frac{e^c}{\gamma}\right)^{11} + \\
&\left( 24(q-1) + 28(q-1)^2 \right)\left(\frac{e^c}{\gamma}\right)^{12} + \left(137(q-1) +
2(q-1)(q-2) \right)\left(\frac{e^c}{\gamma}\right)^{14} + \\
&\left(246(q-1)(q-2) + 24(q-1)(q-2)(q-3) \right) \left(\frac{e^c}{\gamma}\right)^{15} \Big] + \frac{e^c}{2}\sum_{m=16}^{\vert E(G_\Delta)\vert} \left(\frac{6(q-1)e^{1+c}}{\gamma} \right)^m 
\end{align*}
where the last term is a geometric series which has an upper bound:
\[ \frac{e^c}{2}\sum_{m=16}^{\vert E(G_\Delta)\vert} \left(\frac{6(q-1)e^{1+c}}{\gamma} \right)^m \leq \frac{e^c}{2}\frac{\left(\frac{6(q-1)e^{1+c}}{\gamma} \right)^{16}}{1-\left(\frac{6(q-1)e^{1+c}}{\gamma} \right)} \] 
Using $q\geq 2$, we numerically verify that the Lemma holds for $\gamma \geq 29.3(q-1)$.
\end{proof}
}

Lemma~\ref{lem:kotecky_preiss_potts} has an important consequence in addition to guaranteeing the convergence of the cluster expansion, as stated in the original paper of Koteck{\'y} and Preiss \cite{kotecky_cluster_1986}, and rephrased in \cite{jenssen_independent_2020}. Consider the function $\Psi(X)$ defined earlier for any cluster $X$.
An additional consequence \cite{kotecky_cluster_1986,jenssen_independent_2020} of Lemma~\ref{lem:kotecky_preiss_potts} is that $\Psi(X)$ will satisfy the following inequality
\begin{equation}
\sum_{\substack{X\in \mathcal{X}\\ X\nsim \xi}} \vert \Psi(X) \vert \leq c\vert V(\xi) \vert. \label{Psi_kotecky_preiss}
\end{equation}
for any polymer $\xi$, where $\mathcal{X}$ is the set of all clusters of polymers, and a cluster $X\nsim \xi$ if there exists a polymer $\xi'\in X$ such that $\xi'\nsim \xi$. The support of a cluster $X$ is denoted by $\bar{X}$ and is given by $\bar{X}=\bigcup_{\xi\in X} V(\xi)$. 

Consider an arbitrary vertex $v \in G_\Delta$, and let $\xi_v$ be the smallest polymer consisting of six edges of equal weight attached to $v$. From Equation~(\ref{Psi_kotecky_preiss}), we have:
\begin{equation}
\sum_{\substack{X\in \mathcal{X}\\ X\nsim \xi_v}} \vert \Psi(X) \vert \leq c\vert V(\xi_v) \vert = 7c \Rightarrow \sum_{\substack{X\in \mathcal{X}\\ v\,\in \bar{X}}} \vert \Psi(X) \vert \leq \sum_{\substack{X\in \mathcal{X} \\ X\nsim \xi_v}} \vert \Psi(X) \vert \leq 7c. \label{required_Psi_constraint}
\end{equation}

\begin{lemma}
\label{lem:volume_surface_potts}
If for any polymer $\xi \in \mathcal{L}$, there exists a constant $c$ such that 
\[ \sum_{\substack{\xi'\in\mathcal{L} \\ \xi' \nsim \xi} } w(\xi') \exp(c\vert V(\xi')\vert) \leq c \vert V(\xi)\vert ,\] 
then for any connected region $\Lambda_\mathcal{P}$ with boundary $\mathcal{P}$, the partition function $\Xi_\mathcal{P}$ satisfies
\[ \psi\vert \Lambda_\mathcal{P}\vert - 7c\vert \partial\Lambda \vert \leq \ln \Xi_\mathcal{P} \leq  \psi\vert \Lambda_\mathcal{P}\vert + 7c\vert \partial\Lambda \vert .\]
\end{lemma}

\begin{proof}
Denote the set of all clusters $X$ whose support $\bar{X} \subseteq \Lambda_\mathcal{P}$ by  $\mathcal{X}_{\Lambda_\mathcal{P}}$, and the set of all clusters $X$ by $\mathcal{X}$. For any cluster $X \in \mathcal{X}_{\Lambda_\mathcal{P}}$, $1=\left(\sum_{v \in \Lambda_\mathcal{P}} \mathbbm{1}_{v\,\in\bar{X}}\right)/\vert \bar{X}\vert$.
Using the cluster expansion, the log partition function can be expressed as:
\begin{align*}
    \ln \Xi_\mathcal{P} &= \sum_{\substack{X\in\mathcal{X}_{\Lambda_\mathcal{P}}}} \Psi(X) = \sum_{\substack{X\in \mathcal{X}:\\ \bar{X}\subseteq \Lambda_\mathcal{P}}}\Psi(X) = \sum_{\substack{v\,\in \Lambda_\mathcal{P}}} \sum_{\substack{X\in\mathcal{X}:\\ v\,\in \bar{X},\\ \bar{X}\subseteq \Lambda_\mathcal{P}}} \frac{1}{\abs{\bar{X}}} \Psi(X) \\
    &= 
    \sum_{\substack{v\,\in \Lambda_\mathcal{P}}} \left(  \sum_{\substack{X\in\mathcal{X}:\\ v\,\in \bar{X}}} \frac{1}{\abs{\bar{X}}} \Psi(X) - \sum_{\substack{X\in\mathcal{X}:\\ v\,\in \bar{X},\\ \bar{X}  \not\subseteq \Lambda_\mathcal{P} }} \frac{1}{\abs{\bar{X}}} \Psi(X) \right) 
    \nonumber \\
    &= \left(\sum_{\substack{v\,\in \Lambda_\mathcal{P}}} \sum_{\substack{X\in\mathcal{X}:\\ v\,\in \bar{X}}} \frac{1}{\abs{\bar{X}}} \Psi(X)  \right) - \left( \sum_{\substack{v\,\in \Lambda_\mathcal{P}}} \sum_{\substack{X\in\mathcal{X}:\\ v\,\in \bar{X},\\ \bar{X}  \not\subseteq \Lambda_\mathcal{P} }} \frac{1}{\abs{\bar{X}}} \Psi(X)  \right) \label{vol_surface_log_partition_fn}
\end{align*}

We note that the inner sum in the first term:
\[ \psi := \sum_{\substack{X\in\mathcal{X}:\\ v\,\in \bar{X}}} \frac{1}{\abs{\bar{X}}} \Psi(X)   \]
is independent of $v$ and $\Lambda_\mathcal{P}$, and depends only on our polymer model through its dependence on $\Psi(X)$. Moreover, by Equation (\ref{required_Psi_constraint}),  $\vert\psi\vert \leq 7c$. Hence, the first term is $\psi\vert \Lambda_{\mathcal{P}}\vert$. 

Analyzing the second term, we note that if $v\,\in \bar{X}$ for some vertex  $v\,\in\Lambda_\mathcal{P}$ and $\bar{X} \not\subseteq \Lambda_\mathcal{P}$, then $\bar{X}$ must contain some vertex $v' \in \partial \Lambda_{\mathcal{P}}$. Using this intuition, Equation (\ref{required_Psi_constraint}), and the triangle inequality, the absolute value of the second term becomes
\begin{align*}
\left\vert \sum_{\substack{v\,\in \Lambda_\mathcal{P}}} \sum_{\substack{X\in\mathcal{X}:\\ v\,\in \bar{X},\\ \bar{X}  \not\subseteq \Lambda_\mathcal{P} }} \frac{1}{\vert \bar{X}\vert} \Psi(X) \right\vert  &\leq    \sum_{\substack{v\,\in \Lambda_\mathcal{P}}} \sum_{\substack{X\in\mathcal{X}:\\ v\,\in \bar{X},\\ \bar{X}  \not\subseteq \Lambda_\mathcal{P} }} \frac{1}{\vert \bar{X}\vert} \left\vert \Psi(X) \right\vert \nonumber \\
&\leq \sum_{\substack{v'\in\partial\Lambda_{\mathcal{P}}}} \sum_{\substack{X\in\mathcal{X}:\\v'\,\in\bar{X}}} \frac{\vert \bar{X}\bigcap\Lambda_\mathcal{P} \vert}{\vert \bar{X}\vert}\, \left\vert \Psi(X) \right\vert \nonumber \\  
& \leq \sum_{\substack{v'\in\partial\Lambda_{\mathcal{P}}}} \sum_{\substack{X\in\mathcal{X}:\\v'\,\in\bar{X}}} \left\vert \Psi(X) \right\vert \leq 7c\vert \partial\Lambda\vert \nonumber
\end{align*}
The proof of the Lemma then follows from substitution of the above results and the triangle inequality.
\end{proof}

The proof follows on the lines of the proof of a similar Lemma in \cite{cannon_local_2019}, and section 5.7.1 of~\cite{friedli_cluster_2017}. 
Using Lemma \ref{lem:volume_surface_potts}, and noting that $\vert \partial \Lambda_{\mathcal{P}}\vert \leq p(\sigma) \, \forall \sigma \, \in \Omega_\mathcal{P}$ and $\vert \Lambda_\mathcal{P} \vert = n$, we get: 
 \begin{equation}
     n\psi - 7c\,p(\sigma) \leq \ln \Xi_\mathcal{P} \leq n\psi + 7c\,p(\sigma) \label{vol_surface_result_potts}
 \end{equation}

Note that the partition function $Z_\mathrm{Potts}$ is greater than the contribution from particle configurations in $\Omega_\mathcal{P}^0$ where the length of the boundary is the smallest attainable perimeter $\vert \mathcal{P} \vert = p_\mathrm{min}$:
\begin{equation}
    Z_\mathrm{Potts} \geq w(\Omegaref) = \left( \lambda\,\gamma\right)^{-p_\mathrm{min}} \Xi_{\mathcal{P}} \geq \left( \lambda\,\gamma\right)^{-p_\mathrm{min}} e^{n\psi - 7c p_\mathrm{min}}. \label{Z_potts_lower_bound}
\end{equation}

Given $\alpha >1$, let $S_\alpha$ be all configurations that are not $\alpha$-compressed. We will prove that the probability of the set $S_\alpha$ in the stationary distribution is exponentially small for sufficiently large $\lambda, \gamma$:
\begin{lemma}
\label{lem:compression_potts}
Given any $\alpha >1$, when constants $\lambda>1, c=0.0001$, and $\gamma>29.3\,(q-1)$ satisfy
\begin{equation}  
\lambda\,\gamma > (4+2\sqrt{2}))^{\frac{\alpha}{\alpha-1}}
\left( e^{7c}\right)^{\frac{\alpha+1}{\alpha-1}} \label{compression_condn_potts}
\end{equation}
and $n$ is sufficiently large, then the probability that a configuration drawn from the stationary distribution
$\pi_\mathrm{Potts}$ is not $\alpha$-compressed is exponentially small,
$ \pi_\mathrm{Potts}(S_\alpha) < \zeta^{\sqrt{n}}. $ 
\end{lemma}
\noindent Note that Equation (\ref{compression_condn_potts}) is satisfied if $\lambda\,\gamma > 7^{\alpha/(\alpha-1)}$, proving Theorem~\ref{thm:con_comp_thm}.
The proof of the Lemma requires using Lemma \ref{lem:monochromatic}, Lemma \ref{lem:volume_surface_potts} and Equation~(\ref{Z_potts_lower_bound}), and an upper bound on the number of self-avoiding walks of a given length on the triangular lattice from \cite{duminil-copin_connective_2012,cannon_markov_2016}. 

\markred{
\begin{proof}
Consider a real number $\nu$ which satisfies $\nu > 2+\sqrt{2}$ and $\lambda\,\gamma > (q\,\nu)^\frac{\alpha}{\alpha-1}
\left(e^{7c}\right)^\frac{\alpha+1}{\alpha-1}$. Such a $\nu$ exists since Equation~(\ref{compression_condn_potts}) holds and is a strict inequality. 
The probability $\pi(S_\alpha)$ can be calculated as, 
\[ \pi(S_\alpha) = \frac{w(S_\alpha)}{Z_\mathrm{Potts}} = \frac{\sum_{k=\ceil*{\alpha\,p_{\text{min}}}}^{p_{\text{max}}} \sum_{\mathcal{P}:\vert\mathcal{P}\vert = k} w(\Omega_{\mathcal{P}})}{Z_\mathrm{Potts}}
< \frac{q\gamma}{\gamma - 3q}\frac{\sum_{k=\ceil*{\alpha\,p_{\text{min}}}}^{p_{\text{max}}} \sum_{\mathcal{P}:\vert\mathcal{P}\vert = k} 2^k\,w(\Omega_{\mathcal{P}}^0)}{Z_\mathrm{Potts}}\]
where we have used Lemma \ref{lem:monochromatic} in the last inequality.

Using Lemma \ref{lem:volume_surface_potts} and Equation~(\ref{Z_potts_lower_bound}), we get
\begin{align*}
    \pi(S_\alpha)
    &< \frac{q\gamma}{\gamma - 3q} \frac{\sum_{k=\ceil*{\alpha\,p_{\text{min}}}}^{p_{\text{max}}} (\lambda\,\gamma)^{-k} 2^k \sum_{\mathcal{P}:\vert\mathcal{P}\vert = k} \Xi_\mathcal{P}}{(\lambda\,\gamma)^{-p_\mathrm{min}} e^{n\psi - 7c p_\mathrm{min}}} \\
    &\leq \frac{q\gamma}{\gamma - 3q}\frac{\sum_{k=\ceil*{\alpha\,p_{\text{min}}}}^{p_{\text{max}}} \left(\frac{\lambda\,\gamma}{2}\right)^{-k} \sum_{\mathcal{P}:\vert\mathcal{P}\vert = k} e^{n\psi + 7ck}}{(\lambda\,\gamma)^{-p_\mathrm{min}} e^{n\psi - 7c p_\mathrm{min}}} \\
   &< \frac{q\gamma}{\gamma - 3q}\sum_{k=\ceil*{\alpha\,p_{\text{min}}}}^{p_{\text{max}}} \left(\frac{\lambda\,\gamma}{2}\right)^{-k} \nu^k\, e^{7ck}\,(\lambda\,\gamma)^{p_\mathrm{min}} \,e^{7c p_\mathrm{min}}
\end{align*}
in the last inequality, we use an upper bound from \cite{duminil-copin_connective_2012,cannon_markov_2016}, that the number of self-avoiding walks on the triangular lattice of length $k$ is no more than $(2 + \sqrt{2})^k$, which is at most $\nu^k$.

Since $k\geq \alpha\,p_\mathrm{min}$, and $\alpha > 1$, $p_\mathrm{min}\leq k/\alpha$. Substituting, we get:

\begin{align*}
    \pi(S_\alpha) &\leq  \frac{q\gamma}{\gamma - 3q} \sum_{k=\ceil*{\alpha\,p_{\text{min}}}}^{p_{\text{max}}} \left(\frac{\lambda\,\gamma}{2}\right)^{-k} \nu^k\, e^{7ck}\,(\lambda\,\gamma)^{k/\alpha} \,e^{7c k/\alpha} \\
    &= \frac{q\gamma}{\gamma - 3q}\sum_{k=\ceil*{\alpha\,p_{\text{min}}}}^{p_{\text{max}}} \left( \frac{2\,\nu\,e^{7c(1+1/\alpha)}}{(\lambda\,\gamma)^{1-1/\alpha}} \right)^k
\end{align*}

Our choice of $\nu$ ensures that $\left(\nu\,e^{7c(1+1/\alpha)}\right)/(\lambda\,\gamma)^{1-1/\alpha}$ is less than $1$. Since $\alpha\,p_\mathrm{min} = O(\sqrt{n})$, for sufficiently large $n$, there exists a constant $\zeta < 1$ such that $\pi(S_\alpha)<\zeta^{\sqrt{n}}$. This proves the theorem.

Lastly, since $\alpha > 1$ the right-hand side of  Eq.~(\ref{compression_condn_potts}) satisfies:
\begin{align*}
(4+2\sqrt{2}))^{\frac{\alpha}{\alpha-1}}
\left( e^{7c}\right)^{\frac{\alpha+1}{\alpha-1}} & < (4+2\sqrt{2}))^{\frac{\alpha}{\alpha-1}}
\left( e^{7c}\right)^{\frac{2\alpha}{\alpha-1}}
= \left((4+2\sqrt{2})e^{14c}\right)^{\frac{\alpha}{\alpha-1}} < 7^{\frac{\alpha}{\alpha-1}}
\end{align*}

\end{proof}
}

\subparagraph*{The Clock Model:}
\label{section:clockmodelanalysis}
The proof of compression for the clock-model-inspired algorithm follows along the same lines as the  proof for the Potts-model-inspired algorithm. The set of allowed particle configurations is the same as before, so the set of configurations in $\Omega_\mathcal{P}^0$ is in a one-to-one correspondence with compatible collections of polymers with the same polymer model as above, albeit with the weight of a polymers redefined as $w_\mathrm{clock}(\xi)=\prod_{e\in\xi}\gamma^{-d_e}$, where $d_e=1-\cos(2\pi \ell(e) /q)$, and $\ell(e) \in\{1,2,\ldots,q-1\}$ is the label associated with an edge $e \in \xi$. This changes the prefactor in Lemma \ref{lem:monochromatic}, replacing $\gamma$ with $\gamma^{-\cos(2\pi/q)}$, and requiring $\gamma^{-\cos(2\pi/q)} > 3q$. The polymer partition function becomes 
$$ \Xi_\mathcal{P} = \sum_{\substack{\mathcal{L}'\subseteq \mathcal{L}_\mathcal{P}\\\mathrm{compatible}}} \prod_{\xi\in\mathcal{L}'} w_\mathrm{clock}(\xi).$$

\noindent Since the maximum weight of an edge in a polymer is now $\gamma^{-(1-\cos(2\pi/q))}$, instead of $\gamma^{-1}$, the condition for Lemma \ref{lem:kotecky_preiss_potts} to hold becomes $\gamma^{1-\cos(2\pi/q)}>29.3(q-1)$. Lemmas \ref{lem:volume_surface_potts} and Theorem \ref{thm:con_comp_thm} follow without modification except for the modified condition: $\gamma^{1-\cos(2\pi/q)}>29.3(q-1)$ in Theorem \ref{thm:con_comp_thm}.

\subsection{ Alignment in Compressed Configurations}
\label{sec:alignment}


\begin{definition}[Alignment]
\label{defn:alignment}
We say a configuration of $n$ particles with $q$ orientations is $\delta$-aligned if there exists an orientation $\theta \in \{0,1,\ldots,q-1\}$, such that the number of particles of orientation $\theta$ is at least $(1-\delta)n$.
\end{definition}

Our main result is the following theorem:

\begin{theorem}\label{thm:con_align_thm}
Denote by $\pi_{\mathrm{Potts},\mathcal{P}}$ the stationary distribution $\pi_\mathrm{Potts}$ conditioned on the boundary of the configuration being $\mathcal{P}$.
For any $\eta$ where $1/2 < \eta < 1$, there exists a constant $\alpha^\ast = \alpha^\ast(\eta,q) > 1$, such that for all $\alpha$ where $1 < \alpha < \alpha^\ast$, there exists a sufficiently large $\gamma^\ast = \gamma^\ast(\eta,q,\alpha,\alpha^\ast)$ where as long as $\gamma > \gamma^\ast$ and $\mathcal{P}$ is $\alpha$-compressed, the probability that a configuration drawn from $\pi_{\mathrm{Potts},\mathcal{P}}$
is not $(1-\eta)$-aligned is exponentially small.

In particular, possible values of $\alpha^*$ and $\gamma^*$ are:
\begin{align*}
\alpha^\ast(\eta, q) &= \min\left\{\sqrt{\eta}+\sqrt{1-\eta}, \sqrt{q^{-1}}+\sqrt{1-q^{-1}} \right\} \\
\gamma^\ast(\eta,q,\alpha,\alpha^\ast) &= \Big( 3^{\frac{2\alpha}{\alpha^\ast - \alpha}} \cdot 4^{\frac{3}{4}+\frac{\alpha^\ast - 1}{2\delta^\ast(\eta, q)(\alpha^\ast - \alpha)}} \Big)^{q-1} \text{where $\delta^\ast(\eta, q) := \min\{ 1-\eta, q^{-1} \}$.}
\end{align*}
\end{theorem}


For any particle configuration, let $2\pi\theta_p/q,$ be the {\it most popular orientation}, or the orientation possessed by the greatest number of particles, where $\theta_p \in \{0,1,\ldots,(q-1) \},$ and let $\rho_p$ be the fraction of particles with orientation $\theta_p$. Note that $1/q\leq \rho_p \leq 1$, and $\rho_p \geq \eta$ for a $(1-\eta)$-aligned configuration.
%
%
We begin with a few isoperimetric inequalities.

\begin{lemma}(\cite{angel_isoperimetric_2018, oh_sharp_2020})
\label{lem:isoperimetric_triangular_lower}
The minimum perimeter of a region containing $n$ particles on the triangular lattice $G_\Delta$ is greater than or equal to $2\sqrt{3}\sqrt{n - (1/4)}-3$.
\end{lemma}

\begin{proof}
The perimeter of a region containing $n$ particles in $G_\Delta$ which is not connected, or which is connected but has holes is greater than the perimeter of a connected hole-free region containing $n$ particles in $G_\Delta$. 
From \cite{angel_isoperimetric_2018, oh_sharp_2020}, we have for the perimeter of a connected hole-free region containing $n$ particles in $G_\Delta$, the following isoperimetric inequality: $n\leq \floor{\frac{(p_\mathrm{min}+3)^2+3}{12}}$. This gives:
\begin{align*}
n &\leq \floor{\frac{(p_\mathrm{min}+3)^2+3}{12}} \leq \frac{(p_\mathrm{min}+3)^2+3}{12} \Rightarrow 2\sqrt{3}\sqrt{n-(1/4)}\leq p_\mathrm{min}+3
\end{align*}
which proves the Lemma. 
\end{proof}

\begin{lemma}(\cite{cannon_local_2019}, Lemma 2.2)
\label{lem:isoperimetric_triangular_upper}
The minimum perimeter of a region containing $n$ particles on the triangular lattice $G_\Delta$ is at most $2\sqrt{3}\sqrt{n}$.
\end{lemma}

The dual lattice, $G_{\varhexagon}$, to the triangular lattice $G_\Delta$ is obtained by creating a dual vertex in the center of each triangle in $G_\Delta$, and joining these dual vertices with edges if their corresponding triangular faces share an edge. Each edge~$e_\Delta$ of $G_\Delta$ corresponds with the edge $e_{\varhexagon}$ of $G_{\varhexagon}$ that crosses it.
This corresponding edge~$e_{\varhexagon}$
separates the two endpoints of~$e_\Delta$ in~$G_\Delta$.
A \emph{contour} refers to a self-avoiding walk on the edges of the dual lattice $G_{\varhexagon}$. The length of a contour refers to the number of edges in the contour.

In this setting, we distinguish between the \emph{boundary contour} and the \emph{internal boundary contour} of a region $R \subseteq V(\Lambda_\mathcal{P})$. The boundary contour refers to the set of edges on the dual lattice $G_{\varhexagon}$ corresponding to edges between sites in $R$ and sites not in $R$, while the internal boundary contour includes edges only from $E(\Lambda_\mathcal{P})$ rather than all of $E(G_{\varhexagon})$. We make use of the following geometric result, which we show in the full version of the paper:

\begin{lemma}
\label{lem:min_cut_length}
For a connected hole-free $\alpha$-compressed configuration 
with $n$ particles,
a particle-occupied region $R$ containing $\kappa n$
particles has an internal boundary contour $bd_\mathrm{int}(R)$ of length at least $\nu\sqrt{n}(\sqrt{\kappa} + \sqrt{1-\kappa} - \alpha)$ for any $\nu < 2\sqrt{3}$ for sufficiently large $n$.
\end{lemma}

\markred{
\begin{proof}
Applying the result of Lemma \ref{lem:isoperimetric_triangular_lower} to $R$ and $\bar{R}$, and denoting their perimeters by $p(R)$ and $p(\bar{R})$,  
\begin{align*}
p(R) &\geq 2\sqrt{3}\sqrt{\kappa n - (1/4)}-3 \;,\; p(\bar{R}) \geq 2\sqrt{3}\sqrt{(1-\kappa) n - (1/4)} -3 
\end{align*}
Since $\sigma$ is $\alpha$-compressed, its perimeter $p(\sigma) \leq \alpha p_\mathrm{min}(\sigma)$. Using Lemma \ref{lem:isoperimetric_triangular_upper}, we get:
\begin{equation*}
    -p(\sigma) \geq - 2\sqrt{3}\,\alpha \sqrt{n} 
\end{equation*}
Summing the above Equations, we get:
\begin{align*}
p(R) + p(\bar{R}) - p(\sigma) &\geq  2\sqrt{3}\sqrt{n} \left( \sqrt{\kappa - (4n)^{-1}} + \sqrt{(1-\kappa) - (4n)^{-1}} - \alpha \right) -6 \nonumber \\
\Rightarrow p(R) + p(\bar{R}) - p(\sigma) &\geq  \nu\sqrt{n}\left( \sqrt{\kappa} + \sqrt{(1-\kappa)}  - \alpha \right) \;,\; \nu < 2\sqrt{3}\,,\, n\;\mathrm{sufficiently}\,\mathrm{large}
\end{align*}
Note that $R\bigcup \bar{R} = \sigma$, and that the length of the internal boundary contour of $R$ is the same as the length of the internal boundary contour of $\bar{R}$: $\vert bd_\mathrm{int}(R)\vert = \vert bd_\mathrm{int}(\bar{R})\vert$. Furthermore, note that the lengths of the boundary contours of $R, \bar{R}$ and $\sigma$, denoted by $\vert bd(R)\vert, \vert bd(\bar{R})\vert$ and $\vert bd(\sigma)\vert$ satisfy:
\begin{align*}
\vert bd(R) \vert + \vert bd(\bar{R})\vert &= \vert bd(\sigma)\vert + 2\,\vert bd_\mathrm{int}(R)\vert \\ \vert bd(\sigma)\vert &= 2p(\sigma) + 6 \\ \vert bd(R) \vert + \vert bd(\bar{R})\vert &> 2(p(R) + p(\bar{R})) 
\end{align*}
where the last inequality comes from noting that $\vert bd(R) \vert + \vert bd(\bar{R})\vert=2(p(R)+p(\bar{R}))+6$ when $bd_\mathrm{int}(R)$ consists of no closed contours, and $\vert bd(R) \vert + \vert bd(\bar{R})\vert=2(p(R)+p(\bar{R}))$ when $bd_\mathrm{int}(R)$ consists only of closed contours.

Using the above relations, we get 
\begin{align*}
\vert bd(R) \vert + \vert bd(\bar{R})\vert - \vert bd(\sigma)\vert &\geq 2(p(R) + p(\bar{R}) - p(\sigma)) - 6 \nonumber \\
\Rightarrow \vert bd_\mathrm{int}(R) \vert &\geq (p(R) + p(\bar{R}) - p(\sigma)) - 3 \nonumber \\
\Rightarrow \vert bd_\mathrm{int}(R) \vert &\geq \nu\sqrt{n}\left( \sqrt{\kappa} + \sqrt{(1-\kappa)}  - \alpha \right) \;,\; \nu < 2\sqrt{3}\,,\, n\;\mathrm{sufficiently}\,\mathrm{large} \nonumber
\end{align*}
which proves the Lemma.
\end{proof}
}

\noindent
For the rest of this section, we assign particles the color $c_1$ if they are of orientation $\theta_p$, and the color $c_2$ otherwise.
This lets us directly apply the bridging construction from \cite{cannon_local_2019}.

\begin{lemma}(\cite{cannon_local_2019}, Lemma 7.3)
\label{lem:bridging_sigma}
Fix $\delta \in (0,1/2)$. For each particle configuration $\sigma \in \Omega_\mathcal{P}$,
there exists a function $\mathcal{R}_\delta:\Omega_\mathcal{P}\to 2^{\Omega_\mathcal{P}}$ giving a region $\mathcal{R}_\delta(\sigma)$ such that
all particles on the boundary of $\mathcal{R}_\delta(\sigma)$ have the color $c_1$,
all particles on the boundary of its complement $\bar{\mathcal{R}}_\delta(\sigma)$
have the color $c_2$,
$\mathcal{R}_\delta(\sigma)$ contains at most $\delta$ fraction of particles with the color $c_2$, and
$\bar{\mathcal{R}}_\delta(\sigma)$ contains at most $\delta$ fraction of particles with the color $c_1$.
\end{lemma}

\noindent
We use the bridging construction from \cite{cannon_local_2019} to define the region $\mathcal{R}_\delta(\sigma)$ in Lemma \ref{lem:bridging_sigma}. 

\begin{lemma}
\label{lem:partition_sigma_R}
For any particle configuration $\sigma \in \Omega_\mathcal{P}$ with total number of particles $n$ and $\rho_p$ fraction of particles of color $c_1$, given any $\delta > 0$, the region $\mathcal{R}_\delta(\sigma)$ defined in Lemma \ref{lem:bridging_sigma} is such that the number of particles in $\mathcal{R}_\delta(\sigma)$, $n_{\mathcal{R}_\delta}$ satisfies: $(\rho_p - \delta)n/(1-\delta)\leq n_{\mathcal{R}_\delta} \leq (\rho_p n)/(1-\delta)$. 
\end{lemma}
\noindent The proof of Lemma \ref{lem:partition_sigma_R} follows from noting that the particles in $\mathcal{R}_\delta(\sigma)$ and $\bar{\mathcal{R}}_\delta(\sigma)$ are predominantly of the colors $c_1$ and $c_2$ respectively, with an error fraction bounded by $\delta$, and enforcing that the total number of particles with the color $c_1$ is $\rho_p \,n$.

\markred{
\begin{proof}
The total number of particles with the label \emph{aligned} is $\rho_p n$ since the fraction of particles with orientation $\theta_p$ is $\rho_p n$. Moreover, the number of particles of color $c_1$ in $\mathcal{R}_\delta(\sigma)$, $n_{\mathcal{R}_\delta}^{1}$, and the number of particles of color $c_1$ in $\bar{\mathcal{R}}_\delta(\sigma)$, $n_{\bar{\mathcal{R}}_\delta}^{1}$, obey: 
\begin{equation*}
    \min\{ n_{\mathcal{R}_\delta}, \rho_p n\} \geq n_{\mathcal{R}_\delta}^{1} \geq (1-\delta)n_{\mathcal{R}_\delta}\,,\, n_{\bar{\mathcal{R}}_\delta}^{1} \leq \delta (n-n_{\mathcal{R}_\delta}) \Rightarrow n_{\mathcal{R}_\delta} \geq n_{\mathcal{R}_\delta}^{1} \geq \rho_p n - \delta (n-n_{\mathcal{R}_\delta}) 
\end{equation*}
where the first set of inequalities comes from noting that $n_{\mathcal{R}_\delta}^{1}$ must be smaller than the total number of particles with the color $c_1$, $\rho_p n$, as well as the total number of particles in $\mathcal{R}_\delta(\sigma)$, $n_{\mathcal{R}_\delta}$, and must be greater than or equal to the number of bridged particles in $\mathcal{R}_\delta(\sigma)$ which is greater than or equal to $(1-\delta)n_{\mathcal{R}_\delta}$. The second inequality comes from noting that the number of particles with color $c_1$ in $\bar{\mathcal{R}}_\delta(\sigma)$ is the number of unbridged particles in $\bar{\mathcal{R}}_\delta(\sigma)$, which can be at most a $\delta$ fraction of the particles in $\bar{\mathcal{R}}_\delta(\sigma)$.
Using the above inequalities, we get for $n_{\mathcal{R}_\delta}$:
\begin{align*}
& n_{\mathcal{R}_\delta} \geq \rho_p n - \delta (n-n_{\mathcal{R}_\delta}) \Rightarrow n_{\mathcal{R}_\delta} \geq \frac{\rho_p - \delta}{1-\delta}n  \\
& \rho_p n \geq n_{\mathcal{R}_\delta} (1-\delta) \Rightarrow n_{\mathcal{R}_\delta} \leq \frac{\rho_p}{1-\delta}n
\end{align*}
which proves the lemma.
\end{proof}
}

\begin{lemma}
\label{lem:critical_alpha}
For a connected hole-free $\alpha$-compressed configuration $\sigma \in \Omega_\mathcal{P}$ that is not $(1-\eta)$-aligned for some $\eta<1$, given any $\delta$ where $0<\delta < \min\{q^{-1},1-\eta\}$, 
the internal boundary contour length $\vert bd_\mathrm{int}(\mathcal{R}_\delta)\vert$ of the region $\mathcal{R}_\delta(\sigma)$ defined in Lemma \ref{lem:bridging_sigma} obeys the lower bound
$\vert bd_\mathrm{int}(\mathcal{R}_\delta) \vert \geq \nu\sqrt{n} (\alpha_c(\delta, \eta, q) - \alpha)$
for any $\nu < 2\sqrt{3}$ and $n$ sufficiently large, where
\begin{align*}
\alpha_c(\delta, \eta, q) := \min\left\{ \sqrt{\frac{q^{-1}-\delta}{1-\delta}} + \sqrt{\frac{1-q^{-1}}{1-\delta}}\,,\, \sqrt{\frac{\eta}{1-\delta}}+\sqrt{\frac{1-(\eta + \delta)}{1-\delta}} \right\}.
\end{align*}
\end{lemma}

Lemma \ref{lem:critical_alpha} is a direct consequence of Lemmas \ref{lem:partition_sigma_R} and \ref{lem:min_cut_length}.
Given an $\alpha$-compressed boundary~$\mathcal{P}$, let $S_\mathcal{P}^\eta \subseteq \Omega_\mathcal{P}$ be the set of $\alpha$-compressed configurations with boundary $\mathcal{P}$ that are not $(1-\eta)$-aligned for some $\eta<1$. For each configuration $\sigma \in S_\mathcal{P}^\eta$, let $\bar{\mathcal{R}}_\delta(\sigma)$ be the complement of the region $\mathcal{R}_\delta(\sigma)$ defined in Lemma \ref{lem:bridging_sigma}. 

\markred{
\begin{proof}
Since $\sigma$ is not $(1-\eta)$-aligned, the fraction $\rho_p$ of particles with the color $c_1$ obeys $\rho_p \geq q^{-1}\,,\, \rho_p \leq \eta$. By Lemma \ref{lem:partition_sigma_R}, the fraction of the number of particles in the region $\mathcal{R}_\delta$, $n_{\mathcal{R}_\delta}/n$ obeys: 
\[ \frac{q^{-1}-\delta}{1-\delta}  \leq \frac{n_{\mathcal{R}_\delta}}{n} \leq \frac{\eta}{1-\delta} \,.\]
By Lemma \ref{lem:min_cut_length}, $bd_\mathrm{int}(R)$ for a region $R$ which contains a fraction $\kappa$ of the total number of particles, $n$, is at least $\nu \sqrt{n}(\sqrt{\kappa} + \sqrt{1-\kappa} - \alpha)$. We note that the expression $\sqrt{\kappa} + \sqrt{1-\kappa}$, is largest when $\kappa = 1/2$, and monotonically decreases with increasing $\vert \kappa - 1/2\vert$. Hence, $\vert bd_\mathrm{int}(\mathcal{R}_\delta) \vert$ attains its lowest value when evaluated for the maximum and minimum values of $n_{\mathcal{R}_\delta}/n$, giving: 
\[ \frac{\vert bd_\mathrm{int}(\mathcal{R}_\delta) \vert}{\nu\sqrt{n}} \geq  \alpha_c(\delta, \eta, q) -\alpha  \]
which proves the Lemma.
\end{proof}
}

Let $\mathcal{P}^\mathrm{int}_{\bar{\mathcal{R}}_\delta}$ denote the walk on the edges of $G_\Delta$, each of whose endpoints is a particle in $\bar{\mathcal{R}}_\delta(\sigma)$ that is connected by an edge in $G_\Delta$ to a particle in $\mathcal{R}_\delta(\sigma)$. Let $\Theta^\mathrm{int}_{\bar{\mathcal{R}}_\delta}$ denote the set of orientations of particles that are incident to an edge in $\mathcal{P}^\mathrm{int}_{\bar{\mathcal{R}}_\delta}$, where the orientation of a particle appears as many times as the number of edges connecting that particle to a particle in $\mathcal{R}_\delta(\sigma)$. Note that $\vert\Theta^\mathrm{int}_{\bar{\mathcal{R}}_\delta}\vert = \vert bd_\mathrm{int}(\bar{\mathcal{R}}_\delta) \vert$. Let the orientation which appears the most number of times in the set $\Theta^\mathrm{int}_{\bar{\mathcal{R}}_\delta}$ be $2\pi\bar{\theta}_p/q$, where $\bar{\theta}_p\in\{0,1,\ldots,q-1 \}$. We consider a map $f_\eta:S_{\mathcal{P}}^\eta\to\Omega_{\mathcal{P}}$ which applies a cyclic shift to the orientations of all particles in $\bar{\mathcal{R}}_\delta(\sigma)$, so that under $f_\eta$, a particle orientation $\theta$ is mapped to $(\theta+(\theta_p-\bar{\theta}_p))\Mod{q}$. Note that this transformation maps the orientation $\bar{\theta}_p$ to   $\theta_p$.

\begin{lemma}\cite{cannon_local_2019}
\label{lem:num_preimages_f_eta}
For a configuration $\tau\in \mathrm{Im}(f_\eta(S_{\mathcal{P}}^\eta))$, the number of preimages $\sigma \in S_{\mathcal{P}}^\eta$ for which $\vert bd_\mathrm{int}(\mathcal{R}_\delta(\sigma))\vert =\ell$, where $\mathcal{R}_\delta(\sigma)$ is defined in Lemma \ref{lem:partition_sigma_R}, is at most  $q\,3^{\vert\mathcal{P}\vert} 4^{\frac{1+3\delta}{4\delta}\ell}$. 
\end{lemma}
\noindent
The proof follows from Lemma~$7.6$ in \cite{cannon_local_2019} and by noting that once the internal boundary contour of $\mathcal{R}_\delta(\sigma)$ is known, one of $q$ cyclic shifts in $\bar{\mathcal{R}}_\delta(\sigma)$ recovers $\sigma$, given $\tau$.

In this section so far, our results were valid for both the Potts and the clock models. We now consider specifically the case of the Potts model with stationary distribution $\pi_\mathrm{Potts}$.  Using the definition of
$f_\eta$, we find the following.

\begin{lemma}
\label{lem:min_gain_f_eta}
For a configuration $\sigma \in S_{\mathcal{P}}^\eta$, let region $\mathcal{R}_\delta(\sigma)$ be defined as in Lemma \ref{lem:bridging_sigma} with $\vert bd_\mathrm{int}\vert = \ell$. For the new configuration $f_\eta(\sigma)$ under the map $f_\eta$, the ratio $w(\sigma)/w(f(\sigma))$ is at most $(1/\gamma)^{\ell/(q-1)}$. 
\end{lemma}

\markred{
\begin{proof}
By the definition of $f_\eta$, all particles on the boundary of $\bar{\mathcal{R}}_\delta(\sigma)$ with orientation $\bar{\theta}_p$ are mapped to particles with orientation $\theta_p$. Note that $\ell$ is also the number of edges connecting particles in $\bar{\mathcal{R}}_\delta(\sigma)$ to particles in $\mathcal{R}_\delta(\sigma)$. Since particles with orientation $\bar{\theta}_p$ contribute at least $1/(q-1)$ fraction of such edges, and are removed by the transformation $f_\eta$, the number of heterogeneous edges decreases by at least $\ell/(q-1)$.
\end{proof}
}


%
The proof of Theorem \ref{thm:con_align_thm} follows from an information theoretic argument similar to that in \cite{cannon_local_2019}, by showing that the minimum gain in the weight of a configuration under the map $f_\eta$ outweighs the maximum number of preimages of the map, and using Lemma \ref{lem:critical_alpha} to get a lower bound on the gain under $f_\eta$. A key component is ensuring that it is possible to choose the parameter $0<\delta<q^{-1}$, so that the conditions on $\alpha$ and $\gamma$ described in the theorem statement can be simultaneously satisfied.

\markred{
\begin{proof}[Proof of Theorem \ref{thm:con_align_thm}]
Let $S_{\mathcal{P}}^{\eta}$ denote the set of all configurations with the $\alpha$-compressed boundary $\mathcal{P}$, that are not $(1-\eta)$-aligned. Then
\[ \pi_\mathcal{P}(S_{\mathcal{P}}^{\eta}) = \sum_{\sigma \in S_{\mathcal{P}}^{\eta}} \pi_\mathcal{P}(\sigma) \leq \sum_{\tau \in \Omega_{\mathcal{P}}} \sum_{\sigma \in f^{-1}(\tau)} \pi_\mathcal{P}(\sigma) \leq 
\sum_{\tau \in \Omega_{\mathcal{P}}} \pi_\mathcal{P}(\tau) \left(\frac{\sum_{\sigma \in f^{-1}(\tau)} \pi_\mathcal{P}(\sigma)}{\pi_\mathcal{P}(\tau)}\right) \,.
\]
Using Lemma \ref{lem:num_preimages_f_eta} and Lemma \ref{lem:min_gain_f_eta},
\[ 
\frac{\sum_{\sigma \in f^{-1}(\tau)} \pi_\mathcal{P}(\sigma)}{\pi_\mathcal{P}(\tau)} \leq \sum_{\ell} q\,3^{\vert \mathcal{P}\vert}\,4^{ \frac{1+3\delta}{4\delta}\ell}\,\left(\frac{1}{\gamma} \right)^{\frac{\ell}{q-1}}
\] 
Since $1< \alpha < \alpha^\ast(\eta, q)$, and $\alpha_c(\delta, \eta, q)$ defined in Lemma \ref{lem:critical_alpha} is a continuous monotonically decreasing function of $\delta$ with $\alpha_c(0, \eta, q) = \alpha^\ast(\eta, q)$ and $\alpha_c(\delta^\ast, \eta, q) = 1$, there exists $0<\delta_0<\delta^\ast(\eta, q)$ such that $\alpha_c(\delta_0, \eta, q) = (\alpha + \alpha^\ast(\eta, q))/2$. 

Moreover, note that $\alpha_c(\delta,\eta, q)$ is a concave-downward function of $\delta$ since  $\partial^2\alpha_c(\delta,\eta, q)/\partial \delta^2 <0$  $\forall 0<\delta<\delta^\ast(\eta, q)$. Denoting $\alpha_c(\delta,\eta, q)$ by $\alpha_c(\delta)$, $\delta^\ast(\eta, q)$ by $\delta^\ast$ and $\alpha^\ast(\eta, q)$ by $\alpha^\ast$, this gives:
\[ \alpha_c((1-\lambda)\delta^\ast + \lambda \cdot 0 ) > (1-\lambda)\alpha_c(\delta^\ast) + \lambda \alpha_c(0) = (1-\lambda)\alpha_c(\delta^\ast) + \lambda \alpha^\ast \,,\; \forall 0<\lambda < 1 \,.\]
Setting $\lambda = \frac{\tfrac{1}{2}(\alpha^\ast + \alpha) - 1}{\alpha^\ast - 1}$, we get:
\begin{align}
 \alpha_c((1-\lambda)\delta^\ast) &> \frac{\alpha^\ast + \alpha}{2} = \alpha_c(\delta_0) \Rightarrow (1-\lambda)\delta^\ast < \delta_0 \Rightarrow = \frac{\alpha^\ast - \alpha}{2(\alpha^\ast - 1)}\delta^\ast  < \delta_0 \label{concave_delta_inequality}
\end{align}
where the second inequality follows from $\alpha_c(\delta)$ being monotonically decreasing, and we substitute the value of $\lambda$ to get the last inequality. 


By Lemma \ref{lem:critical_alpha}, $\ell\geq \nu \sqrt{n}(\alpha_c(\delta_0, \eta, q) - \alpha)$ for any $\nu<2\sqrt{3}$ and $n$ sufficiently large. Since $\vert \mathcal{P}\vert \leq \alpha p_\mathrm{min} < 2\alpha\sqrt{3}\sqrt{n}$ by Lemma \ref{lem:isoperimetric_triangular_upper}, $\vert \mathcal{P}\vert$ satisfies:
\begin{equation}
\vert \mathcal{P}\vert <  \frac{2\alpha\sqrt{3}\ell}{\nu(\alpha_c(\delta_0, \eta, q)-\alpha)} = \frac{4\alpha\sqrt{3}\ell}{\nu(\alpha^\ast-\alpha)} \label{mod_P_inequality}
\end{equation}

Using the inequalities in equations~(\ref{mod_P_inequality}), (\ref{concave_delta_inequality}), we get: 
\begin{align*}
\frac{\sum_{\sigma \in f^{-1}(\tau)} \pi_\mathcal{P}(\sigma)}{\pi_\mathcal{P}(\tau)} &\leq q\sum_{\ell=\nu\sqrt{n}(\alpha^\ast -\alpha)/2}^{\infty} 3^{4\alpha\sqrt{3}\ell/\left(\nu(\alpha^\ast - \alpha)\right)}\,4^{\frac{3}{4}+\frac{1}{4\delta_0}\ell}\,\left(\frac{1}{\gamma} \right)^{\frac{\ell}{q-1}} \\
&<  q\sum_{\ell=\nu\sqrt{n}(\alpha^\ast -\alpha)/2}^{\infty} \left(\frac{ 3^{4\alpha\sqrt{3}/\left(\nu(\alpha^\ast - \alpha)\right)}\,4^{\frac{3}{4}+\frac{\alpha^\ast - 1}{2(\alpha^\ast - \alpha)}}}{\gamma ^{\frac{1}{q-1}}}\right)^\ell \,.
\end{align*}

Since we consider $\gamma$ such that: 
\[\gamma > \left( 3^\frac{2\alpha}{\alpha^\ast(\eta, q) - \alpha}\,4^{\frac{3}{4}+\frac{\alpha^\ast(\eta, q) - 1}{2\delta^\ast(\eta, q)\left(\alpha^\ast(\eta, q) - \alpha\right)}} \right)^{q-1}\,,\] 
there exists $\nu<2\sqrt{3}$ such that 
\[ \frac{3^{2\alpha\sqrt{3}/(\nu(\alpha_c - \alpha))}\,4^{\frac{1+3\delta^\ast}{4\delta^\ast}}}{\gamma^{1/(q-1)}} < 1 \Rightarrow \,\frac{\sum_{\sigma \in f^{-1}(\tau)} \pi_\mathcal{P}(\sigma)}{\pi_\mathcal{P}(\tau)} < \zeta^{\sqrt{n}}\,,\] 
for some $\zeta<1$.
Noting that $\sum_{\tau\in\Omega_\mathcal{P}} \pi_\mathcal{P}(\tau) = 1$, the proof of the Theorem follows.
\end{proof}
}

\subparagraph*{The Clock Model:}
Lemma \ref{lem:min_gain_f_eta} and Theorem \ref{thm:con_align_thm} hold for the clock model with stationary distribution $\pi_\mathrm{clock}$, with $\gamma$ replaced by $\gamma^{1-\cos(2\pi/q)}$ in both. The proofs follow on similar lines as for the Potts model.

\subsection{Non-Alignment and Expansion in Connected SOPS}
An interesting artifact of the alignment algorithm is that when $\lambda, \gamma$ are small, the opposite properties are achieved, namely non-alignment (Theorem~\ref{thm:conn_non_alignment}) and expansion (Theorem~\ref{thm:conn_expansion}).

\begin{theorem}[Non-alignment in Connected SOPS]
\label{thm:conn_non_alignment}
For any $q \geq 2$ and $\epsilon \in (0,\frac{1}{q})$, when $\gamma > 0$ satisfies:
\begin{align*}
\gamma^3 < \left( 1 - \epsilon \frac{q}{q-1} \right)^{\frac{q-1}{q}-\epsilon} \left(1+\epsilon\,q \right)^{\frac{1}{q}+\epsilon} = 1+ \frac{\epsilon^2 q^2}{q-1} + O(\epsilon^3) \,,
\end{align*}
the probability that a configuration sampled from the stationary distribution of the Markov chain algorithm $\pi_\mathrm{Potts}$ is not $\epsilon$-\emph{non-aligned} is exponentially small,  for sufficiently large $n$.
\end{theorem}

For $\epsilon > 0$, we say a configuration is $\epsilon$-\emph{non-aligned} if the fraction of particles of each orientation is within an $\epsilon$-neighborhood of $q^{-1}$.
Let $S_\mathcal{P}^{\epsilon}$ denote the set of configurations which have boundary $\mathcal{P}$ and are not $\epsilon$-\emph{non-aligned}, and let $S^\epsilon$ be the set of configurations that are not $\epsilon$-\emph{non-aligned}. 
We show that when $\gamma$ is sufficiently close to $1$, the probability that a configuration drawn from the stationary distribution of the Markov chain algorithm $(\pi_\mathrm{Potts})$ is not $\epsilon$-\emph{non-aligned} is exponentially small.

To prove Theorem~\ref{thm:conn_non_alignment}, we first show Lemma~\ref{lem:non_alignment_P}, which applies to configurations corresponding to a specific choice of boundary $\mathcal{P}$. This is as the probabilities of configurations drawn from this restricted subset are independent of the choice of the parameter $\lambda$. Because of this, we observe that our non-alignment result (Theorem~\ref{thm:conn_non_alignment}) does not rely on our choice of $\lambda$ in the distribution $\pi_\mathrm{Potts}$.

\begin{lemma}
\label{lem:non_alignment_P}
For any $q \geq 2$ and $\epsilon \in (0,\frac{1}{q})$, for any given boundary $\mathcal{P}$ of a particle configuration with $n$ particles, if $\gamma > 0$ satisfies
\[ \gamma^3 < \left( 1 - \epsilon \frac{q}{q-1} \right)^{\frac{q-1}{q}-\epsilon} \left(1+\epsilon\,q \right)^{\frac{1}{q}+\epsilon}
= 1+ \frac{\epsilon^2 q^2}{q-1} + O(\epsilon^3),\]
the probability that a configuration sampled from the stationary distribution $\pi_\mathcal{P}$ is not $\epsilon$-\emph{non-aligned} is exponentially small.
\end{lemma}
The proof follows from using Stirling's approximation \cite{robbins_remark_1955} to estimate the number of configurations that are not $\epsilon$-\emph{non-aligned}, and using rough lower and upper bounds on the weight of configurations in $\Omega_\mathcal{P}$.

\begin{proof}
For any configuration in $S_\mathcal{P}^{\epsilon}$,
the fraction of particles oriented along at least one direction is $0<\delta<1$, where $\vert \delta - q^{-1} \vert \geq \epsilon$. 
Let $\Omega_\mathcal{P}^{\delta}$ denote the set of configurations with boundary $\mathcal{P}$ where the fraction of particles oriented along at least one direction is exactly $\delta$.
The probability of the set of configurations $\Omega_\mathcal{P}^{\delta}$ then satisfies the following upper bound:
\[
\pi_\mathcal{P}(\Omega_\mathcal{P}^{\delta}) = \frac{w(\Omega_\mathcal{P}^{\delta})}{w(\Omega_\mathcal{P})} < \frac{q \binom{n}{\delta n}\,(q-1)^{(1-\delta)\,n}}{q^n\,\gamma^{-(3n-3)}}\,,
\]
where we have estimated an upper bound for  $w(\Omega_\mathcal{P}^{\delta})$ by choosing an orientation for the $\delta n$ particles in $q$ ways, assigning the other $q-1$ orientations to the rest of the $(1-\delta)n$ particles in $(q-1)^{(1-\delta)n}$ ways, and assigning the highest weight of $1$ to all configurations in $\Omega_\mathcal{P}^{\delta}$, and estimated a lower bound for $w(\Omega_\mathcal{P})$ by assigning one of $q$ orientations to each particle in $q^n$ ways, and assigned each configuration the lowest possible weight which is $\gamma^{-(3n-3)}$. 

Using Stirling's approximation \cite{robbins_remark_1955}, we can show that the following holds:
\[ \binom{n}{\delta n} \leq \frac{1}{\sqrt{2\pi n\delta(1-\delta)}}\frac{1}{(1-\delta)^{(1-\delta) n} (\delta)^{\delta n}}\,.\]

Substituting, we get for  $\pi_\mathcal{P}(\Omega_\mathcal{P}^{\delta})$, 
\begin{align*}
 \pi_\mathcal{P}(\Omega_\mathcal{P}^{\delta}) &\leq 
\frac{q\,(q-1)^{(1-\delta)n}}{\gamma^3\,q^n\,\gamma^{-3n}\,(1-\delta)^{(1-\delta)\,n}(\delta)^{\delta n}\sqrt{2\pi\,n\,\delta(1-\delta)}} \\
&= \frac{q}{\gamma^3\sqrt{2\pi n\delta(1-\delta)}} \frac{\gamma^{3n}}{\left( \frac{q(1-\delta)}{q-1}\right)^{(1-\delta)n}(q\delta)^{\delta n} } \\
&= \frac{q}{\gamma^3\sqrt{2\pi n\delta(1-\delta)}} \left( \frac{\gamma^{3}}{\left( \frac{q(1-\delta)}{q-1}\right)^{(1-\delta)}(q\delta)^{\delta} } \right)^n
\end{align*}

We show that for sufficiently large values of $n$, this upper bound for $\pi_\mathcal{P}(\Omega_\mathcal{P}^{\delta})$, for $\delta \in [0,q^{-1} - \epsilon]\cup[q^{-1}+\epsilon,1]$ is maximized at $\delta = q^{-1}+\epsilon$. 
To do this, it suffices to show that the expression $\left(\frac{q(1-\delta)}{q-1}\right)^{1-\delta}(q\delta)^{\delta}$ is minimized at $\delta = q^{-1} + \epsilon$.
We denote
\begin{align*}
h(x) &:= \left(\frac{1-(q^{-1}+x)}{1-q^{-1}}\right)^{1-(q^{-1}+x)}\left(q(q^{-1}+x)\right)^{q^{-1}+x} \\
&= \left(1-\frac{x}{1-q^{-1}}\right)^{1-q^{-1}-x}\left(1+\frac{x}{q^{-1}}\right)^{q^{-1}+x} \,,
\end{align*}
so that $h(\delta-q^{-1}) = \left(\frac{q(1-\delta)}{q-1}\right)^{1-\delta}(q\delta)^{\delta}$.
By showing that 
\begin{align*}
h'(x) = \left(1-\frac{x}{1-q^{-1}}\right)^{1-q^{-1}-x}\left(1+\frac{x}{q^{-1}}\right)^{q^{-1}+x} \log \frac{1 + x/q^{-1}}{1 - x/(1-q^{-1})} ,
\end{align*}
we have $h'(x) > 0$ for $x < 0$ and $h'(x) > 0$ for $x > 0$, which implies that the minimum value of $h(x)$ in the domain $[-q^{-1},-\epsilon]\cup[\epsilon,1-q^{-1}]$ is either $h(-\epsilon)$ or $h(\epsilon)$. To show that $h(\epsilon) \leq h(-\epsilon)$, we define
\begin{align*}
g(x) &:= h(-x)/h(x)
= \left(\frac{1-q^{-1}+x}{1-q^{-1}-x}\right)^{1-q^{-1}}
\left(\frac{q^{-1}-x}{q^{-1}+x}\right)^{q^{-1}}
\left(\frac{1-\frac{x^2}{(1-q^{-1})^2}}{1-\frac{x^2}{(q^{-1})^2}}\right)^{x} \,.
\end{align*}
Differentiating $g(x)$, we find that 
\begin{align*}
g'(x) &= \left(\frac{1-q^{-1}+x}{1-q^{-1}-x}\right)^{1-q^{-1}}
\left(\frac{q^{-1}-x}{q^{-1}+x}\right)^{q^{-1}}
\left(\frac{1-\frac{x^2}{(1-q^{-1})^2}}{1-\frac{x^2}{(q^{-1})^2}}\right)^{x} 
\log \left(\frac{1-\frac{x^2}{(1-q^{-1})^2}}{1-\frac{x^2}{(q^{-1})^2}}\right) \\
&= g(x) \log \left(\frac{1-\frac{x^2}{(1-q^{-1})^2}}{1-\frac{x^2}{(q^{-1})^2}}\right)
\end{align*}
Assuming $q \geq 2$, we have $q^{-1} \leq 1 - q^{-1}$, so $g'(x)$ is positive whenever $g(x) > 0$. Because $g(0) = 1$, if there exists an $x > 0$ where $g(x) \leq 0$, by setting $x^* := \inf\{x > 0 : g(x) \leq 0\}$, we would obtain a contradiction as this would imply the existence of a value $x \in (0,x^*)$ where $g'(x) < 0$ but $h(x) > 0$. Thus $g'(x)$ is positive for all $x > 0$, which implies that $g(x) \geq 1$ for all $x > 0$, and in particular $g(\epsilon) \geq 1$, so $h(\epsilon) \leq h(-\epsilon)$.

Thus, the probability of the set of configurations $S_\mathcal{P}^{\epsilon}$ can be estimated as:
\begin{align*}
\pi_\mathcal{P}(S_\mathcal{P}^{\epsilon})
&\leq \sum_{m\in[n],|q^{-1}-\frac{m}{n}| \geq \epsilon} \pi_\mathcal{P}(\Omega_\mathcal{P}^{m/n}) \\
    &\leq n\,\pi_\mathcal{P}(\Omega_\mathcal{P}^{q^{-1}+\epsilon})
    \\
    &\leq n  \,\frac{q}{\gamma^3\sqrt{2\pi n(q^{-1}+\epsilon)(1-q^{-1}-\epsilon)}} \left( \frac{\gamma^{3}}{\left( 1-\frac{\epsilon q}{q-1}\right)^{(\frac{q-1}{q}-\epsilon)}(1+\epsilon\,q)^{\epsilon+q^{-1}} } \right)^n < \zeta^n \,,
\end{align*}
for sufficiently large values of $n$ for some $\zeta < 1$, when $\gamma$ satisfies:
\begin{align*}
\gamma^3 < \left( 1 - \epsilon \frac{q}{q-1} \right)^{\frac{q-1}{q}-\epsilon} \left(1+\epsilon\,q \right)^{\frac{1}{q}+\epsilon}  \,.
\end{align*}

To obtain the the asymptotic expansion of the right-hand side of the above inequality, we begin by simplifying the expression as follows.
\begin{align*}
\left( 1 - \epsilon \frac{q}{q-1} \right)^{\frac{q-1}{q}-\epsilon} \left(1+\epsilon\,q \right)^{\frac{1}{q}+\epsilon} &= \left( 1 - \epsilon \frac{q}{q-1} \right)^{1-\frac{1}{q}-\epsilon} \left(1+\epsilon\,q \right)^{\frac{1}{q}+\epsilon} \\
&= \left( 1 - \epsilon \frac{q}{q-1} \right) \left( \frac{1+\epsilon\,q}{1 - \epsilon \frac{q}{q-1}}\right)^{\frac{1}{q}+\epsilon}.
\end{align*}
Using $(1-x)^{-1} = 1+x+x^2 + O(x^3),\,0<x<1$, we get:
\begin{align*}
& \left( 1 - \epsilon \frac{q}{q-1} \right) \left( \frac{1+\epsilon\,q}{1 - \epsilon \frac{q}{q-1}}\right)^{\frac{1}{q}+\epsilon} \\
&= \left( 1 - \epsilon \frac{q}{q-1} \right) \left( (1+\epsilon\,q)\left(1 + \epsilon \frac{q}{q-1} + \epsilon^2 \frac{q^2}{(q-1)^2} + O(\epsilon^3) \right)\right)^{\frac{1}{q}+\epsilon} \\
&= \left( 1 - \epsilon \frac{q}{q-1} \right) \left( 1+\epsilon \frac{q^2}{q-1} + \epsilon^2 \frac{q^3}{(q-1)^2} + O(\epsilon^3) \right)^{\frac{1}{q}+\epsilon} \\
&=  \left( 1 - \epsilon \frac{q}{q-1} \right) \left( 1+\left(\frac{1}{q}+\epsilon\right)\left(\epsilon \frac{q^2}{q-1} + \epsilon^2 \frac{q^3}{(q-1)^2} \right) + O(\epsilon^3)\right)\\
&= 1 + \frac{\epsilon^2 q^2}{q-1} + O(\epsilon^3)
\end{align*}
\end{proof}

This allows us to prove Theorem~\ref{thm:conn_non_alignment}.

\begin{proof}[Proof of Theorem~\ref{thm:conn_non_alignment}]
Let $\zeta<1$ be such that $\pi_\mathcal{P}(S_\mathcal{P}^{\epsilon}) < \zeta^n$. Then we have for $\pi_\mathrm{Potts}(S^\epsilon)$,
\begin{align*}
\pi_\mathrm{Potts}(S^\epsilon) &= \sum_{\mathcal{P}} \pi_\mathrm{Potts}(S_\mathcal{P}^{\epsilon}) = \sum_{\mathcal{P}} \pi_\mathrm{Potts}(\Omega_\mathcal{P}) \,\pi_\mathcal{P}(S_\mathcal{P}^{\epsilon}) \leq \sum_{\mathcal{P}} \pi_\mathrm{Potts}(\Omega_\mathcal{P}) \zeta^n \leq \zeta^n \,.
\end{align*}
\end{proof}
The result also holds for the clock model with $\gamma$ replaced with $\gamma^2$.

\begin{theorem}[Expansion in Connected SOPS]
\label{thm:conn_expansion}
We say a configuration $\sigma$ is $\beta$-expanded when its perimeter $p(\sigma)$ is greater than $\beta\,p_\mathrm{max}$, where $0<\beta<1$.
For constants $\lambda, \gamma > 0, c_1 = 2.17, c_2 = 2+\sqrt{2}$ such that $\lambda \,\gamma^{5/2} < c_1$, and for any $\beta$ such that
\[0< \beta < \frac{\log c_1 - \log\lambda - \frac{5}{2}\log\gamma}{\log c_2 - \log\lambda-\log\gamma}\,,\] the probability that a configuration drawn from the stationary distribution $\pi$ is not $\beta$-expanded is exponentially small.
\end{theorem}
This definition of expansion corresponds with the definition in \cite{cannon_markov_2016}.
Denote by $S_\beta$ the set of configurations that are not $\beta$-expanded.
We can get rough upper and lower bounds for the weight of configurations in $\Omega_\mathcal{P}$ by estimating the number of ways of getting a fixed perimeter using the bounds in \cite{duminil-copin_connective_2012,cannon_markov_2016}.

\begin{proof}
The probability of the configurations in $S_\beta$ can be estimated as follows:
\[\pi(S_\beta) \leq \frac{\sum_{k=p_\mathrm{min}}^{\beta\,p_\mathrm{max}} \sum_{\mathcal{P}:\vert\mathcal{P}\vert = k} w(\Omega_\mathcal{P})}{\sum_{\mathcal{P}:\vert\mathcal{P}\vert = p_\mathrm{max}} w(\Omega_\mathcal{P})} \]
The weight $w(\Omega_\mathcal{P})$ of configurations with given perimeter $\mathcal{P}$ obeys: $(\lambda\,\gamma)^{-\vert \mathcal{P}\vert} q^n\,\gamma^{-(3n-3)} \leq w(\Omega_\mathcal{P} \leq (\lambda\,\gamma)^{-\vert \mathcal{P}\vert} q^n )$
Moreover, from \cite{cannon_markov_2016}, the number of configurations with perimeter $\mathcal{P}:\vert \mathcal{P}\vert = k$ is at most $\nu^k$ for any $\nu>2+\sqrt{2}$, and the number of configurations with perimeter $\mathcal{P}:\vert \mathcal{P}\vert = p_\mathrm{max}$ is at least $2.17^{p_\mathrm{max}}/(0.13)$. Moreover, note that $p_\mathrm{max}=2n-2= 2\,(3n-3)/3$ Substituting, we get:
\begin{align*}
    \pi(S_\beta) &\leq \frac{(0.13)^{-1}\sum_{k=p_\mathrm{min}}^{\beta\,p_\mathrm{max}} (\lambda\,\gamma)^{-k}\,\nu^k\, q^n}{(\lambda\,\gamma)^{-p_\mathrm{max}} (2.17)^{p_\mathrm{max}}\,\gamma^{-(3n-3)}\,q^n} =  \frac{(0.13)^{-1}\sum_{k=p_\mathrm{min}}^{\beta\,p_\mathrm{max}} (\lambda\,\gamma)^{-k}\,\nu^k}{(\lambda\,\gamma)^{-p_\mathrm{max}} (2.17)^{p_\mathrm{max}}\,\gamma^{-3\,p_\mathrm{max}/2}} \\
    &= (0.13)^{-1}\sum_{k=p_\mathrm{min}}^{\beta\,p_\mathrm{max}} \left(\frac{\nu}{\lambda\,\gamma}\right)^k \left( \frac{\lambda\,\gamma^{5/2}}{2.17}\right)^{p_\mathrm{max}} .
\end{align*}
Since $k/\beta \leq p_\mathrm{x}$ and $\lambda\,\gamma^{5/2}<2.17$, we get:
\begin{align*}
\pi(S_\beta) &\leq (0.13)^{-1}\sum_{k=p_\mathrm{min}}^{\beta\,p_\mathrm{max}} \left(\frac{\nu}{\lambda\,\gamma}\right)^k \left( \frac{\lambda\,\gamma^{5/2}}{2.17}\right)^{k/\beta} \\
&= (0.13)^{-1}\sum_{k=p_\mathrm{min}}^{\beta\,p_\mathrm{max}} \left(\frac{\nu}{\lambda\,\gamma} \left( \frac{\lambda\,\gamma^{5/2}}{2.17}\right)^{1/\beta}\right)^k  < \zeta^{\sqrt{n}}\,,
\end{align*}
with $\zeta<1$, when the following condition is satisfied: 
\begin{align*}
 \frac{\nu}{\lambda\,\gamma} \left( \frac{\lambda\,\gamma^{5/2}}{2.17}\right)^{1/\beta}<1  &\iff \frac{2+\sqrt{2}}{\lambda\,\gamma} \left( \frac{\lambda\,\gamma^{5/2}}{2.17}\right)^{1/\beta} < 1 \\
&\iff \log(2+\sqrt{2})-\log\lambda-\log\gamma + \frac{\log\lambda + \frac{5}{2}\log\gamma - \log 2.17}{\beta} < 0 \\
&\iff \beta < \frac{\log 2.17 - \log\lambda - \frac{5}{2}\log\gamma}{\log(2+\sqrt{2})-\log\lambda-\log\gamma}\,. 
\end{align*}
The above condition is satisfied by the statement of the theorem. 
\end{proof}

The same theorem holds for the clock model, with $\gamma^{5/2}$ replaced by $\gamma^4$ in the theorem statement, and the proof follows on similar lines.

\section{Aggregation and Alignment in General SOPS}
\label{sec:general}

In general SOPS, occupying any selection of $n$ out of the $N$ possible sites of $G_\Delta$ is a valid configuration. Hence, we apply the same Metropolis-Hastings Markov chain as the connected SOPS model, with the exception that any move into an unoccupied location is considered valid regardless of connectivity effects.
In this disconnected setting, particles exist on a lattice region with toroidal boundary conditions. We assume the particles occupy a constant fraction $\rho$ of the lattice. Specifically, we define a $\rho \in (0,\frac{1}{3})$ so that $n = \rho N$. 
The set of possible configurations is denoted $\OmegaAggrRho$.

%
Similar to before, \emph{boundary contour} $bd(R)$ of a region $R \subseteq V(\TorusLattice)$ refers to the set of dual edges on $\TorusDual$ corresponding to edges between sites in $R$ and $V(\TorusLattice) \setminus R$. The \emph{boundary length} of $R$ is $|bd(R)|$.
Let $bd_{\min}(k)$ denote the minimum boundary length of a region of $k$ sites in $V(\TorusLattice)$.
We restrict $\rho$ to be less than $\frac{1}{3}$ as cases with so many particles (filled sites) that minimum boundary length configurations wrap around the torus $\TorusLattice$ is not instructive for our purposes (a precise explanation for this restriction is in the full version of the paper).

We show that in this general SOPS model, both alignment and aggregation can be achieved with high probability using only local movements. Alignment is defined in Section~\ref{sec:alignment}, and aggregation is defined as follows:
\begin{definition}[Aggregation]
\label{defn:aggregation}
For $\alpha > 1$, $\delta>0$ we say a configuration of $n$ particles is $\alpha,\delta$-aggregated if there exists a region $\mathcal{R}$ such that
\begin{enumerate}
\item The number of empty sites within $\mathcal{R}$ is at most $\delta|\mathcal{R}|$.
\item The number of particles outside of $\mathcal{R}$ is at most $\delta(N-|\mathcal{R}|)$
\item The boundary length of $\mathcal{R}$ is at most $\alpha \cdot bd_{\min}(n)$.
\end{enumerate}
\end{definition}

Note that changes in the perimeter of the configuration cannot be locally computed if the set of particles is disconnected. So instead, we make use of the boundary contour length to define our Hamiltonian. 
More precisely, we consider the following Potts Hamiltonian, another variant of the site-diluted Potts Hamiltonian \cite{chayes_aggregation_1995}, on $\TorusLattice$:
\[ \tilde{H}_\mathrm{Potts}(\sigma) = -J\sum_{\langle i,j\rangle} \left[ n_i n_j\,\left(\delta_{\theta_i, \theta_j}-1\right) + \left(n_i (n_j-1) + n_j (n_i - 1) \right)\right] \,,\]
where the sum is over all pairs of adjacent sites: $\langle i,j \rangle$ i.e., sites connected by a single lattice edge in $\TorusLattice$, $n_i \in \{0,1\}$ indicates whether site $i$ is occupied or not, $\theta_i$ indicates the orientation of the particle on site $i$, and $J$ is a positive constant. We only consider configurations $\sigma$ in $\OmegaAggrRho$ i.e., where the total number of particles is equal to $n$. 

The probability of a configuration $\tilde{\pi}_\mathrm{Potts}(\sigma)$ is given by the Boltzmann distribution which can be expressed in terms of the parameter $\lambda = \exp(\beta J)$ as:
\begin{equation}
    \piFinite(\sigma) = \frac{\wFinite(\sigma)}{\ZFinite}\,,\; \wFinite(\sigma) = \lambda^{-a(\sigma) -h(\sigma)}\,,\;  \ZFinite = \sum_{\sigma' \in \OmegaAggrRho} \wFinite(\sigma') \label{potts_distribution_unconnected}
\end{equation}
where $\lambda > 0$, $h(\sigma)$ is the number of heterogeneous edges in the configuration $\sigma$, and $a(\sigma)$ is the number of edges between occupied and unoccupied sites in $\TorusLattice$.

We prove the following theorem that establishes aggregation and alignment for appropriate settings of the parameters.

\begin{theorem}
\label{theorem:aggregationwithalignment}
Fix $\rho < \frac{1}{3}$ and assume that there will always be exactly $\rho N$ filled sites on the lattice.
For any $\delta > 0$ and $\alpha > 1$, there exists a $\lambda_0 = \lambda_0(q,\rho,\alpha,\delta)$ such that for all $\lambda > \lambda_0$,
with probability $1-\zeta^{\sqrt{N}}$ for some constant $\zeta = \zeta(q, \rho, \alpha, \delta, \lambda) < 1$, 
there exists a region $\mathcal{R} \subseteq V(\TorusLattice)$, where
\begin{enumerate}
\item There is an orientation $\theta \in \{0,1,\ldots,q-1\}$ where the number of filled sites with orientation $\theta$ in $\mathcal{R}$ is at least $(1-\delta)|\mathcal{R}|$.
\item The number of filled sites not in $\mathcal{R}$ is at most $\delta (N-|\mathcal{R}|)$
\item The boundary length of $\mathcal{R}$ is at most $\alpha \cdot bd_{\min}(\rho N)$.
\end{enumerate}
\end{theorem}



Recall that $bd_{\min}(k)$ denotes the minimum boundary length of a region of $k$ sites in $V(\TorusLattice)$.
$bd_{\min}(k)$ grows in the following manner:
\begin{lemma}
\label{lem:bdmin_growth}
Let $c$ be a constant. 
\begin{enumerate}
\item If $0 < c < 1/3$, we have
$bd_{\min}(cN) = 4\sqrt{3c}\sqrt{N} + O_N(1)$.
\item If $1/3 \leq c \leq 2/3$, we have $bd_{\min}(cN) = 4\sqrt{N} + O_N(1)$.
\item Finally, if $c > 2/3$, we have $bd_{\min}(cN) = bd_{\min}((1-c)N) = 4\sqrt{3(1-c)}\sqrt{N} + O_N(1)$.
\end{enumerate}
\end{lemma}

\begin{proof}
For values of $k$ that are small enough that minimal boundary length regions enclosing $k$ particles do not wind around the lattice $\TorusLattice$, we may apply the the isoperimetric lemmas (Lemmas \ref{lem:isoperimetric_triangular_lower} and \ref{lem:isoperimetric_triangular_upper}) and the fact that the boundary length of a connected, hole-free region of perimeter $p$ is $2p+6$ to deduce that $4\sqrt{3}\sqrt{k-1/4} \leq bd_{\min}(k) \leq 4\sqrt{3}\sqrt{k}+6$. As a region that winds around $\TorusLattice$ must have a boundary length of at least $4\sqrt{N}$, we can say that $bd_{\min}(cN) = 4\sqrt{3c}\sqrt{N} + O_N(1)$ for $0 < c < 1/3$.
Note that for $1/3 \leq c \leq 2/3$, we have $bd_{\min}(cN) = 4\sqrt{N} + O_N(1)$ and for $c > 2/3$, we have $bd_{\min}(cN) = bd_{\min}((1-c)N) = 4\sqrt{3(1-c)}\sqrt{N} + O_N(1)$.
\end{proof}

We can treat the general SOPS problem as a $q+1$-state Potts model on $\TorusLattice$ with $q+1$ orientations $\{-1,0,1,\ldots,q-1\}$ in which the number of sites assigned $-1$ is fixed to be exactly $(1-\rho)N$, where $N = |V(\TorusLattice)|$.
In other words, sites of the lattice are no longer filled or unfilled, but are instead assigned one of $q+1$ 
 orientations with the special spin $-1$ assigned to unoccupied lattice sites. 
 We refer to any edge between particles of differing orientations as ``heterogeneous edges,'' including those assigned the special orientation $-1$.

We again use a Peierls argument to show that for suffiently large $\lambda,$ the configuration will compress and one of the $q$ orientations will dominate, with high probability. This proof is an adaptation of the bridging argument used for separation in~\cite{cannon_local_2019-1, cannon_local_2019} and thus follows their arguments very closely.
The following sections build up to a proof of Theorem~\ref{theorem:aggregationwithalignment}.


We observe that the result of  Theorem~\ref{theorem:aggregationwithalignment} will imply both alignment and aggregation (for some values of $\alpha$ and $\delta$) as given in Definitions~\ref{defn:alignment} and \ref{defn:aggregation}.
The key component of our proof is the construction of a $\delta$-bridge system ($\delta \in (0,1)$ is a positive constant) for each configuration in $\OmegaAggrRho$.
Recall that a bridge system is a connected network of the long contours of a configuration $\sigma$, that is used to ``remove'' long contours in the Peierls argument to show that they are unlikely. It will also be used to define the region $\mathcal{R}$ required for Theorem~\ref{theorem:aggregationwithalignment}.



Let $\Ewrap$ be the set of edges on $\TorusDual$ corresponding to the edges on $\TorusLattice$ that wrap around the torus. Thus $|\Ewrap| = 2\sqrt{N}-1$.
In a setting with more than three possible orientations, regions of differing orientations are divided up by networks of contours rather than closed walks separating two different orientations. We call these contour networks \emph{complex contours}.
Formally, a complex contour refers to a connected subgraph of $G_{\varhexagon}$ of minimum degree at least $2$.
For a given configuration $\sigma \in \OmegaAggrRho$, the set of edges $\mathcal{C}$ on $\TorusDual$ corresponding to its heterogeneous edges will be a union of complex contours. The complex contours of $\sigma$ thus refers to the edge sets of connected components of the subgraph induced by $\mathcal{C}$ in $\TorusDual$.


We now define a bridge system $(B,I,\Theta)$ where the set $I$ represents the complex contours in the bridge system, $B$ represents the bridges used to connect these complex contours, and $\Theta$ is a mapping that assigns an orientation to each of the components formed after removing  the edges of $\TorusLattice$ corresponding to the edges in $I$.

\begin{definition}[Bridge Systems]
Fix $\delta > 0$. Consider a tuple $(B, I, \Theta)$, where $B$ and $I$ are subsets of $E(\TorusDual)$ and $\Theta: V(\TorusLattice) \to \{-1,0,1,\ldots,q-1\}$ is a function assigning each vertex an orientation or the value $-1$ (which we will use to represent vacant sites).
We say $(B, I, \Theta)$ is a $\delta$-bridge system if:
\begin{enumerate}
\item The subgraph induced in $\TorusDual$ by $I$ has no vertex of degree less than $2$. Practically, $I$ represents a union of complex contours that subdivides $\TorusLattice$ into regions.
\item The subgraph induced in $\TorusDual$ by $B \cup I \cup \Ewrap$ is connected and has no vertex of degree less than $2$.
\item $B \cap I = \emptyset$ and $|B| \leq \frac{1-\delta}{2\delta} |I|$
\item For any two neighboring sites $u,v \in \TorusLattice$, $\Theta(u) = \Theta(v)$ if and only if the dual edge corresponding to $\{u,v\}$ is not in $I$.
\end{enumerate}
\end{definition}


Consider a set of edges $I$, that is a union of the edge sets of complex contours.
Let $\sigma$ be a configuration in $\OmegaAggrRho$.
We say a complex contour $C$ of $\sigma$ is \emph{bridged} (by $I$) if $C \subseteq I$. 
We say a site $v$ is bridged (by $I$) if there is a path over $\TorusLattice$ using only sites of the same orientation (including $-1$) in $\sigma$ as $v$ to a site incident to an edge in $I$.
Consider a region $R \subseteq V(\TorusLattice)$ that is connected as an induced subgraph of $\TorusLattice$.
We call $R$ a \emph{bridged region} if $bd(R) \subseteq I$ and a \emph{minimal bridged region} if there is no bridged region $R'$ where $R' \subseteq R$.
Notably, the edge set $I$ partitions $V(\TorusLattice)$ into minimal bridged regions.

\begin{definition}[Bridge System for a Configuration]
Fix $\delta > 0$ and a configuration $\sigma \in \OmegaAggrRho$.
We say a tuple $(B, I, \Theta)$ is a $\delta$-bridge system for a configuration $\sigma$ If
\begin{enumerate}
\item Each minimal bridged region $R$ by $(B,I,\Theta)$ contains at most $\delta |R|$ unbridged particles.
\item No complex contour $C$ of $\sigma$ meets any edge in $B\cup I\cup \Ewrap$. Formally, the edge-induced subgraphs $\TorusLattice[C]$ and $\TorusLattice[B\cup I\cup \Ewrap]$ do not share any vertices.
\item For each minimal bridged region $R$, $\Theta(v)$ must have the same value for every site $v \in R$ and this value $\Theta(v)$ must correspond to the orientation in $\sigma$ of some bridged particle in $R$.
\end{enumerate}
\end{definition}

\begin{definition}[Orientation of a Minimal Bridged Region]
\label{defn:regionorientation}
Given a $\delta$-bridge system $(B,I,\Theta)$ for a configuration $\sigma \in \OmegaAggrRho$. We can associate with each minimal bridged region $R$ of $I$ an orientation $y_R \in \{-1,0,1,\ldots,q-1\}$.

To determine $y_R$, we denote by $R^*$ the set of sites $v \in R$ with a path over $\TorusLattice$ using only sites of the same orientation in $\sigma$ as $v$ to a site incident to an edge in $bd(R)$. We note that $bd(R) \subseteq I$ and the edges $B \cup I \cup \Ewrap$ connect the components of $bd(R)$ in $\TorusDual$. This implies that every vertex in $R^*$ must have the same orientation in $\sigma$, as any contour $C$ between regions of differing orientations in $R^*$ must intersect $B \cup I \cup \Ewrap$, implying that $C$ also must be included in the set $I$, allowing us to subdivide $R$, contradicting its minimality.
The orientation $y_R$ of $R$ is thus defined to be the common orientation of the sites of $R^*$.
\end{definition}
Thus, for each minimal bridged region $R$ with orientation $y_R$, we must have $\Theta(v) = y_R$ for all $v \in R$.
The proofs of the Lemmas will be given in the long version of the paper. 

Our next step is to associate with each $\sigma \in \OmegaAggrRho$ a $\delta$-bridge system.
\begin{lemma}
\label{lem:aggr_bridge_system_construction}
For each $\sigma \in \OmegaAggrRho$ and $\delta \in (0,1)$, we can construct a $\delta$-bridge system $\mathcal{B}_\delta(\sigma) = (B_\delta(\sigma), I_\delta(\sigma), \Theta_\delta(\sigma))$, which is a $\delta$-bridge system of $\sigma$.
\end{lemma}

\input{diagram_bridging}

\begin{proof}
We initialize $B$ and $I$ to be empty sets and 
start by including all complex contours connected to an edge in $\Ewrap$ over $\TorusDual$ in $I$.
Take note that whenever any edge of a complex contour $C$ is included in $I$, we will always also add the entirety of $C$ in $I$.
The complex contours included in $I$ subdivides $V(\TorusLattice)$ into bridged regions.


As long as there exists a minimal bridged region $R$ where more than $\delta|R|$ of its particles are unbridged, there must be a column of particles $R_x := (\{x\} \times \{1,2,\ldots,\sqrt{N}-1\}) \cap R$ (for some $x \in \{1,2,\ldots,\sqrt{N}-1\}$) containing more than $\delta|R_x|$ unbridged particles.

For each particle $(x,y)$ in $R_x$, if $(x,y)$ is already bridged, we add the dual edges corresponding to the edges to the column to the right (to $(x+1,y)$ and $(x+1,y+1)$) to $B$, if they are not already in $I$ or $B$.
Next, we will add to $I$ every complex contour that touches any of the newly added edges in $B$. Refer to Figure \ref{fig:bridging} for an example.
As for why this maintains the property that $|B| \leq \frac{1-\delta}{2\delta}|I|$, we observe that each contiguous block $Y$ of unbridged particles in $R_y$ must be surrounded by some complex contour $C_Y$ (meaning $C_X$ must be connected in $\TorusDual$), which we would add to $I$. This complex contour $C_Y$ cannot be currently bridged else the block $Y$ would already be bridged, which implies that $C_Y$ will not touch the boundary of the region, nor will it wrap around the edges of the torus. Thus, adding the edges of $C_Y$ to $I$ will add at least $4|Y|$ adges to $I$, so doing this for all continguous blocks $Y$ will add at least $4\delta|R_x|$ edges to $I$, while adding at most $2(1-\delta)|R_x|$ edges to $B$.
In addition, as every edge added to $I$ will be adjacent to at least one other edge (in $I$, $B$ or $\Ewrap$) on each end, we maintain the property that the graph induced by $I \cup B \cup \Ewrap$ is connected and has no vertex of degree less than $2$.

As we can repeat this process as long as more than $\delta|R|$ particles in any minimal bridged region $R$ are unbridged and each repetition of this process can only cause more unbridged particles to become bridged, we will eventually attain a pair $(B,I)$ where every minimal bridged region $R$ will have at most $\delta|R|$ unbridged particles. Finally, we remove all edges from $B$ that are also in $I$ so that $B \cap I = \emptyset$. 

To construct $\Theta$, we use the minimal bridged regions obtained by partitioning $V(\TorusLattice)$ using $I$. Using the method described in Definition \ref{defn:regionorientation}, we assign each minimal bridged region $R$ an orientation $y_R$, which will be the orientation of every bridged particle in $R$. We thus set $\Theta(v) := y_R$ for each $v \in R$. Thus, $(B,I,\Theta)$ satisfies the requirements for a $\delta$-bridge system of $\sigma$.
\end{proof}
%
Without reference to any specific configuration in $\OmegaAggrRho$, we use the connectedness requirement of bridge systems to compute an upper bound on the number of bridge systems that is exponential on $|I|$. This is important as the Peierls argument ``removes'' the heterogeneous edges in $I$, which gives an improvement in weight of a similar order of growth.
\begin{lemma}
\label{lem:num_bridge_systems}
The number of $\delta$-bridge systems $(B,I,\Theta)$ where $|I| = \ell$ is at most $7 \cdot 6^{2\sqrt{N}-1} \cdot (3(q+1))^{\frac{1+\delta}{2\delta}\ell}$.
\end{lemma}

\begin{proof}
As the subgraph of $\TorusDual$ induced by $B \cup I \cup \Ewrap$ is connected, we can apply a depth-first search on this subgraph to identify all of these edges. To reconstruct $B \cup I \cup \Ewrap$, we start from a pre-determined vertex incident to an edge in $\Ewrap$. From this vertex, as $\TorusDual$ is a hexagonal lattice, there are three directions one can initially move in, which gives us $2^3-1 = 7$ possibilities, as either one, two or all three paths can be taken. From then on, each time we traverse an edge, as our induced subgraph has no vertex of degree less than $2$, there are $3$ possibilities on the next branch - either the path goes only left, only right, or branches both to the left and to the right. We keep track of each time the path branches (including on the initial vertex) in the same fashion as a depth-first search and each time we hit a dead end in the path we are currently following (which happens if the possibilities forward we have chosen only lead into edges that have already been traversed), we backtrack to the most recent branch and continue from there. In total, this gives us an upper bound of $7 \cdot 3^{|B \cup I \cup \Ewrap|}$ possibilities.

Next, we note that edges in $\Ewrap$ may also be part of $B$ or $I$ (but not both, as $B \cap I = \emptyset$). This means that if we want to determine the exact set of edges $B \cup I$,
for each edge in $\Ewrap$, we need to additionally determine if it is a part of $B \cup I$, which gives us $2^{|\Ewrap|}$ possibilities.

Finally, each edge in $B \cup I$ corresponds to an edge $\{u,v\}$ on $\TorusLattice$. Thus $\Theta$ can be reconstructed by encoding the difference between $\Theta(u)$ and $\Theta(v)$ across each edge, which can be done in most $(q+1)^{|B|+|I|}$ ways. As this difference is $0$ if and only if the edge is in $B$ rather than $I$, this also allows us to identify if the edge belongs to $B$ or $I$.
To wrap up, as $|B| \leq \frac{1-\delta}{2\delta}|I| = \frac{1-\delta}{2\delta}\ell$ and $|\Ewrap| = 2\sqrt{N} - 1$, the number of $\delta$-bridge systems is at most
$7 \cdot 6^{2\sqrt{N}-1} \cdot (3(q+1))^{\frac{1+\delta}{2\delta}\ell}$.
\end{proof}


Assuming $\delta \in (0,\rho)$, we define $\OmegaAggrRho_\ell := \{\sigma \in \OmegaAggrRho : |I_\delta(\sigma)| = \ell\}$,
where $I_\delta(\sigma)$ is comes from the $\delta$-bridge system constructed for $\sigma$.
Also, let $\OmegaAggr{\leq \delta N}$ be the the set of configurations over $\TorusLattice$ where at least $(1-\delta)N$ sites have orientation $-1$ (this corresponds to empty sites in our model). Note that $\OmegaAggr{\leq \delta N} \not\subseteq \OmegaAggrRho$.


\subparagraph*{Constructing the mappings used by the Peierls argument.} For the Peierls argument, we define two functions, $f^1_\ell : \OmegaAggrRho_\ell \to \OmegaAggr{\leq \delta N}$ and $f^2 : \OmegaAggr{\leq \delta N} \to \OmegaAggrRho$.
The function $f^1_\ell$ is used to erase the heterogeneous edges in $I$, creating a configuration of significantly higher weight, though not one with $\rho N$ particles. To fix this, a second function, $f^2$ is used to restore the number of particles back to $\rho N$. This way, $f^2 \circ f^1_\ell$ maps each $\sigma$ in $\OmegaAggrRho_\ell$ to a valid configuration with exactly $\rho N$ filled sites.

We first define the function $f^1$ for each integer $\ell \geq 0$.
For $\sigma \in \OmegaAggrRho_\ell$, we consider its $\delta$-bridge system $\mathcal{B}_\delta(\sigma) = (B_\delta(\sigma), I_\delta(\sigma), \Theta_\delta(\sigma))$.
To define $f^1(\sigma)$, we look at each site $u \in \TorusLattice$, we consider its orientation $\sigma(u)$ in $\sigma$ and its orientation $\Theta(u)$ in the bridge system. Its orientation $f^1(\sigma)(u)$ in $f^1(\sigma)$ is then defined as
$$f^1(\sigma)(u) := (\sigma(u) - \Theta(u)) \Mod{(q+1)} - 1.$$
As $\sigma(u) = \Theta(u)$ for each bridged site $u \in \TorusLattice$, this has the effect of converting every bridged particle to orientation $-1$.
At most $\delta N$ sites are unbridged by $\mathcal{B}_\delta(\sigma)$, which implies that $f^1(\sigma) \in \OmegaAggr{\leq \delta N}$.

We now define the function $f^2$ from $\OmegaAggr{\leq \delta N}$ to $\OmegaAggrRho$, with a ``banking'' argument, to restore the number of filled sites to exactly $\rho N$.
For any integer $m \geq 1$ and central location $v \in V(\TorusLattice)$, we can construct a region $S_m(v) = \{s_1(v),s_2(v),\ldots,s_m(v)\} \subseteq V(\TorusLattice)$ site by site, by selecting $s_1(v) = v$, then building a spiral outward from $v$ with $s_2(v),\ldots,s_m(v)$. Note that for any $m \leq \rho N$, $S_m(v)$ will have the lowest boundary length out of any region of $m$ particles.

Fix $\sigma \in \OmegaAggr{\leq \delta N}$ and fix $m := \ceil{\frac{\rho}{1-\delta}N}$. We pick a central location $u \in V(\TorusLattice)$ such that the region $S_m(u)$ has less than $\delta m$ sites that are not of orientation $-1$. This is possible because if we consider the region $S_m(v)$ for each $v \in V(\TorusLattice)$, every site in $V(\TorusLattice)$ will be in exactly $m$ of these regions and as there are less than $\delta N$ sites not of orientation $-1$, the total number of sites not of orientation $-1$ across all of these regions is less than $\delta nm$, implying that there is a region $S_m(v)$ with less than $\delta m$ sites not of orientation $-1$.

Constructing the function $f^2 : \OmegaAggr{\leq \delta N} \to \OmegaAggrRho$ requires us to restore the number of sites not of orientation $-1$ to exactly $\rho N$.
To do this, we apply the map $i \mapsto (i+1) \Mod{(q+1)} - 1$ to the orientations of the sites $s_1(u), s_2(u),\ldots s_m(u)$ in sequence.
Applying this map to the orientation of a site increases the number of sites of orientation $-1$ by $1$ if the site was not originally of orientation $-1$ in $\sigma$ and decreases the number of sites of orientation $-1$ by $1$ if it was.
We stop when the number of sites of orientation $-1$ is exactly $\rho N$, giving us a configuration $f^2(\sigma) \in \OmegaAggrRho$.
As we assumed that $\delta < \rho$ and there are $k \leq \delta m$ of these sites that are originally not of orientation $-1$ in $\sigma$, we need to go through at most $2k + (\rho N - k) \leq \delta m + \rho N$ sites in this sequence before achiving our goal. This is always possible as $\delta m + \rho N \leq m$, which is true if and only if $m \geq \frac{\rho}{1-\delta}N$.

As the bridge system with just a polynomial amount of additional information can be used to reconstruct $\sigma$ from $f^2 \circ f^1_\ell$, our upper bound on the number of bridge systems can be used to upper bound $|(f^2 \circ f^1_\ell)^{-1}(\tau)|$ for any $\tau$ in the image of $f^2 \circ f^1_\ell$. This allows us to prove the following Lemma:

\begin{lemma}
\label{lemma:alphaaggregation}
Fix $\rho < \frac{1}{3}$, any $\alpha > 1$, $\delta \in (0, \min\{\rho, 1-\frac{1}{\alpha^2}\})$ and $\lambda > \lambda_0(q,\rho,\alpha,\delta)$ sufficiently large, where:
$$\lambda_0(q,\rho,\alpha,\delta) := \left((3(q+1))^{\alpha\frac{1+\delta}{2\delta}} 36^{\frac{1}{4\sqrt{3\rho}}}\right)^{\frac{1}{\alpha - \frac{1}{\sqrt{1-\delta}}}}.$$
Denote by $\OmegaAggrRho_{\geq \alpha \cdot bd_{\min}(\rho N)}$ the set of configurations $\sigma$ where $|I_\delta(\sigma)| \geq \alpha \cdot bd_{\min}(\rho N)$, where $bd_{\min}(k)$ is the minimum possible boundary length of a region of $k \in \mathbb{N}$ particles.
Then there exists a constant $\zeta = \zeta(q,\rho,\alpha,\delta,\lambda) < 1$ such that $\piFinite(\OmegaAggrRho_{\geq \alpha \cdot bd_{\min}(\rho N)}) < \zeta^{\sqrt{N}}$ for all sufficiently large values of $N$.
\end{lemma}
\begin{proof}
For this proof, we define the weight $w(\sigma)$ of a configuration $\sigma \in \bigcup_{k \geq 0}\OmegaAggr{k}$ as
$$w(\sigma) := \lambda^{-a(\sigma) - h(\sigma)}$$
where $a(\sigma)$ represents the number of edges between filled and unfilled sites in $\TorusLattice$. Note that for $\sigma \in \OmegaAggrRho$, we have $\piFinite(\sigma) = w(\sigma)/\ZFinite$.
We apply a Peierls argument using $f^2 \circ f^1: \OmegaAggrRho_\ell \to \OmegaAggrRho$, using the functions $f^1$ and $f^2$ defined before.
%
As $\sigma \in \OmegaAggrRho_\ell$ can be reconstructed from $f^1(\sigma)$ using its bridge system and there at most $7 \cdot 6^{2\sqrt{N}-1} \cdot (3(q+1))^{\frac{1+\delta}{2\delta}\ell}$ bridge systems of length $\ell$, for any $\tau \in f^1(\OmegaAggrRho_\ell)$, we have $|(f^1)^{-1}(\tau)| \leq 7 \cdot 6^{2\sqrt{N}-1} \cdot (3(q+1))^{\frac{1+\delta}{2\delta}\ell}$.
%
As going from $\sigma$ to $f^1(\sigma)$ removes exactly the contours in $I_\delta(\sigma)$, we have $w(\sigma) = \lambda^{-\ell}w(f^1(\sigma))$. 
%
$\sigma \in \OmegaAggr{\leq \delta N}$ can be reconstructed from $f^2(\sigma)$ using the sequence $S_m(u), u \in V(\TorusLattice)$ and the stopping point for the application of the map, which can be easily reversed to reconstruct $\sigma$. Thus, for any $\tau \in f^2(\OmegaAggr{\leq \delta N})$, we have $|(f^2)^{-1}(\tau)| \leq |V(\TorusLattice)| \cdot m = N\ceil{\frac{\rho}{1-\delta}}$.
%
In the function $f^2$, the map $i \mapsto (i+1) \Mod{(q+1)} - 1$, when applied to two adjacent sites $u$ and $v$, will not change the contribution by the edge $\{u,v\}$ to the weight of the configuration. As this map is applied to a region of at most $m := \ceil{\frac{\rho}{1-\delta}N}$ particles of minimal boundary length, we must have $w(\sigma) \leq \lambda^{bd_{\min}(m)}w(f^2(\sigma))$.
%
%
Thus we have:
\begin{align*}
w(\OmegaAggrRho_\ell)
= \sum_{\sigma \in \OmegaAggrRho_\ell} \pi_\delta(\sigma) 
&\leq \sum_{\sigma \in \OmegaAggrRho_\ell} \lambda^{bd_{\min}(\ceil{\frac{\rho}{1-\delta}N}) - \ell} \pi(f^2 \circ f^1(\sigma)) \\
&= \sum_{\tau \in \OmegaAggrRho} \sum_{\sigma \in (f^2\circ f^1)^{-1}(\tau)} \lambda^{bd_{\min}(\ceil{\frac{\rho}{1-\delta}N}) - \ell} \pi(\tau) \\
&= \sum_{\tau \in \OmegaAggrRho} |(f^2\circ f^1)^{-1}(\tau)| \lambda^{bd_{\min}(\ceil{\frac{\rho}{1-\delta}N}) - \ell} \pi(\tau) \\
&\leq 7 \cdot 6^{2\sqrt{N}-1} \cdot (3(q+1))^{\frac{1+\delta}{2\delta}\ell} \cdot N\ceil{\frac{\rho}{1-\delta}} \cdot \lambda^{bd_{\min}(\ceil{\frac{\rho}{1-\delta}N}) - \ell} \ZFinite
\end{align*}

Which implies that for any $\alpha > 1$ and a sufficiently small $\delta > 0$, denoting $C_{q,\delta} := (3(q+1))^{\frac{1+\delta}{2\delta}}$ and assuming $\lambda > C_{q,\delta}$, we have
\begin{align*}
&\piFinite(\OmegaAggrRho_{\ell \geq \alpha \cdot bd_{\min}(\rho N)}) 
 = \sum_{\ell \geq \alpha \cdot bd_{\min}(\rho N)} \piFinite(\OmegaAggrRho_\ell) \\
&\leq 7 \cdot 6^{2\sqrt{N}-1}
\cdot N\ceil{\frac{\rho}{1-\delta}}
\lambda^{bd_{\min}(\ceil{\frac{\rho}{1-\delta}N})}
\sum_{\ell \geq \alpha \cdot bd_{\min}(\rho N)}
\left(\frac{C_{q,\delta}}{\lambda}\right)^{\ell} \\
&= \mathrm{poly}(N)
36^{\sqrt{N}}
\lambda^{bd_{\min}(\ceil{\frac{\rho}{1-\delta}N})}
\left(\frac{C_{q,\delta}}{\lambda}\right)^{\alpha \cdot bd_{\min}(\rho N)} \\
%
%
&= \mathrm{poly}(N)
\exp \left\{\sqrt{N}\left(
\log 36 +
\alpha \cdot 4\sqrt{3\rho} \log C_{q,\delta}
- (4\sqrt{3\rho}\log \lambda) (\alpha - \frac{1}{\sqrt{1-\delta}})
\right) + O_N(1)\right\}
\end{align*}

Where the last equality uses the fact that $bd_{\min}(cN) = 4\sqrt{3c}\sqrt{N} + O_N(1)$ for any constant $c \in (0,1/3)$ (Lemma~\ref{lem:bdmin_growth}).
Therefore, as $\delta < 1 - \frac{1}{\alpha^2}$ implies that $\alpha - \frac{1}{\sqrt{1-\delta}} > 0$,  as $\lambda$ had been set sufficiently large as stated in the Lemma,
there exists a constant $\zeta = \zeta(q,\rho,\alpha,\zeta,\lambda) < 1$ such that $\piFinite(\OmegaAggrRho_{\ell \geq \alpha \cdot bd_{\min}(\rho N)}) < \zeta^{\sqrt{N}}$ for all sufficiently large values of $N$.
\end{proof}


The use of Lemma~\ref{lemma:alphaaggregation} along with some results on the minimum possible boundary lengths of regions of $k$ particles allows us to show that there will exist a low perimeter region dominated by a single color, allowing us to prove Theorem~\ref{theorem:aggregationwithalignment}.

\begin{proof}[Proof of Theorem \ref{theorem:aggregationwithalignment}]
Without loss of generality we assume that $\delta < 1$.
We start by defining a few auxiliary variables. We set $\epsilon := 1 - \sqrt{1-\delta} \in (0,1)$ and $x := \frac{1}{\sqrt{1-\epsilon}} - 1 \in (0,1)$. We then define $\alpha' := \min\{\alpha, 1 + \frac{1}{3}x\} > 1$ and pick a sufficiently small $\delta' > 0$ such that
$$\delta' < \min\left\{1 - \sqrt{1-\delta}, 1-\frac{1}{\alpha'^2}, \rho\left(1-(1-\frac{1}{3}x)^2\right)\right\}.$$
Note that setting $\delta'$ this way also ensures the following properties, that we will need later, are true: $\delta' < \rho < 1/3$ and $\delta' \leq \delta$. The latter property is true because $1 - \delta' \geq \sqrt{1-\delta} \geq 1-\delta$.

If we pick $\sigma \in \OmegaAggrRho$ at random according to the distribution $\piFinite$ and consider its $\delta'$-bridge system $\mathcal{B}_{\delta'}(\sigma) = (B_{\delta'}(\sigma), I_{\delta'}(\sigma), \Theta_{\delta'}(\sigma))$,
by Lemma \ref{lemma:alphaaggregation}, there exists a $\lambda_0 = \lambda_0(q,\rho,\alpha',\delta')$ and a $\zeta = \zeta(q, \rho, \alpha', \delta', \lambda) < 1$ such that for all $\lambda > \lambda_0$, with probability at least $1 - \zeta^{\sqrt{N}}$, we have
$|I_{\delta'}(\sigma)| < \alpha \cdot bd_{\min}(\rho N)$.

The edge set $I_{\delta'}(\sigma)$ is a collection of complex contours, dividing $\TorusLattice$ into regions and the assignment $\Theta$ denotes the ``main'' color of each region.
Thus, we define the region $\mathcal{R} := \{v \in V(\TorusLattice) : \Theta_{\delta'}(\sigma)(v) \neq -1\}$. This region is of low perimeter because
$$|bd(\mathcal{R})| \leq |I_\delta(\sigma)| \leq \alpha' \cdot bd_{\min}(\rho N) \leq \alpha \cdot bd_{\min}(\rho N).$$ 
Also, as $V(\TorusLattice)\setminus R$ has at most $\delta'(N-|\mathcal{R}|) \leq \delta(N-|\mathcal{R}|)$ unbridged particles, the number of unfilled sites not in $\mathcal{R}$ is at most $\delta(N-|\mathcal{R}|)$. This gives us the second and third properties required by the Theorem. The remainder of the proof will show that the first property is true.

For $i \in \{-1,0,1,\ldots,q-1\}$, denote by $n_i$ the number of sites $v \in V(\TorusLattice)$ where $\Theta(v) = i$.
As the fraction of unbridged particles in each minimal bridged region of $I_{\delta'}(\sigma)$ is at most $\delta$,
denoting $n_{\geq 0} := \sum_{i=0}^{q-1} n_i$ and taking note that there are exactly $\rho N$ sites not of orientation $-1$ in $\sigma$, we have
$$(1 - \delta')n_{\geq 0} \leq \rho N \leq n_{\geq 0} + (N-n_{\geq 0})\delta',$$
which we can rewrite as
$$N\left(\frac{\rho-\delta}{1-\delta'}\right) \leq n_{\geq 0} \leq N\left(\frac{\rho}{1-\delta'}\right).$$

Denoting by $bd_{\min}(k)$ the minimum possible boundary length of a region of $k$ particles in $\TorusLattice$, we must have
\begin{align}
\label{eqn:bdmineqn}
\frac{1}{2}\sum_{i=-1}^{q-1} bd_{\min}(n_i) \leq |I_\delta(\sigma)| \leq \alpha' \cdot bd_{\min}(\rho N).
\end{align}
As $\delta',\rho < 1/3$, we have $n_{\geq 0} \leq N\frac{\rho}{1-\delta'} < \frac{N}{2}$.
Thus by Lemma~\ref{lem:bdmin_growth}, as $n_{\geq 0} \geq N\frac{\rho - \delta'}{1 - \delta'}$, $\frac{\rho - \delta'}{1 - \delta'} < \rho < \frac{1}{3}$ and $bd_{\min}(n_{-1}) = bd_{\min}(N-n_{-1}) = bd_{\min}(n_{\geq 0})$, we have:
\begin{align*}
bd_{\min}(n_{-1}) \geq \begin{cases}
4\sqrt{3\frac{n_{\geq 0}}{N}N} + O_N(1) &\text{if $n_{\geq 0} < \frac{N}{3}$} \\
4\sqrt{N} + O_N(1) &\text{otherwise}
\end{cases}
\geq 4\sqrt{3N} \sqrt{\frac{\rho - \delta'}{1 - \delta'}} + O_N(1).
\end{align*}
Applying this to Equation \ref{eqn:bdmineqn} we get
$$\sum_{i=0}^{q-1} \sqrt{n_i} \leq 2\alpha' \sqrt{\rho N} - \sqrt{N\frac{\rho - \delta'}{1 - \delta'}} + O_N(1),$$
which we can rewrite as
$$\sum_{i=0}^{q-1} \sqrt{\frac{n_i}{\rho N}}  \leq \left(2\alpha' - \sqrt{\frac{1-\delta'/\rho}{1-\delta'}}\right) + O\left(\frac{1}{\sqrt{N}}\right).$$

To show that one of the orientation takes the majority, we make use of the following claim: Suppose $0 \leq y_i \leq 1$ for $i \in \{1,2,\ldots,k\}$, $\sum_{i=1}^k y_i = 1$ and $\sum_{i=1}^k \sqrt{y_i} \leq \frac{1}{\sqrt{1-\epsilon}}$. Then there exists an $i \in \{1,2,\ldots,k\}$ where $y_i \geq 1-\epsilon$.

To prove this claim, without loss of generality we assume that $\max_i y_i = y_1$, so $\sqrt{y_i} \leq \sqrt{y_1}$ for all $i$. Thus, $\sum_{i=1}^k \sqrt{y_i} = \sum_{i=1}^k y_i/\sqrt{y_i} \geq \sum_{i=1}^k y_i/\sqrt{y_1} = 1/\sqrt{y_1}$. This implies that $1/\sqrt{x_1} \leq 1/\sqrt{1-\epsilon}$, or $x_1 \geq 1 - \epsilon$.
Note that a more precise bound for the minimum value of the largest $x_i$ can be obtained with more work, which we will omit for the sake of keeping the proof simple.

To make use of this claim, we show that $2\alpha' - \sqrt{\frac{1-\delta'/\rho}{1-\delta'}} < \frac{1}{\sqrt{1-\epsilon}}$. Our specific choices of $x, \alpha'$ and $\delta'$ at the start of the proof were to achieve exactly this. $\delta' < \rho(1-(1-\frac{1}{3}x)^2)$ implies that $1 - \frac{\delta'}{\rho} > (1-\frac{1}{3}x)^2$, which gives us
$$\frac{\sqrt{1-\delta'/\rho}}{1-\delta'} > \sqrt{1-\delta'/\rho} > 1 - \frac{1}{3}x.$$
As $\alpha' \leq 1 + \frac{1}{3}x$ and $x := \frac{1}{1-\epsilon}-1$, we can conclude that
$$2\alpha' - \sqrt{\frac{1-\delta'/\rho}{1-\delta'}} < 2(1+\frac{1}{3}x) - (1-\frac{1}{3}x) = 1+x = \frac{1}{\sqrt{1-\epsilon}}.$$
Thus by the above claim, for all sufficiently large $N$, there must exist a $\theta \in \{0,1,\ldots,q-1\}$ where $n_\theta \geq (1-\epsilon) \rho N$.
As $|\mathcal{R}| = n_{\geq 0} \leq N\frac{\rho}{1-\delta'}$ and as we had set $\epsilon := 1 - \sqrt{1-\delta}$ and $\delta' < 1-\sqrt{1-\delta}$, we have
\begin{align*}
n_i=\theta \geq (1-\epsilon)\rho N 
&\geq (1-\epsilon)(1-\delta')|\mathcal{R}| \\
&\geq (1-\delta)|\mathcal{R}|
\end{align*}

\end{proof}

\section{Non-Alignment and Dispersion in General SOPS}

\newcommand{\OmegaAligned}{\widetilde{\Omega}^{\alpha,\delta}}

\begin{theorem}[Non-Alignment in General SOPS]
\label{thm:gen_non_alignment}
For any $q \geq 2$ and $\epsilon \in (0,\frac{1}{q})$, when $\lambda > 0$ satisfies:
\begin{align*}
\lambda^6 < \left( 1 - \epsilon \frac{q}{q-1} \right)^{\frac{q-1}{q}-\epsilon} \left(1+\epsilon\,q \right)^{\frac{1}{q}+\epsilon} = 1+ \frac{\epsilon^2 q^2}{q-1} + O(\epsilon^3) \,,
\end{align*}
the probability that a configuration sampled from the stationary distribution of the Markov chain algorithm $\piFinite$ is not $\epsilon$-\emph{non-aligned} is exponentially small,  for sufficiently large $n$.
\end{theorem}

\newcommand{\Snnaligned}{S^\epsilon}
\newcommand{\Snnadelta}{S^\epsilon_{\theta,\delta}}
\begin{proof}

Denote by $\Snnaligned$ the set of configurations of $\OmegaAggrRho$ that are not $\epsilon$-non-aligned.
For any configuration in $\Snnaligned$, there is a direction $\theta$ where the fraction $\delta$ of particles along orientation $\theta$ is such that $|\delta - 1/q| \geq \epsilon$. Let $\Snnadelta$ denote the set of configurations where the fraction of particles along some direction is exactly $\delta$.

By counting the number of configurations in $\Snnadelta$, similar to the proof of Lemma~\ref{lem:non_alignment_P} but with a lowest possible weight of a configuration being $\lambda^{-6n}$, we have
\begin{align*}
\piFinite{P}(\Snnaligned) = \frac{\wFinite(\Snnaligned)}{\wFinite(\OmegaAggrRho)} < \frac{{N \choose n}\cdot q {n \choose \delta n} (q-1)^{(1-\delta)\,n}}{{N \choose n} \cdot q^n\,\lambda^{-6n}}\,.
\end{align*}
And similarly with an equality obtained from the Stirling's approximation, we have
\begin{align*}
 \piFinite(\Snnaligned)
&\leq \frac{q}{\sqrt{2\pi n\delta(1-\delta)}} \left( \frac{\lambda^6}{\left( \frac{q(1-\delta)}{q-1}\right)^{(1-\delta)}(q\delta)^{\delta} } \right)^n\,.
\end{align*}
Thus, following the same argument as the proof of Lemma~\ref{lem:non_alignment_P}, as long as $\lambda^6 < 1 + \frac{\epsilon^2 q^2}{q-1} + O(\epsilon^3)$, there exists a $\zeta < 1$ such that $\piFinite(\Snnaligned) \leq \zeta^n$ for sufficiently large $n$.
\end{proof}

\begin{theorem}[Dispersion in General SOPS]
\label{thm:gen_dispersion}
Fix $\rho < \frac{1}{3}$ and assume that there will always be exactly $\rho N$ filled sites on the lattice. For sufficiently small values of $\delta > 0$, there exists a $\lambda_s = \lambda_s(\rho,\delta)$ such that for all $\lambda < \lambda_s$, for any constant number of orientations $q$ and for any $\alpha > 1$, with probability $\zeta_s^N$ for some $\zeta_s = \zeta_s(\rho,\delta) < 1$,
there exists no region $\mathcal{R} \subseteq V(\TorusLattice)$ that simultaneously satisfies the three following properties:
\begin{enumerate}
\item The number of unfilled sites in $\mathcal{R}$ is at most $\delta|\mathcal{R}|$.
\item The number of filled sites not in $\mathcal{R}$ is at most $\delta (N-|\mathcal{R}|)$
\item The boundary length of $\mathcal{R}$ is at most $\alpha \cdot bd_{\min}(\rho N)$.
\end{enumerate}
\end{theorem}

As an additional remark, this same dispersion proof applies as long as the boundary length requirement for $\mathcal{R}$ (the third property) is $o(n/\log n)$.

\begin{proof}
Denote by $\OmegaAligned$ the set of configurations with such a region $\mathcal{R}$ that satisfies the three properties for a given $\delta$ and $\alpha$. We upper bound the number of configurations in $\OmegaAligned$.

Each configuration in $\OmegaAligned$ has a region $\mathcal{R} \subseteq V(\TorusLattice)$ of some boundary length $\ell \leq \alpha \cdot bd_{\min}(\rho N)$. As each edge of this boundary can be in one of $3N$ possible locations, a simple upper bound for the number of such regions is $(3N)^{\ell} \leq (3N)^{\alpha \cdot bd_{\min}(\rho N)}$. Within the region $\mathcal{R}$, there are at most $\delta|\mathcal{R}|$ unfilled sites, which gives us ${|\mathcal{R}| \choose \delta|\mathcal{R}|}$ ways to define the filled/unfilled status of the sites in $\mathcal{R}$. For a similar reason, there are ${N-|\mathcal{R}| \choose \delta(N-|\mathcal{R}|)}$ to define the filled/unfilled status of the sites outside of $\mathcal{R}$. Thus, the number of ways to define the filled/unfilled status of configurations in $\OmegaAligned$ is at most:
\begin{align*}
(3N)^{\alpha \cdot bd_{\min}(\rho N)} \cdot
{|\mathcal{R}| \choose \delta|\mathcal{R}|} {N-|\mathcal{R}| \choose \delta(N-|\mathcal{R}|)}
&\leq (3N)^{\alpha \cdot bd_{\min}(\rho N)} \cdot
\left(\frac{e}{\delta}\right)^{\delta|\mathcal{R}|}
\left(\frac{e}{\delta}\right)^{\delta(N-|\mathcal{R}|)} \\
&= (3N)^{\alpha 4\sqrt{3\rho}\sqrt{N} + O_N(1)} \cdot
\left(\frac{e}{\delta}\right)^{\delta N}
\end{align*}
Where the relation $bd_{\min}(\rho N) = 4\sqrt{3\rho}\sqrt{N} + O_N(1)$ comes from Lemma~\ref{lem:bdmin_growth}.
Similarly, the filled/unfilled status for the set $\OmegaAggrRho$ of configurations with exactly $\rho N$ filled particles can be defined in ${N \choose \rho N} \geq (1/\rho)^{\rho N}$ ways.

For any configuration $\sigma \in \OmegaAggrRho$, we note that $\lambda^{-6\rho N} \leq \wFinite(\sigma) \leq 1$.
For any constant $c > 1$, as $(3N)^{\alpha 4\sqrt{3\rho}\sqrt{N} + O_N(1)} = \left((3N)^{\alpha 4\sqrt{3\rho} + O(\frac{1}{\sqrt{N}})}\right)^{\sqrt{N}} = o\left((c^{\sqrt{N}})^{\sqrt{N}}\right) = o(c^N)$, we have:
\begin{align*}
\piFinite(\OmegaAligned)
&= \frac{\wFinite(\OmegaAligned)}{\wFinite(\OmegaAggrRho)}
= \frac{\sum_{\sigma \in \OmegaAligned} \wFinite(\sigma) }{\sum_{\sigma \in \OmegaAggrRho} \wFinite(\sigma)}
\leq \frac{|\OmegaAligned|}{|\OmegaAggrRho|\lambda^{-6\rho N}} \\
&\leq \frac{
q^{\rho N} N^{\alpha 4\sqrt{3\rho}\sqrt{N} + O_N(1)} \cdot \left(\frac{e}{\delta}\right)^{\delta N} \lambda^{6\rho N}}
{q^{\rho N} \left(\frac{1}{\rho}\right)^{\rho N}}
< \left( \frac{c \cdot (e/\delta)^\delta \lambda^{3\rho}}{(1/\rho)^\rho} \right)^N
\end{align*}
Hence, as long as $\frac{(e/\delta)^\delta \lambda^{3\rho}}{(1/\rho)^\rho} < 1$, there exists a $\zeta_s < 1$ such that $\piFinite(\OmegaAligned) < \zeta_s^N$ for all sufficiently large $N$. This is achieved as long as
\begin{align*}
\frac{(e/\delta)^\delta \lambda^{3\rho}}{(1/\rho)^\rho} < 1
\iff \lambda < \left( \frac{(1/\rho)^\rho}{(e/\delta)^\delta} \right)^{1/3\rho}.
\end{align*}
This is possible with $\lambda > 1$ for any $\rho \in (0,1)$ as long as $\delta$ is sufficiently small such that $(e/\delta)^\delta < (1/\rho)^\rho$.
\end{proof}

\bibliography{refs}



\end{document}